\documentclass[acmsmall]{acmart}

\setcopyright{rightsretained}
\acmDOI{10.1145/3649822}
\acmYear{2024}
\copyrightyear{2024}
\acmSubmissionID{oopslaa24main-p30-p}
\acmJournal{PACMPL}
\acmVolume{8}
\acmNumber{OOPSLA1}
\acmArticle{105}
\acmMonth{4}
\received{20-OCT-2023}
\received[accepted]{2024-02-24}

\usepackage[shortlabels]{enumitem}
\usepackage{subcaption}
\usepackage{mathpartir}
\usepackage{turnstile}
\usepackage{mathtools}
\mathtoolsset{showmanualtags}
\usepackage{stmaryrd}
\let\detensor\lightning
\let\tensor\circleddot
\usepackage{multirow}
\usepackage{multicol}
\usepackage{algorithm}
\usepackage[noend]{algpseudocode}

\algnewcommand\algorithmicforeach{\textbf{for each}}
\algdef{S}[FOR]{ForEach}[1]{\algorithmicforeach\ #1\ \algorithmicdo}
\usepackage{pseudo}
\pseudoset{indent-length=0.7em}
\usepackage{tikz}
\usetikzlibrary{shapes,positioning,quotes}
\usepackage{tikz-cd}
\usepackage{qtree}
\usepackage{xfrac}
\usepackage{hhline}
\usepackage{rotating}
\usepackage{wrapfig}
\usepackage{cleveref}
\crefname{theorem}{Thm.}{Thms.}
\crefname{lemma}{Lem.}{Lemmas}
\crefname{corollary}{Cor.}{Cors.}
\crefname{figure}{Fig.}{Figs.}
\crefname{definition}{Defn.}{Defns.}
\crefname{table}{Tab.}{Tabs.}
\crefformat{section}{\S#2#1#3}
\crefmultiformat{section}{\S#2#1#3}{ and~\S#2#1#3}{, \S#2#1#3}{ and~\S#2#1#3}
\crefname{example}{Ex.}{Exs.}
\crefname{item}{item}{items}
\crefname{footnote}{footnote}{footnotes}
\crefname{observation}{Obs.}{Obs.}
\crefname{remark}{Remark}{Remarks}
\crefname{proposition}{Prop.}{Props.}
\crefname{equation}{Eqn.}{Eqns.}
\crefname{counterexample}{Counterexample}{Counterexamples}
\crefname{property}{Property}{Properties}
\crefname{algorithm}{Algorithm}{Algorithms}
\usepackage{shortcuts}

\usepackage{framed}
\usepackage{nicematrix}
\usepackage{tablefootnote}
\usepackage{wasysym}
\usepackage{mdframed}
\usepackage{setspace}

\newenvironment{wrapped}[1]
 {\def\wrappedcurrent{#1}%
  \setlength{\columnwidth}{\parshapelength\numexpr\prevgraf+2\relax}%
  \csname #1\endcsname}
 {\csname end\wrappedcurrent\endcsname}

\makeatletter 
\AtEndDocument{%
\def\cb@checkPdfxy#1#2#3#4#5{%
\cb@@findpdfpoint{#1}{#2}%
\ifdim#3sp=\cb@pdfx\relax      
\ifdim#4sp=\cb@pdfy\relax      
\ifdim#5=\cb@pdfz\relax
\else
\cb@error
\fi
\else
\cb@error
\fi
\else
\cb@error
\fi
}}
\makeatother 

\usepackage{mathrsfs}
\usepackage[cal=cm]{mathalfa}
\renewcommand{\mathbb}[1]{\vvmathbb{#1}}

\newtheoremstyle{acmremark}%
  {.5\baselineskip\@plus.2\baselineskip\@minus.2\baselineskip}
  {.5\baselineskip\@plus.2\baselineskip\@minus.2\baselineskip}
  {\itshape}
  {\parindent}
  {\itshape}
  {.}
  {.5em}
  {\thmname{#1}\thmnumber{ #2}\thmnote{ {(#3)}}}
\AtEndPreamble{%
  \theoremstyle{acmplain}%
  \newtheorem*{theorem*}{Theorem}%
  \newtheorem*{lemma*}{Lemma}%
  \newtheorem*{proposition*}{Proposition}%
  \theoremstyle{acmdefinition}%
  \theoremstyle{acmremark}%
  \newtheorem{remark}[theorem]{Remark}%
  \theoremstyle{acmplain}%
}

\AtBeginDocument{%
  \setlength\abovedisplayskip{0.5\abovedisplayskip}%
  \setlength\belowdisplayskip{0.5\belowdisplayskip}%
  \setlength\abovedisplayshortskip{0.5\abovedisplayshortskip}%
  \setlength\belowdisplayshortskip{0.5\belowdisplayshortskip}%
  \setlength\floatsep{0.5\floatsep}%
  \setlength\abovecaptionskip{0.5\abovecaptionskip}%
}
\setlist{leftmargin=\parindent}
\makeatletter
\renewcommand{\paragraph}{%
  \@startsection{paragraph}{4}%
  {\z@}{-.2\baselineskip \@plus -2\p@ \@minus -.2\p@}{-3.5\p@}%
  {\bfseries\@adddotafter}%
}
\makeatother

\newcommand{\m}[1]{\mathsf{#1}}
\newcommand{\mi}[1]{\mathit{#1}}

\newcommand{\gcho}[1]{\mathbin{\prescript{}{#1\!}{\Diamond}}}
\newcommand{\pcho}[1]{\mathbin{\prescript{}{#1\!}{\oplus}}}
\newcommand{\aord}{\sqsubseteq}
\newcommand{\azero}{\underline{0}}
\newcommand{\aone}{\underline{1}}

\DeclareMathOperator{\lfp}{\mathrm{lfp}}

\DeclareMathOperator{\dom}{\mathrm{dom}}

\newcommand{\infholder}{{\circlearrowleft}}
\newcommand{\leftleadsto}{\mathrel{\reflectbox{$\leadsto$}}}

\newcommand{\dashcupi}{\mathchoice
	{\sbox0=\hbox{\hbox to -0.01\wd0{\hbox to 0.01\wd0{$\displaystyle-$}$\displaystyle-$}$\displaystyle\cup$}}
	{\sbox0=\hbox{\hbox to -0.01\wd0{\hbox to 0.00\wd0{$\textstyle-$}$\textstyle-$}$\textstyle\cup$}}
	{\sbox0=\hbox{\hbox to 0.04\wd0{\hbox to 0.00\wd0{$\scriptstyle-$}$\scriptstyle-$}$\scriptstyle\cup$}}
	{\sbox0=\hbox{\hbox to 0.03\wd0{\hbox to 0.00\wd0{$\scriptscriptstyle-$}$\scriptscriptstyle-$}$\scriptscriptstyle\cup$}}
}
\newcommand{\dashcup}{\mathbin{\dashcupi}}
\newcommand{\Dashcupi}{\mathchoice
	{\sbox0=\hbox{\hbox to -0.1\wd0{\hbox to 0.3\wd0{\hbox to 0.9\wd0{$\displaystyle-$}$\displaystyle-$}$\displaystyle-$}$\displaystyle\bigcup$}}
	{\sbox0=\hbox{\hbox to -0.08\wd0{\hbox to 0.2\wd0{$\textstyle-$}$\textstyle-$}$\textstyle\bigcup$}}
	{\sbox0=\hbox{\hbox to -0.04\wd0{\hbox to 0.12\wd0{$\scriptstyle-$}$\scriptstyle-$}$\scriptstyle\bigcup$}}
	{\sbox0=\hbox{\hbox to -0.03\wd0{\hbox to 0.09\wd0{$\scriptscriptstyle-$}$\scriptscriptstyle-$}$\scriptscriptstyle\bigcup$}}
}

\newcommand{\dplus}{\mathbin{+\mkern-5mu+}}

\newcommand{\framework}{NPA-PMA}

\newif\iflong
\longtrue

\begin{document}

\title{Newtonian Program Analysis of Probabilistic Programs}
\iflong
\subtitle{Technical Report}
\fi         

\author{Di Wang}
\affiliation{
  \department{Key Lab of High Confidence Software Technologies, Ministry of Education Department of Computer Science and Technology, School of Computer Science}
  \institution{Peking University}
  \country{China}
}

\author{Thomas Reps}
\affiliation{
  \department{Computer Sciences Department}
  \institution{University of Wisconsin}
  \country{USA}
}

\renewcommand{\shortauthors}{Wang and Reps}

\begin{abstract}
Due to their \emph{quantitative} nature,
probabilistic programs pose non-trivial
challenges for designing compositional and efficient program analyses.
Many analyses for probabilistic programs rely on \emph{iterative} approximation.
This article presents an interprocedural dataflow-analysis framework, called \emph{\framework{}},
for designing and implementing (partially) \emph{non-iterative} program analyses of probabilistic
programs with unstructured control-flow, nondeterminism, and general recursion.
%
\framework{} is based on 
Newtonian Program Analysis
(NPA), a generalization of Newton's method to solve equation systems over semirings.
The key challenge for developing \framework{} is to handle multiple kinds of \emph{confluences}
in both the algebraic structures that specify analyses
and the equation systems that encode control flow:
semirings support a single confluence operation,
whereas \framework{} involves three confluence operations (conditional, probabilistic, and nondeterministic).

Our work
introduces \emph{$\omega$-continuous pre-Markov algebras} ($\omega$PMAs) to factor out common parts of different analyses;
adopts \emph{regular infinite-tree expressions} to
encode probabilistic programs with unstructured control-flow;
and presents a \emph{linearization} method that makes Newton's method applicable to the setting of regular-infinite-tree equations over $\omega$PMAs.
%
\framework{} allows analyses to supply a non-iterative strategy to solve linearized equations.
%
Our experimental evaluation demonstrates that
(i) \framework{} holds considerable promise for outperforming Kleene iteration,
and (ii) provides great generality for designing program analyses.
\end{abstract}



\maketitle

\section{Introduction}
\label{Se:Introduction}

%
%
%

Probabilistic programs, which augment deterministic programs with
\emph{data randomness} (e.g., sampling) and \emph{control-flow randomness} (e.g., probabilistic branching),
provide a rich framework for implementing and analyzing randomized algorithms and probabilistic systems~\cite{book:MM05}.
To establish properties of probabilistic programs, people have developed
many \emph{program analyses}, as well as frameworks for encoding existing analyses
and developing new analyses.
Two prominent instances are Probabilistic Abstract Interpretation (PAI)~\cite{ESOP:CM12}
and the Pre-Markov Algebra Framework (PMAF)~\cite{PLDI:WHR18}.
Both frameworks derive from \emph{abstract interpretation}~\cite{POPL:CC77}:
one proves the soundness of an analysis by establishing a \emph{soundness relation} between
concrete and abstract semantics,
and implements the analysis via \emph{iterative approximation} of fixed-points.

However, the \emph{quantitative} nature of probabilistic programs usually causes iteration-based program analysis to be inefficient.
%
For example, consider the probabilistic loop $P$ ~ \defeq ~ ``\kw{while} \kw{prob}($\frac{3}{4}$) \kw{do} $x {\coloneqq} x {+} 1$ \kw{od}'',
which increments program variable $x$ by one with probability
$\frac{3}{4}$ in each iteration, and otherwise exits the loop.
Suppose that a program analysis aims to reason about the expected difference between the
final and initial value of $x$.
An iteration-based analysis would try to obtain the expected difference, denoted by $r$, by
approximating the fixed-point of
$r = f(r) \defeq \frac{3}{4} \cdot (1+r) + \frac{1}{4} \cdot 0$,
where $f$
intuitively captures the fact that $P \equiv ``\kw{if}~\kw{prob}(\frac{3}{4})~\kw{then}~x {\coloneqq} x {+} 1; P~\kw{else}~\kw{skip}~\kw{fi}$''.
For example, \emph{Kleene iteration} yields an approximation sequence
$\{f^i(0)\}_{i \ge 0}$, which converges to the least fixed-point of $f$ but
will not converge in any finite number of steps.
On the other hand, $r$ is the solution of the linear equation
$r = \frac{3}{4} \cdot (1+r) + \frac{1}{4} \cdot 0$, which directly yields $r=3$.
The following equation demonstrates
the difference between the iterative and non-iterative computation:
\begin{equation}\label{Eq:IterationAndNonIteration}
\textbf{iterative} \enskip (\tfrac{3}{4})^1 + (\tfrac{3}{4})^2 + (\tfrac{3}{4})^3 + (\tfrac{3}{4})^4 + \cdots = 3  \qquad 3 = \tfrac{3}{4} \times \tfrac{1}{1 - {3}/{4}} \enskip \textbf{non-iterative} \tag{$\star$}
\end{equation}

This article presents a framework, which we call \emph{\framework{}} (for Newtonian Program Analysis with Pre-Markov Algebras),
for designing and implementing compositional
and efficient \emph{interprocedural} program analyses of probabilistic programs.
In particular, \framework{} enables
analyses with computations of the form shown on the right side of \eqref{Eq:IterationAndNonIteration},
and achieves three
major desiderata:
\begin{itemize}[nosep]
  \item \emph{Compositionality}. \framework{} analyzes a program by analyzing its constituents and then combining
  the analysis results.
  
  \item \emph{Efficiency}. \framework{} allows an instance program analysis to supply a more efficient strategy
  (compared to Kleene iteration) for abstracting repetitive behavior.
  \framework{} then provides a unified and sound approach for efficient \emph{intra}procedural
  analysis (of loops) and \emph{inter}procedural analysis (of recursion) via the analysis-supplied strategy.

  \item \emph{Generality}.
  \framework{} covers
  a multitude of
  analyses of probabilistic programs.
  For our experimental evaluation (\cref{Se:CaseStudies}),
  we instantiated \framework{} for finding
  (i) probability distributions of program states,
  (ii) upper bounds on moments of random variables,
  and (iii) linear and (iv) non-linear expectation invariants.
  We reformulated two existing abstract domains for (i) and (ii),
  and designed two new abstract domains for (iii) and (iv).
\end{itemize}
Our approach is motivated by \emph{Newtonian Program Analysis} (NPA)~\cite{ICALP:EKL08,JACM:EKL10}.
Many interprocedural dataflow analyses can be reduced to finding the
least fixed-point of an equation system $\vec{X} = \vec{f}(\vec{X})$ over a
semiring~\cite{POPL:BET03,SAS:RSJ03,FSTTCS:RLK07}.
\citeauthor{JACM:EKL10} considered interprocedural analysis of \emph{loop-free} programs and developed an iterative approximation 
method with the following scheme,
where $\textsc{LinearCorrection}(\vec{f},\vec{\nu}^{(i)})$ is a correction term, which is the solution
to a ``linearized'' equation system of $\vec{f}$ at $\vec{\nu}^{(i)}$
(we will discuss linearization in \cref{Se:PriorWorkNewtonianProgramAnalysis}):
\begin{align}
  \label{Eq:NPAIntuitive}
  \vec{\nu}^{(0)} & = \vec{f}(\vec{\bot}), & \vec{\nu}^{(i+1)} & = \vec{f}(\vec{\nu}^{(i)}) \sqcup \textsc{LinearCorrection}(\vec{f},\vec{\nu}^{(i)}).
\end{align}
NPA is \emph{partially iterative} and \emph{partially non-iterative}:
\cref{Eq:NPAIntuitive} defines the iteration that produces the successive values of the Newton sequence, whereas the solution to the linearized equation system of $\vec{f}$ at each value $\vec{\nu}^{(i)}$ can be computed non-iteratively.

\begin{figure}
\centering
\includegraphics[width=0.85\textwidth]{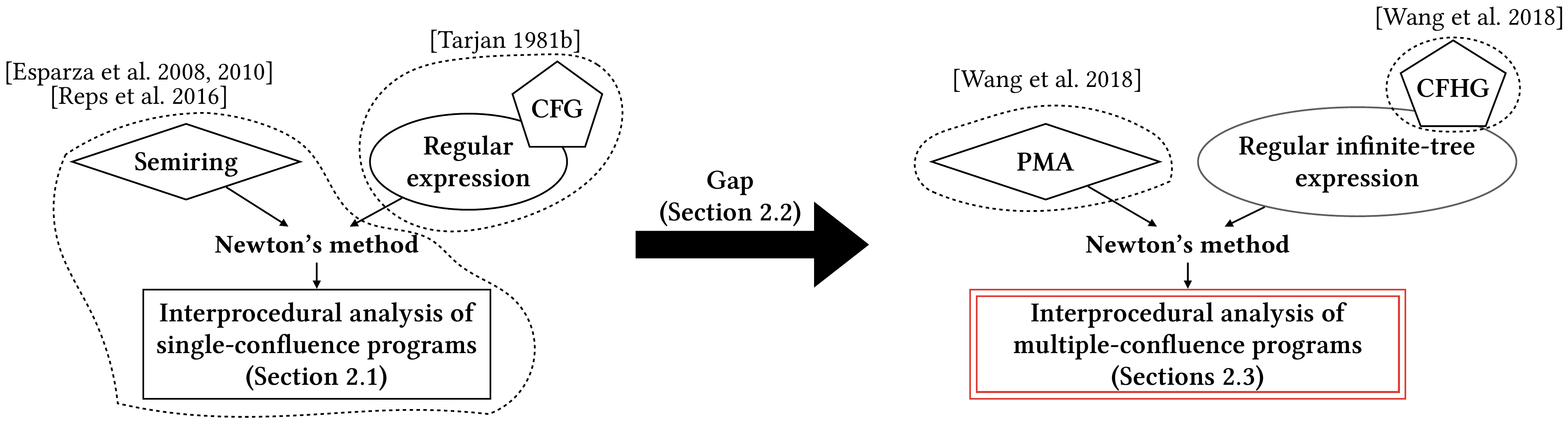}
\caption{A roadmap for the Overview section (\cref{Se:Overview}).
We use dashed curves for concepts from prior work and red double borders for contributions of this work.
For each variant of Newton's method for solving $\vec{X} = \vec{f}(\vec{X})$,
we use $\Diamond$ for the algebraic structure of the elements,
$\ocircle$ for the form of the right-hand sides $\vec{f}(\vec{X})$,
and $\Square$ for the supported kind of programs.
We put a $\pentagon$ (for a program representation) on top of a $\ocircle$ (for a form of right-hand sides)
to denote that there is a sound transformation from $\pentagon$ to $\ocircle$.
}
\label{Fi:Roadmap}
\end{figure}

NPA and related variants of Newton's method have been used for analyzing
probabilistic models, such as recursive Markov chains~\cite{STACS:EY05},
probabilistic pushdown automata~\cite{LICS:EKM04}, and probabilistic one-counter
automata~\cite{CAV:BKK11}.
However, those models are not an ideal representation of probabilistic programs:
they require that either the state space is finite, or the branching is purely probabilistic.
In particular, \emph{nondeterminism} arises naturally when the analyzed program has an infinite state space,
because an analysis usually relies on an \emph{abstraction} of the state space, and consequently can only treat some conditional choices as nondeterministic ones, due to loss of information.
\textbf{
Thus, the key challenge for developing \framework{} is to handle multiple kinds of
confluence operations: conditional, probabilistic, and nondeterministic.}

\cref{Fi:Roadmap} provides a roadmap for the presentation of our \framework{} framework for interprocedural analysis of probabilistic programs with multiple confluence operations.
The starting point is the NPA framework for interprocedural analysis of programs with a single confluence
operation.
Non-probabilistic programs are usually formalized as single-confluence programs:
in the control-flow graph (CFG) of a non-probabilistic program, a node can have multiple
successors, but the operation applied at the confluence nodes is usually the same operation (e.g., join).
In contrast, as observed by~\citet{PLDI:WHR18}, because of multiple kinds of
confluence operations, it is more suitable to use \emph{control-flow hyper-graphs} (CFHGs) to
represent probabilistic programs.
In a CFHG,
conditional-, probabilistic-,
and nondeterministic-branching are represented by different kinds of \emph{hyper-edges},
each of which consists of one source node and two destination nodes.
Therefore, the challenges we face are two-fold:
(i) to find a suitable notion of ``hyper-path expressions'' that encode the control-flow hyper-path through a CFHG, and
(ii) to devise an appropriate method---required by Newton's method---of linearizing such hyper-path expressions.

We implemented a prototype of \framework{}, with which our preliminary evaluation confirmed that
the quantitative nature of probabilistic programs unleashes the potential of Newton's method.
As mentioned earlier, we instantiated \framework{} with two existing and two new abstract domains.
For the two analyses based on existing domains, we evaluated them on synthetic multi-procedure benchmark programs.
In addition, for one of the analyses we implemented two different analysis-supplied strategies, showing the generality of \framework{}.
In general, \framework{} outperforms Kleene iteration.
For example, 
for one domain, \framework{} greatly reduces the number of iterations (3,070.42 $\to$ 9.15), contributing to an overall runtime speedup of 4.08x.
(As others have also reported, each Newton iteration takes much more time than a single Kleene iteration.)
For the two analyses
based on new domains,
we evaluated them on benchmark programs collected from the literature.
\framework{} can derive non-trivial expectation invariants compared to state-of-the-art techniques, which shows that \framework{} is flexible and holds considerable promise for supporting program analyses of probabilistic programs.

\Omit{
We now discuss the challenges in the two steps and our solutions to them.

\paragraph{Adding probability}
\citet{PLDI:WHR18} observed that nondeterministic probabilistic programs exhibit more
than one kinds of \emph{confluences} and proposed \emph{pre-Markov algebras} (PMAs) that
contain three confluence operations: conditional-choice, probabilistic-choice, and nondeterministic-choice.
\citeauthor{PLDI:WHR18} then developed an analysis framework of probabilistic programs based on PMAs, but they used an
iterative approximation method based on Kleene iteration, hindering more efficient strategies for analyzing repetitive behavior.
To apply the recipe of NPA, we need to generalize Newton's method to solve equation systems over PMAs,
and the challenges are to develop a proper form of equation right-hand sides
and devise a linearization method for PMAs.

\cref{Se:ThisWorkAConfluenceCentricAnalysisFramework} presents an overview of our solution.
First, we devise a new family of algebraic structures, which we call \emph{$\omega$PMAs}, as
a refinement of both semirings and PMAs.
Intuitively, the semiring-like structure allows us to carry out Newton's method
and the PMA-like structure allows us to encode different kinds of program confluences.
Then, we develop \emph{simple tree expressions} to represent the right-hand sides of an interprocedural equation system.
The internal nodes of a simple tree expression corresponds to program commands, e.g., sequencing, procedure-calls, and choices,
and every root-to-leaf path in the tree represents a program-execution path.
We also present an interpretation of simple tree expressions in an $\omega$PMA.
Finally, we devise a mechanism to compute the \emph{differential} of simple tree expressions
and use it to perform NPA linearization over $\omega$PMAs.
\framework{} then uses the analysis-supplied strategy to solve linearized equation systems.

\paragraph{Adding unstructured control flow}
Loops, or more generally, unstructured control flow, can be easily
represented by control-flow graphs (CFGs) for non-probabilistic programs.
\citet{JACM:Tarjan81:Alg} proposed a fast algorithm to compute a \emph{regular expression} encoding
all program paths through a CFG, which leads to the development of \emph{algebraic program analysis}~\cite{JACM:Tarjan81,arxiv:FK13,CAV:KRC21},
a method for solving intraprocedural dataflow-analysis problems by reinterpreting regular path expressions
in a \emph{regular algebra}, which contains an operation corresponding to Kleene-star $^\ast$.
\citet{POPL:RTP16} observed that algebraic program analysis fits nicely into the NPA framework:
one can run a path-expression algorithm on the CFGs of
procedures to construct an interprocedural equation system with regular-expression right-hand sides.
\citeauthor{POPL:RTP16} then developed the differential of Kleene-star for semirings that admits a regular algebra
and thus obtained a generalization of NPA that can analyze programs with unstructured control-flow.

However, as observed by~\citet{PLDI:WHR18,ENTCS:WHR19},
because of multiple kinds of program confluences,
it is more suitable to use \emph{control-flow hyper-graphs} (CFHG), rather than ordinary CFGs, to represent probabilistic programs.
PMAs are also tightly coupled with the hyper-graph-based representation,
where conditional-branching, probabilistic-branching, and nondeterministic-branching are represented
by \emph{hyper-edges}, each of which consists of one source node and two destination nodes.
In an ordinary CFG, a node can also have multiple successors, but the operation
applied at confluence nodes is usually the same operation (e.g., join).
Therefore, to apply the recipe of algebraic program analysis (and then fit it into the NPA framework),
the challenges are to devise a notion of ``path expressions'' and a ``path-expression'' algorithm for CFHGs,
as well as develop the interpretation and the differential of such ``path expressions'' in an $\omega$PMA.

\cref{Se:ThisWorkRegularHyperPathExpressions,Se:ThisWorkDifferentialsOfRegularHyperPathExpressions} present
an overview of our solution.
First, observing that program paths are compositions of edges in ordinary CFGs, we study \emph{hyper-paths},
which are (possibly infinite) compositions of hyper-edges in CFHGs,
and devise \emph{regular hyper-path expressions} that characterize control-flow hyper-paths,
in a sense similar to the way regular expressions characterize regular sets of control-flow paths.
%
%
%
Then, we propose a new algorithm, which is a variation of Gaussian elimination, to compute a regular hyper-path
expression for a CFHG.
We also develop an interpretation of regular hyper-path expressions in an $\omega$PMA, using
the analysis-supplied strategy that can solve linear equation systems.
%
%
Finally, we define the differential of regular hyper-path expressions with respect to an $\omega$PMA,
in a way that Newton's method can be used to solve equation systems with regular-hyper-path-expression
right-hand sides over $\omega$PMAs.
}



\begin{table}
\centering
\caption{Comparison in terms of supported program and analysis features. ``recur.'' stands for ``recursion;'' ``ucf.'' stands for ``unstructured control-flow;'' ``mco.'' stands for ``multiple confluence operations;'' and ``iter.'' and ``non-iter.'' stand for ``iterative'' and ``non-iterative,'' respectively.}
\label{Ta:FeatureComparison}
\begin{footnotesize}
\begin{tabular}{c|c c c|c c c}
  \hline
  & \multicolumn{3}{c|}{Program features} & \multicolumn{3}{c}{Analysis features} \\ \hhline{~------}
  & recur. & ucf. & mco. & loops & linear recur. & non-linear recur. \\ \hline
  PAI~\cite{ESOP:CM12} & & \checkmark & \checkmark & iter. & N/A & N/A \\
  PMAF~\cite{PLDI:WHR18} & \checkmark & \checkmark & \checkmark & iter. & iter. & iter. \\
  NPA~\cite{JACM:EKL10} & \checkmark & & & non-iter. & non-iter.\tablefootnote{NPA creates subproblems that are equivalent to analyzing linearly recursive programs, but does not provide a non-iterative method for solving linear recursion. \citet{POPL:RTP16} devised a non-iterative method for solving linear recursion (NPA-TP).} & iter. \\ \hline
  \framework{} (this work) & \checkmark & \checkmark & \checkmark & non-iter. & non-iter.\tablefootnote{We allow an instance analysis to supply a non-iterative linear-recursion-solving strategy. See \cref{Se:ThisWorkAConfluenceCentricAnalysisFramework,Se:SolvingLinearEquations}. } & iter. \\
  \hline
\end{tabular}
\end{footnotesize}
\end{table}

\paragraph{Contributions}
The article's contributions include the following:
\begin{itemize}[nosep]
  \item
    We develop \framework{}, an algebraic framework for interprocedural program analyses of probabilistic programs.
    \framework{} provides a novel approach to analyzing loops and recursion in probabilistic programs with nondeterminism and unstructured control-flow
    (\cref{Se:Overview}).
    \cref{Ta:FeatureComparison} gives a comparison between \framework{} and prior analysis frameworks for probabilistic programs in terms of what program features and analysis features are supported.

  \item
    We adopt \emph{regular infinite-tree expressions} to encode the hyper-path through a control-flow hyper-graph and devise a formulation of their differential (required by Newton's method) in a new family of algebras, \emph{$\omega$-continuous pre-Markov algebras} ($\omega$PMAs)
    (\cref{Se:TechnicalDetails}).
    We also develop a theory of regular hyper-paths that justifies the use of regular infinite-tree expressions in our framework.
  
  \item
    We created a prototype implementation of \framework{} and instantiated it with four different abstract domains for analyzing probabilistic programs, two of which are new (\cref{Se:CaseStudies}).
    Our preliminary evaluation
    shows
    that \framework{} outperforms Kleene iteration on the two existing abstract domains.
    We developed two new abstract domains for expectation-invariant analysis to show the
    generality of \framework{} for designing program analyses.
\end{itemize}
\cref{Se:RelatedWork} discusses related work.
\cref{Se:Conclusion} concludes.


\section{Overview}
\label{Se:Overview}

This section explains the main ideas of \framework{}.
\cref{Se:PriorWorkNewtonianProgramAnalysis} reviews Newtonian Program Analysis (NPA)
for interprocedural dataflow analysis.
\cref{Se:TheGap} discusses the gap between NPA and our work in terms of the support
of general state space and multiple kinds of branching.
\cref{Se:ThisWorkAConfluenceCentricAnalysisFramework}
describes \framework{}'s algebraic foundation and program representation
for bridging the gap, and how it applies Newton's method to solve
interprocedural analysis of probabilistic programs via analysis-supplied strategies.

\subsection{Review: Newtonian Program Analysis (NPA)}
\label{Se:PriorWorkNewtonianProgramAnalysis}

A \emph{semiring} is an algebraic structure $\calS = \tuple{D,{\oplus},{\otimes},{\azero},{\aone}}$,
where $\tuple{D,{\otimes},\aone}$ is a monoid
(i.e., $\otimes$ is an associative binary operation with $\aone$ as its identity element);
$\tuple{D,{\oplus},\azero}$ is a monoid in which $\oplus$ (\emph{combine}) is commutative;
and $\otimes$ (\emph{extend}) distributes over $\oplus$.
An \emph{$\omega$-continuous semiring} is a semiring in which
(i) the relation ${\aord} \defeq \{ (a,b) \in D \times D \mid \Exists{d} a \oplus d = b \}$ is a partial order,
(ii) every $\omega$-chain $a_0 \aord a_1 \aord \cdots \aord a_n \aord \cdots$ has a supremum with respect to $\aord$, denoted by $\bigsqcup^{\uparrow}_{i \in \bbN} a_i$, and
(iii) both $\oplus$ and $\otimes$ are $\omega$-continuous in both arguments.
A mapping $f : D \to D$ is \emph{$\omega$-continuous},
  if $f$ is monotone and for any $\omega$-chain $(a_i)_{i \in \bbN}$ it holds that $f(\bigsqcup^{\uparrow}_{i \in \bbN} a_i) = \bigsqcup^{\uparrow}_{i \in \bbN} f(a_i)$.

\begin{example}[The real semiring]\label{Exa:RealSemiring}
  A canonical example of $\omega$-continuous semirings is the \emph{real semiring} whose
  elements are nonnegative real numbers together with infinity: $\tuple{\bbR_{\ge 0} \cup \{\infty\}, {+}, {\cdot}, 0, 1}$.
\end{example}

NPA generalizes Newton's method to solve interprocedural equation systems over a semiring.
We demonstrate NPA with the finite-state probabilistic program shown in \cref{Fi:RunningExampleNPA}(a).
Its termination probability can be obtained as the least solution to an equation system
that encodes the program with respect to the real semiring.
We encode procedure $X$ as an equation $X = f(X)$ where
$f(X) \defeq \underline{\frac{1}{2}} {\oplus} (\underline{\frac{1}{2}} {\otimes} X {\otimes} X)$, where
underlining denotes semiring constants.
The equation can be viewed as a representation of
procedure $X$'s interprocedural control-flow graph (CFG);
see \cref{Fi:RunningExampleNPA}(b).

\begin{wrapfigure}{r}{0.22\textwidth}
\centering
\vspace{-0.2em}
\begin{subfigure}[b]{0.22\textwidth}
\centering
\begin{small}
\begin{pseudo*}
  \kw{proc} $X$() \kw{begin} \\+
    \kw{if} \kw{prob}($\frac{1}{2}$) \kw{then} \\+
      \kw{skip} \\-
    \kw{else} \\+
      $X$(); $X$() \\-
    \kw{fi} \\-
  \kw{end}
\end{pseudo*}
\end{small}
\vspace{-0.5em}
\caption{An example program}
\end{subfigure}
\begin{subfigure}[b]{0.22\textwidth}
\centering
\textsf{\small{proc $X$}}
\vspace{0.5em}

\begin{tikzpicture}[node/.style={circle,fill,minimum size=0pt,inner sep=2pt,font=\small},every edge quotes/.style={font=\small},node distance=0.4cm]
  \node[node] (v1) {};
  \node[node] (v2) [below right=of v1] {};
  \node[node] (v3) [below=of v2] {};
  \node[node] (v4) [below left=of v3] {};

  \draw (v1.south east) edge["$\,\,\sfrac{1}{2}\!\!$",above,inner sep=1.5pt,->] (v2.north west);
  \draw (v2.south) edge["$\!X$",right,->] (v3.north);
  \draw (v3.south west) edge["$X$",below right,inner sep=0.5pt,->] (v4.north east);
  \draw (v1.south) edge["$\sfrac{1}{2}\!$",left,->] (v4.north);
\end{tikzpicture}
\caption{$X=f(X)$}
\end{subfigure}
\begin{subfigure}[b]{0.22\textwidth}
\centering
\textsf{\small{proc $Y$}}
\vspace{0.5em}

\begin{tikzpicture}[node/.style={circle,fill,minimum size=0pt,inner sep=2pt,font=\small},every edge quotes/.style={font=\small},node distance=0.4cm]
  \node[node] (v1) {};
  \node[node] (v2) [below right=0.4cm and 0.2cm of v1] {};
  \node[node] (v22) [right=of v2] {};
  \node[node] (v3) [below=of v2] {};
  \node[node] (v32) [right=of v3] {};
  \node[node] (v4) [below left=0.4cm and 0.2cm of v3] {};

  \draw (v1.south east) edge["$\sfrac{1}{2}$",right,inner sep=0.5pt,->] (v2.north west);
  \draw (v2.south) edge["$\!Y$",right,->] (v3.north);
  \draw (v3.south west) edge["$\nu$",right,inner sep=1.5pt,->] (v4.north east);
  \draw (v1.west) edge["$\delta\!$",left,->,bend right=90,min distance=0.8cm] (v4.west);
  \draw (v1.east) edge["$\sfrac{1}{2}$",right,inner sep=3.5pt,->] (v22.north);
  \draw (v22.south) edge["$\nu$",right,->] (v32.north);
  \draw (v32.south) edge["$Y$",below right,inner sep=1pt,->] (v4.east);
\end{tikzpicture}
\caption{$Y=\delta \oplus \calD f|_{\nu}(Y) $}
\end{subfigure}
\caption{The running example of NPA.}
\label{Fi:RunningExampleNPA}
\end{wrapfigure}
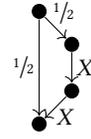
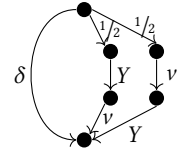

One can apply \emph{Kleene iteration} to find the least fixed-point of $X=f(X)$ via
the sequence $\kappa^{(0)} = \azero$ and $\kappa^{(i+1)} = f(\kappa^{(i)})$ for $i \in \bbN$.
With NPA, we solve the following sequence of subproblems for $\nu$:\footnote{
\citeauthor{ICALP:EKL08} described a general technique for solving the multivariate case,
i.e., finding the least fixed-point of $\vec{X} = \vec{f}(\vec{X})$.
In \cref{Se:PriorWorkNewtonianProgramAnalysis}, for brevity, we only describe their technique for solving the univariate case.}
\begin{wrapped}{align}
  \nu^{(0)} & = f(\azero), & \nu^{(i+1)} = \nu^{(i)} \oplus \Delta^{(i)}, \label{Eq:NewtonIteration}
\end{wrapped}
where $\Delta^{(i)}$ is the least solution of
\begin{wrapped}{equation}\label{Eq:NewtonSubproblem}
Y = \delta^{(i)} \oplus \calD f|_{\nu^{(i)}}(Y),
\end{wrapped}
$\delta^{(i)}$ is any element that satisfies $f(\nu^{(i)}) = \nu^{(i)} \oplus \delta^{(i)}$,
and $\calD f|_{\nu^{(i)}}(Y)$ is the \emph{differential} of $f$ at $\nu^{(i)}$.\footnote{
  When \cref{Eq:NewtonIteration} is compared with the high-level intuition given as \cref{Eq:NPAIntuitive}, the application of $f$ to $\nu^{(i)}$ on the right-hand side of \cref{Eq:NewtonIteration} is a bit hidden.
  $\Delta^{(i)}$ is the least solution of \cref{Eq:NewtonSubproblem}, and therefore $\Delta^{(i)} \sqsupseteq \delta^{(i)} \oplus \calD f|_{\nu^{(i)}}(\azero)$.
  Because $f(\nu^{(i)}) = \nu^{(i)} \oplus \delta^{(i)}$, the right-hand side of \cref{Eq:NewtonIteration} can be seen as $\nu^{(i)} \oplus \Delta^{(i)} \sqsupseteq \nu^{(i)} \oplus \delta^{(i)} \oplus \calD f|_{\nu^{(i)}}(\azero) = f(\nu^{(i)}) \oplus \calD f|_{\nu^{(i)}}(\azero)$.
}

\begin{definition}[\cite{ICALP:EKL08,JACM:EKL10}]\label{De:UnivariateDifferential}
  Let $f(X)$ be a function expressed as an expression in an $\omega$-continuous semiring.
  The \emph{differential} of $f(X)$ at $\nu$, denoted by $\calD f|_{\nu}(Y)$, is defined as follows:
  {\small\begin{wrapped}{equation*}
  \begin{array}{|l|}
  \hline
  \calD f|_{\nu}(Y) \defeq \begin{dcases*}
    \azero & \textrm{if} $f = a \in D$ \\
    Y & \textrm{if} $f = X$ \\
    \calD g|_\nu(Y) \oplus \calD h|_\nu(Y) & \textrm{if} $f = g \oplus h$ \\
    (\calD g|_\nu(Y) \otimes h(\nu)) \oplus (g(\nu) \otimes \calD h|_\nu(Y)) & \textrm{if} $f = g \otimes h$
  \end{dcases*} \\ \hline
  \end{array}
  \end{wrapped}}
\end{definition}

As discussed in \cref{Se:Introduction}, \cref{Eq:NewtonIteration,Eq:NewtonSubproblem} resemble Kleene iteration,
except that in each iteration, $f(\nu^{(i)})$ is adjusted by a \emph{linear} correction term $\calD f|_{\nu^{(i)}}(Y)$.
A consequence of \cref{De:UnivariateDifferential} is that the right-hand side of
\cref{Eq:NewtonSubproblem} is of the form $m_1 \oplus \cdots \oplus m_\ell$,
where $\ell \ge 1$, and $m_1,\cdots,m_\ell$ are each \emph{linear} expressions of the form
$a \otimes Y \otimes b$ for $a,b\in D$.

\begin{remark}
  Note how the case for ``$g \otimes h$'' in \cref{De:UnivariateDifferential} resembles the
  (univariate)
  Leibniz
  product rule from differential calculus:
  $
  d(g \ast h) = d g \ast h + g \ast d h .
  $
\end{remark}

\begin{remark}
  The key property of the differentials is that for any elements $\nu$ and $\Delta$, it holds that
  \begin{equation}\label{Eq:NewtonLinearizationProperty}
  f(\nu) \oplus \calD f|_\nu(\Delta) \aord f(\nu \oplus \Delta),
  \end{equation}
  i.e., the linearization of $f(\nu \oplus \Delta)$ at $\nu$ is always an
  \textbf{under-approximation}, thus the iteration defined in \cref{Eq:NewtonIteration}
  will not ``overshoot'' the least fixed-point.
  We include a formal discussion in \cref{Se:AnalysisNewton}.
\end{remark}

Recall the interprocedural equation $X = f(X)$ of the example probabilistic program shown in \cref{Fi:RunningExampleNPA}(a)
where $f(X) \defeq \underline{\frac{1}{2}} \oplus (\underline{\frac{1}{2}} \otimes X \otimes X)$.
The differential of $f$ at $\nu$ is
$
  \calD (\underline{\tfrac{1}{2}} {\oplus} (\underline{\tfrac{1}{2}} {\otimes} X {\otimes} X)) |_\nu(Y)  =
  \azero {\oplus} ( (\azero {\otimes} \nu {\otimes} \nu) {\oplus} (\underline{\tfrac{1}{2}} {\otimes} Y {\otimes} \nu) {\oplus} (\underline{\tfrac{1}{2}} {\otimes} \nu {\otimes} Y) )
  = (\underline{\tfrac{1}{2}} {\otimes} Y {\otimes} \nu) {\oplus} (\underline{\tfrac{1}{2}} {\otimes} \nu {\otimes} Y) .
$
Let $\delta$ be an element that satisfies $f(\nu) = \nu \oplus \delta$.
From \cref{Eq:NewtonSubproblem}, we then obtain the following linearized equation, which is also illustrated
graphically in \cref{Fi:RunningExampleNPA}(c):
\begin{equation}\label{Eq:NewtonLinearizedEquation}
Y = \delta \oplus (\underline{\tfrac{1}{2}} \otimes Y \otimes \nu) \oplus (\underline{\tfrac{1}{2}} \otimes \nu \otimes Y).
\end{equation}
Recall that we interpret the algebraic equation with respect to the real semiring to compute the termination probability.
Thus, on the $(i+1)^\textit{st}$ Newton round, we can compute $\delta$ as $\delta \defeq f(\nu^{(i)}) - \nu^{(i)}$, where
$\nu^{(i)}$ is the value for $\nu$ obtained in the $i^\textit{th}$ round.
Consequently, on the $(i+1)^\textit{st}$ Newton round, \cref{Eq:NewtonLinearizedEquation}
becomes the linear numeric equation
$
Y = ( ( \tfrac{1}{2} + \tfrac{1}{2} \cdot \nu^{(i)} \cdot \nu^{(i)} ) - \nu^{(i)}) + \tfrac{1}{2} \cdot Y \cdot \nu^{(i)} + \tfrac{1}{2} \cdot \nu^{(i)} \cdot Y,
$
which means that $\Delta^{(i)} \defeq (1-\nu^{(i)})/2$, and thus the iterative step in \cref{Eq:NewtonIteration}
is $\nu^{(i+1)} \gets \nu^{(i)} \oplus \Delta^{(i)} = (1+\nu^{(i)})/2$.
Because $\nu^{(0)} \defeq f(\azero) = \frac{1}{2}$, we have $\nu^{(i)} = 1-2^{-(i+1)}$ for any $i \in \bbN$, i.e.,
the Newton sequence converges to $1$ (so the example program shown in \cref{Fi:RunningExampleNPA}(a) terminates with probability one),
and gains one bit of precision per Newton iteration.
In contrast, Kleene iteration would yield a sequence such that $\kappa^{(i)} \le 1 - \frac{1}{i+1}$ for $i \in \bbN$,
and thus one needs to perform about $2^i$ Kleene rounds to obtain $i$ bits of precision---i.e.,
Kleene iteration converges exponentially slower.

So far we have assumed that we can
encode a procedure $X$ as an equation $X=f(X)$
where $f(X)$ is \emph{polynomial} of $X$ in a semiring,
i.e., an expression formed by operations $\oplus$, $\otimes$, $X$, and semiring constants.
However, in the general case where the control-flow of $X$ contains loops or
jumps (e.g., \kw{goto}, \kw{break}, and \kw{continue}), it is unclear if we can
still encode $f(X)$ as some finite representation on which we
can define the differentiation operator.
It is a standard workaround that one can transform a program to a
loop-free one by introducing auxiliary procedures.
On the other hand, recently, \citet{POPL:RTP16} proposed a technique that extends NPA to solve equations
whose right-hand sides are expressions in a regular algebra,
which uses a Kleene-star operation $\circledast$ to encode
looping control-flow.
%
\citeauthor{POPL:RTP16}'s technique then
computes the differential of Kleene-star as $
\calD f|_\nu(Y) \defeq (g(\nu))^{\circledast} \otimes \calD g|_{\nu}(Y) \otimes (g(\nu))^{\circledast}$,
if $f = g^\circledast$.
This rule---like the other rules in \cref{De:UnivariateDifferential}---produces a linear expression.\footnote{
The technique was developed for \emph{idempotent} semirings;
the semirings one usually works with in a probabilistic setting, e.g., the real semiring,
are typically not idempotent.}
Thus, by integrating with a path-expression algorithm (e.g., \cite{JACM:Tarjan81:Alg})
that computes a regular expression that recognizes the set of paths through a procedure,
one can apply NPA to analyze programs with unstructured control-flow.

\subsection{Gap: General State Space and Multiple Confluences}
\label{Se:TheGap}

\begin{figure}
\centering
\begin{subfigure}[b]{0.32\textwidth}
\centering
\begin{small}
\begin{pseudo*}
  \kw{proc} $X_1$() \kw{begin} \\+
    \kw{if} \kw{prob}($\frac{2}{3}$) \kw{then} \\+
      \kw{skip} \\-
    \kw{else} \\+
      $t \coloneqq t + 1$; \\
      $X_1$(); \\
      $X_1$() \\-
    \kw{fi} \\-
  \kw{end}
\end{pseudo*}
\end{small}
\caption{General state space}
\end{subfigure}
\begin{subfigure}[b]{0.32\textwidth}
\centering
\begin{small}
\begin{pseudo*}
  \kw{proc} $X_2$() \kw{begin} \\+
    \kw{if} $b$ \kw{then} \\+
      \kw{skip} \\-
    \kw{else} \\+
      \kw{if} \kw{prob}($\frac{1}{3}$) \kw{then} $b \coloneqq \kw{true}$ \\
      \kw{else} $b \coloneqq \kw{false}$ \kw{fi}; \\
      $X_2$() \\-
    \kw{fi} \\-
  \kw{end}
\end{pseudo*}
\end{small}
\caption{Multiple kinds of branching}
\end{subfigure}
\begin{subfigure}[b]{0.32\textwidth}
\centering
\begin{small}
\begin{pseudo*}
  \kw{proc} $X_3$() \kw{begin} \\+
    \kw{if} $z \le 0$ \kw{then} \\+
      \kw{skip} \\-
    \kw{else} \\+
      \kw{if} \kw{prob}($\frac{1}{2}$) \kw{then} $z \coloneqq z + 1$ \\
      \kw{else} $z \coloneqq z - 2$ \kw{fi}; \\
      $X_3$() \\-
    \kw{fi} \\-
  \kw{end}
\end{pseudo*}
\end{small}
\caption{Combination of (a) and (b)}
\end{subfigure}
\caption{Examples to demonstrate the gap.}
\label{Fi:GapExamples}
\end{figure}

\cref{Se:PriorWorkNewtonianProgramAnalysis} showed that NPA can be instantiated with
the real semiring to analyze the termination probability of a finite-state probabilistic
program whose branches are all purely probabilistic.
A natural question then arises: is NPA capable of analyzing general-state probabilistic
programs with more than purely probabilistic branches?
In the four examples below, we argue that NPA can partly support those features, but
\emph{cannot} deal with a combination of them in general.

\begin{example}[General state space]\label{Exa:NPAGeneralStateSpace}
NPA uses $\omega$-continuous semirings as its algebraic foundation, which does
not impose any restriction on the finiteness of the state space.
Therefore, when all the branches are probabilistic (like $\kw{prob}(\frac{1}{2})$ in
\cref{Fi:RunningExampleNPA}(a)) and we know how to turn probabilities into semiring
constants, we can use NPA to carry out the analysis.
For example, \cref{Fi:GapExamples}(a) shows a program that manipulates a real-valued
program variable $t$; thus, the state space of this program is not finite.
We now construct a semiring $\calS_1 = \tuple{D_1,\oplus_1,\otimes_1,\azero_1,\aone_1}$
to analyze the expected difference between the final and initial value of $t$.
The domain includes pairs of nonnegative real numbers
$D_1 \defeq (\bbR_{\ge 0} \cup \{\infty\}) \times (\bbR_{\ge 0} \cup \{\infty\})$,
in the sense that a pair $(p,d)$ records the termination probability $p$ and
the expected difference $d$.\footnote{We formulate $(p,d)$ in a way that $d$ is scaled by $p$; that is, the actual expected difference is $\frac{d}{p}$.}
We then map probabilities into $\calS_1$ as $\underline{p} \defeq (p,0)$
and assignments to $t$ as $\interp{t \coloneqq t + d} \defeq (1,d)$.
The binary operators and constants of $\calS_1$ are defined as follows:
\begin{alignat*}{4}
  (p_1,d_1) {\oplus_1} (p_2,d_2) & \defeq (p_1{+}p_2,d_1{+}d_2), \enskip &
  (p_1,d_1) {\otimes_1} (p_2,d_2) & \defeq (p_1 p_2, p_1d_2 {+}p_2d_1), \enskip &
  \azero_1 & \defeq (0,0), \enskip & \aone_1 & \defeq (1,0) .
\end{alignat*}
Using this semiring, NPA is capable of deriving that for the procedure $X_1$
in \cref{Fi:GapExamples}(a), the termination probability is $1$ and the expected
difference of $t$ is also $1$.
\end{example}

\begin{example}[Multiple kinds of branching]\label{Exa:NPAMultipleBranching}
Probabilistic programs usually exhibit multiple kinds of branching behavior, e.g.,
probabilistic and conditional.
\cref{Fi:GapExamples}(b) gives a program with a single Boolean-valued program variable $b$,
where both probabilistic branching ($\kw{if}~\kw{prob}(\frac{1}{3})~\ldots$) and
conditional branching ($\kw{if}~b~\ldots$) exist.
At first glance, it seems that we cannot use NPA to analyze this program,
because a semiring contains only \emph{one} combine operation $\oplus$, whereas
this program requires \emph{two} different kinds of branching.
Nevertheless, because the program's state space is finite, we can transform procedure
$X_2$ into the following two procedures $X_{2,\kw{true}}$ and $X_{2,\kw{false}}$, which use
only purely probabilistic branches and correspond to the cases where the initial value
of $b$ is $\kw{true}$ and $\kw{false}$, respectively:
\begin{minipage}{\textwidth}
\centering
\begin{minipage}{0.3\textwidth}
\centering
\begin{small}
\begin{pseudo*}
  \kw{proc} $X_{2,\kw{true}}$() \kw{begin} \\+
    \kw{skip} \\-
  \kw{end}
\end{pseudo*}  
\end{small}
\end{minipage}
\begin{minipage}{0.6\textwidth}
\centering
\begin{small}
\begin{pseudo*}
  \kw{proc} $X_{2,\kw{false}}$() \kw{begin} \\+
    \kw{if} \kw{prob}($\frac{1}{3}$) \kw{then} $X_{2,\kw{true}}$() \kw{else} $X_{2,\kw{false}}$() \kw{fi} \\-
  \kw{end}
\end{pseudo*}  
\end{small}
\end{minipage}
\end{minipage}
In other words, we create extra procedures to simulate state changes explicitly.
This approach has been used to analyze finite-state probabilistic programs, by means
of modeling those programs as recursive Markov chains~\cite{STACS:EY05} and
probabilistic pushdown automata~\cite{LICS:EKM04}, as well as some special classes
of infinite-state programs that can be modeled by, e.g., probabilistic one-counter
automata~\cite{CAV:BKK11}.
\end{example}

\begin{example}[The general case]\label{Exa:NPAGeneralCase}
It is now natural to consider the combination of the previous two examples:
what if we want to analyze infinite-state probabilistic programs with multiple kinds
of branching?
For example, \cref{Fi:GapExamples}(c) implements a one-axis random walk, using
an integer-valued program variable $z$ to keep track of the current position.
The program uses a conditional choice ($z \le 0$) to determine the stopping criterion
and a probabilistic choice ($\kw{prob}(\frac{1}{2})$) to move randomly on the axis.
We can still follow the methodology of \cref{Exa:NPAGeneralStateSpace} to devise
a semiring to capture semantic properties; for example, we can use distribution
transformers $\bbN \to \underline{\bbD}(\bbN)$ as the domain, where $\underline{\bbD}(\bbN)$
denotes the set of sub-probability distributions on natural numbers,
and then define the combine
operation $\oplus$ to be pointwise distribution addition and the extend operation
$\otimes$ to be transformer composition~\cite{JCSS:Kozen81}.

However, such a domain is usually too complex to be automated as a program analysis.
Instead, people usually apply \emph{abstractions} when designing program analyses.
For example, several lines of work have applied predicate abstraction~\cite{CAV:GS97}
to analyze probabilistic programs and systems~\cite{CAV:HWZ08,VMCAI:KKN09,TACAS:HHW10,ATVA:FHH12,ICML:HBM18}.
A consequence of using abstraction is that some parts of the analyzed program need to
be abstracted by \emph{nondeterminism}.
To illustrate this issue, let us consider analyzing the program in \cref{Fi:GapExamples}(c)
with predicates $(z \le 0)$ and $(z = 1)$.
When both predicates are false---which means that $z > 1$---we have to model the 
assignment $z \coloneqq z - 2$ as a nondeterministic assignment to the two predicates,
because, e.g., $(z =1)$ should become \kw{true} if $z$ is $3$, otherwise it should
be \kw{false}.
In general, the analysis needs to support \emph{nondeterministic branching} because the abstraction
might lack information to model some branching conditions.
It is worth noting that nondeterminism also arises in some other scenarios of analyzing probabilistic
programs~\cite{SFM:FKN11,book:MM05}; for example, the analyzed program might interact with
components with unknown behavior (e.g., system calls).

We have already seen three kinds of branching: probabilistic, conditional, and nondeterministic.
A semantic algebra is then supposed to have three operations for them, respectively;
in this article, we refer to such operations as \emph{confluence operations}.
Although in some cases (such as \cref{Exa:NPAMultipleBranching}) we can unify some subsets
of those confluence operations, there is in general no easy way to unify probabilistic-
and nondeterministic-choices.
One intuitive reason is that a probabilistic confluence calculates a weighted sum,
whereas a nondeterministic one calculates a set union.
For example, if we were using a semiring to analyze the program below, where
the $\star$ symbol denotes nondeterministic-choice,
\begin{minipage}{\textwidth}
\centering
\begin{small}
\begin{pseudo*}
\kw{if} $\star$ \kw{then} \quad (\kw{if} \kw{prob}($\frac{1}{2}$) \kw{then} $z \coloneqq 0$ \kw{else} $z \coloneqq 1$ \kw{fi}) \quad \kw{else} \quad (\kw{if} \kw{prob}($\frac{2}{3}$) \kw{then} $z \coloneqq 0$ \kw{else} $z \coloneqq 1$ \kw{fi}) \quad \kw{fi}
\end{pseudo*}  
\end{small}
\end{minipage}
we would obtain the algebraic expression $(\underline{\frac{1}{2}} \otimes \interp{z \coloneqq 0}) \oplus (\underline{\frac{1}{2}} \otimes \interp{z \coloneqq 1}) \oplus (\underline{\frac{2}{3}} \otimes \interp{z\coloneqq 0}) \oplus (\underline{\frac{1}{3}} \otimes \interp{z \coloneqq 1})$,
which looks like we derive a distribution in which the probability mass sums to $2$!
Therefore, the algebraic foundation of NPA---semirings---is not sufficient in the general case
of analyzing general-state probabilistic programs with multiple confluence operations. 
\end{example}

\paragraph{Summary}
The article presents a solution to support general state space and multiple
confluences and thus develops a framework for Newtonian program analysis of
probabilistic programs. 
We sketch our solution in \cref{Se:ThisWorkAConfluenceCentricAnalysisFramework}
and present it in greater detail in \cref{Se:TechnicalDetails}.

\subsection{Algebras, Expressions, and Newton's Method for Multiple Confluence Operations}
\label{Se:ThisWorkAConfluenceCentricAnalysisFramework}

\begin{wrapfigure}{r}{0.22\textwidth}
\centering
\begin{subfigure}{0.22\textwidth}
\centering
\begin{small}
\begin{pseudo*}
  \kw{proc} $X$() \kw{begin} \\+
    \kw{if} \kw{prob}($\frac{1}{3}$) \kw{then} \kw{skip} \\
    \kw{else} \\+
      $X$(); \\
      \kw{if} $\star$ \kw{then} $X$() \kw{fi} \\-
    \kw{fi} \\-
  \kw{end}
\end{pseudo*}
\end{small}
\vspace{-0.5em}
\caption{An example program}
\end{subfigure}
\begin{subfigure}{0.22\textwidth}
\centering
\begin{tikzpicture}[op/.style={fill=gray!40,inner sep=3pt},node distance=0.15cm,font=\small]
  \node (a) {$\mi{prob}[\frac{1}{3}]$};
  \node (b) [below left=0.15cm and -0.1cm of a] {$\aone$};
  \node (c) [below right=0.15cm and -1cm of a] {$\mi{call}[X]$};
  \node (d) [below=0.15cm of c] {$\mi{ndet}$};
  \node (e) [below left=0.15cm and -0.5cm of d] {$\mi{call}[X]$};
  \node (f) [below=0.15cm of e] {$\aone$};
  \node (g) [below right=0.15cm and -0.1cm of d] {$\aone$};

  \draw (a.south) ++ (-5pt,0) edge[->] (b.north);
  \draw (a.south) edge[->] (c.north);
  \draw (c.south) edge[->] (d.north);
  \draw (d.south) edge[->] (e.north);
  \draw (e.south) edge[->] (f.north);
  \draw (d.south) ++ (5pt,0) edge[->] (g.north);
\end{tikzpicture}
\vspace{-0.5em}
\caption{$f(X)$ as a tree}
\end{subfigure}
\caption{A nondeterministic recursive program.}
\label{Fi:ExampleProgramWithNonlinearRecursionAndNondeterminism}
\end{wrapfigure}
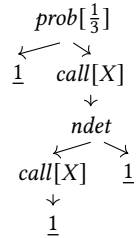

To overcome the multiple-confluence-operation issue,
we devise \emph{$\omega$-continuous pre-Markov algebras} ($\omega$PMAs), which can
be seen as a refinement of $\omega$-continuous semirings and previously proposed pre-Markov algebras (PMAs)~\cite{PLDI:WHR18}.
In this section, we first demonstrate an $\omega$PMA $\calR$, then present the
\emph{representation} and \emph{interpretation} of loop-free probabilistic programs
with respect to $\calR$ and how we define the \emph{differentiation} operator
over such a representation, and finally sketch how we exercise the
representation-interpretation-differentiation recipe to support
probabilistic programs with unstructured control-flow, e.g., loops and jumps.

\paragraph{Algebra}
To demonstrate $\omega$PMAs, let us
consider a termination-probability analysis that computes
the \emph{lower bound} on the termination probability.
For probabilistic-choice (e.g., ``\kw{if} \kw{prob}($\frac{1}{2}$) \kw{then} \ldots \kw{else} \ldots \kw{fi}''),
we need an operation that performs addition,
as we have shown in \cref{Se:PriorWorkNewtonianProgramAnalysis}.
For nondeterministic-choice, we instead need an operation that computes the minimum of two values.
Therefore,
we extend the real semiring to be an $\omega$PMA
$\calR = \tuple{\bbR_{\ge 0} \cup \{\infty\}, {\oplus}, {\otimes}, {\gcho{\varphi}}, {\pcho{p}}, {\dashcup}, \azero, \aone}$,
where $\tuple{\bbR_{\ge 0} \cup \{\infty\}, {\oplus}, {\otimes}, \azero,\aone}$ is the real semiring,
and
$\gcho{\varphi}$ (parameterized by a logical condition $\varphi$), $\pcho{p}$ (parameterized by $p \in [0,1]$), and $\dashcup$
are three binary confluence operations
that correspond to conditional-choice, probabilistic-choice, and nondeterministic-choice, respectively.
The confluence operations are defined as $r_1 \pcho{p} r_2 \defeq p\cdot r_1 + (1{-}p) \cdot r_2$ for any $p \in [0,1]$
and $r_1 \dashcup r_2 \defeq \min(r_1,r_2)$.
Because the example does not contain conditional branching, we do not need the $\gcho{\varphi}$ operation.
In general, for any $r_1,r_2$ that satisfy $r_1 \sqsupseteq r_2$, we let $r_1 \ominus r_2$ denote any element $d$ that satisfies $r_2 \oplus d = r_1$.
Here, for nonnegative real numbers, as $r_1 \sqsupseteq r_2 \iff r_1 \ge r_2$, we can define $r_1 \ominus r_2 \defeq r_1 - r_2$.

\begin{remark}
For analysis of non-probabilistic programs with semirings,
the combine operation $\oplus$ plays two roles:
(i) it defines the partial order---sometimes called the abstraction order---in the semiring, and
(ii) it interprets nondeterministic-choice.
However, when probability comes into the setting, especially when
probability interacts with nondeterminism, the order induced by
nondeterministic-choice is usually different from the partial order
in the semantic domain~\cite{ENTCS:TKP09,ENTCS:WHR19}.
Therefore, our design of $\omega$PMAs has both the combine $\oplus$ and
the nondeterministic-choice $\dashcup$ operations, and $\oplus$ is mainly
used to define a partially ordered additive structure. 
\end{remark}

\paragraph{Representation}
We use the program in \cref{Fi:ExampleProgramWithNonlinearRecursionAndNondeterminism}(a)
to demonstrate our design of the form for right-hand sides of equation systems.
We encode a loop-free probabilistic program
as a \emph{simple tree expression}, in which every node corresponds to a program command and every root-to-leaf
path corresponds to a program path.
Confluences are encoded as internal nodes with multiple (typically two) child nodes in the tree.
\cref{Fi:ExampleProgramWithNonlinearRecursionAndNondeterminism}(b) presents the abstract syntax tree (AST) of
the following simple tree expression that corresponds to \cref{Fi:ExampleProgramWithNonlinearRecursionAndNondeterminism}(a),
where $\aone$ denotes \kw{skip}, $\mi{call}[X](\cdot)$ is a procedure call to $X$
whose argument encodes the continuation of the call,
$\mi{prob}[p]({\cdot},{\cdot})$ is a
probabilistic-choice command that takes the left branch with probability $p$, and $\mi{ndet}({\cdot},{\cdot})$ is nondeterministic-choice:
\begin{align}\label{Eq:ExampleRecursionNondeterminism}
  X & = f(X) \defeq \mi{prob}[\tfrac{1}{3}]\bigl(\aone, \mi{call}[X]\bigl( \mi{ndet}( \mi{call}[X]( \aone ), \aone ) \bigr) \bigr) .
\end{align}

\paragraph{Interpretation}
We use a syntax-directed
method to interpret simple tree expressions over $\calR$,
where we parameterize the
interpretation by a value $\nu$ for interpreting procedure calls to $X$:
{\begin{align*}
   \calR\interp{\underline{r}}(\nu) & \defeq \underline{r}, & \calR\interp{\mi{prob}[p](E_1,E_2)}(\nu) & \defeq \calR\interp{E_1}(\nu) \pcho{p}  \calR\interp{E_2}(\nu), \\
  \calR\interp{\mi{call}[X](E)}(\nu) & \defeq {\nu} \otimes \calR\interp{E}(\nu), & \calR\interp{\mi{ndet}(E_1,E_2)}(\nu) & \defeq \calR\interp{E_1}(\nu) \dashcup \calR\interp{E_2}(\nu).
\end{align*}}%
We usually use $f$ to denote its own interpretation, i.e., $f(\nu)$ means $\calR\interp{f(X)}(\nu)$.
The solution of an equation $X=f(X)$ is then a value $\nu$ that satisfies $\nu = f(\nu)$.

\paragraph{Differentiation}
We now describe how our extension of NPA works by showing how \framework{} solves \cref{Eq:ExampleRecursionNondeterminism} over $\calR$.
Recall that NPA requires a notion of \emph{differential} (see \cref{De:UnivariateDifferential} for the semiring case).
We want to define the differential of $f(X)$ (expressed as a simple tree expression) at a value $\nu$,
still denoted by $\calD f|_\nu(Y)$.
Because $\calR$ is an extension of the real semiring, it is natural for us to define $\calD f|_\nu(Y)$ as a
conservative extension of \cref{De:UnivariateDifferential}:
\begin{itemize}
  \item
  If $f = \underline{r}$: Because $\calR\interp{\underline{r}}(\nu) = \underline{r}$ and
  \cref{De:UnivariateDifferential} defines the differential of any constant to be zero,
  we derive $\calD f|_\nu(Y) \defeq \azero$.
  \item
  If $f = \mi{call}[X](g)$: Because $\calR\interp{\mi{call}[X](g)}(\nu) = \nu \otimes \calR\interp{g}(\nu)$ and the differential of $\otimes$ resembles the
  Leibniz
  product rule, we can treat $\mi{call}[X](g)$ as $X \otimes g$ and its differential at $\nu$ is $(Y \otimes g) \oplus (\nu \otimes \calD g|_\nu(Y))$.
  The monomial $Y \otimes g$ can then be transformed back to
  the tree form $\mi{call}[Y](g)$.
  For the monomial $\nu \otimes \calD g|_\nu(Y)$,
  where $\nu$ is a specific known value,
  we introduce a sequencing command $\mi{seq}[\underline{r}]$ whose interpretation is $\calR\interp{\mi{seq}[\underline{r}](E)}(\nu) \defeq \underline{r} \otimes \calR\interp{E}(\nu)$.
  Thus, we derive $\calD f|_\nu(Y) \defeq \mi{call}[Y](g(\nu)) \oplus \mi{seq}[\nu](\calD g|_\nu(Y))$.
  \item
  If $f = \mi{ndet}(g,h)$:
  Recall that we want $\calD f|_\nu(Y)$ to satisfy the under-approximation property \cref{Eq:NewtonLinearizationProperty}:
  \[
  \mi{ndet}(g(\nu),h(\nu)) \oplus \calD f|_\nu(Y) \aord \mi{ndet}(g(\nu \oplus Y), h(\nu \oplus Y)).
  \]
  Applying \cref{Eq:NewtonLinearizationProperty} to the sub-tree-expressions $g$ and $h$, we have
  \[
  g(\nu) \oplus \calD g|_\nu(Y) \aord g(\nu \oplus Y), \qquad h(\nu) \oplus \calD h|_\nu(Y) \aord h(\nu \oplus Y).
  \]
  If $\mi{ndet}$ is a \emph{monotone} operation, we have
  \[
  \mi{ndet}(g(\nu) \oplus \calD g|_\nu(Y), h(\nu) \oplus \calD h|_\nu(Y)) \aord \mi{ndet}(g(\nu \oplus Y), h(\nu \oplus Y) ).
  \]
  Thus, we can safely define $\calD f|_\nu(Y) \defeq \mi{ndet}(g(\nu) \oplus \calD g|_\nu(Y), h(\nu) \oplus \calD h|_\nu(Y)) \ominus \mi{ndet}(g(\nu),h(\nu))$.
  \item
  If $f = \mi{prob}[p](g,h)$:
  We could repeat the derivation for $\mi{ndet}$ by assuming $\mi{prob}[p]$ is a monotone operation;
  however, we take a different approach because of an observation that $\mi{prob}[p]$ usually
  satisfies more algebraic properties than $\mi{ndet}$, i.e., for any elements $a,b,c,d$,
  \[
  \mi{prob}[p](a \oplus b, c \oplus d) = \mi{prob}[p](a, c) \oplus \mi{prob}[p](b, d).
  \]
  For the $\omega$PMA $\calR$, the equation holds trivially because
  $p \cdot (a+b) + (1-p) \cdot (c+d) = p \cdot a + (1-p) \cdot c + p \cdot b + (1-p) \cdot d$.
  We then have
  \begin{align*}
  \mi{prob}[p](g(\nu \oplus Y), h(\nu \oplus Y))
  & \sqsupseteq \mi{prob}[p](g(\nu) \oplus \calD g|_\nu(Y), h(\nu) \oplus \calD h|_\nu(Y))\\
  & = \mi{prob}[p](g(\nu),h(\nu)) \oplus \mi{prob}[p](\calD g|_\nu(Y), \calD h|_\nu(Y)) .
  \end{align*}
  Thus, we can safely define $\calD f|_\nu(Y) \defeq \mi{prob}[p](\calD g|_\nu(Y), \calD h|_\nu(Y))$.
\end{itemize}
To summarize, we have derived
the following rules:
{\begin{equation*}\begin{array}{|l|} \hline
\calD f|_\nu(Y) \defeq \begin{dcases*}
  \underline{0} & if $f = \underline{r}$ \\
  \mi{call}[Y](g(\nu)) \oplus \mi{seq}[\nu]( \calD g|_\nu(Y)) & if $f = \mi{call}[X](g)$ \\
  \mi{ndet}(g(\nu) \oplus \calD g|_\nu(Y), h(\nu) \oplus \calD h|_\nu(Y)) \ominus \mi{ndet}(g(\nu), h(\nu)) & if $f = \mi{ndet}(g,h)$ \\
  \mi{prob}[p](\calD g|_\nu(Y), \calD h|_\nu(Y)) & if $f=\mi{prob}[p](g,h)$
\end{dcases*} \\ \hline \end{array}
\end{equation*}}

\paragraph{Newton's method for $\omega$PMAs}
We then apply Newton's method to find the least fixed-point of $X = f(X)$
by solving a sequence of subproblems for $\nu$:
$\nu^{(0)} = f(\underline{0})$ and $\nu^{(i+1)}  = \nu^{(i)} \oplus \Delta^{(i)}$ for $i \ge 0$,
where $\Delta^{(i)}$ is the least solution of
$
Y = \delta^{(i)} \oplus \calD f|_{\nu^{(i)}} (Y),
$
and $\delta^{(i)} \defeq f(\nu^{(i)}) \ominus \nu^{(i)}$.
\cref{Fi:RightHandSideOfTheLinearizedEquation} graphically shows the subproblem in terms of $Y$ for
the example program shown in \cref{Fi:ExampleProgramWithNonlinearRecursionAndNondeterminism}(a) when $X$ has some value $\nu$.
A key feature of the AST for the differential is that every root-to-leaf path contains
\emph{at most} one $\mi{call}[\cdot]$ node;
thus, the subproblem in terms of $Y$ is again a \emph{linearized} equation.
We now unfold the definition of all algebraic operations and
obtain the following linear numeric equation:
\begin{equation*}
Y = (\tfrac{1}{3} + \tfrac{2}{3} \cdot \nu \cdot \min(\nu, 1) - \nu) + \tfrac{2}{3} \cdot (Y \cdot \min(\nu, 1) + \nu \cdot ( \min(\nu + Y, 1 ) - \min(\nu,1) ) ).
\end{equation*}%
\begin{wrapfigure}{r}{0.3\textwidth}
\centering
\begin{tikzpicture}[op/.style={fill=gray!40,inner sep=3pt},node distance=0.3cm,font=\small]
  \node (a) {$\oplus$};
  \node (b) [below left=0.2cm and -0.1cm of a] {$\delta$};
  \node (c) [below right=0.2cm and -0.1cm of a] {$\mi{prob}[\frac{1}{3}]$};
  \node (d) [below left=0.2cm and -0.3cm of c] {$\azero$};
  \node (e) [below right=0.2cm and -0.5cm of c] {$\oplus$};
  \node (f) [below left=0.2cm and -0.1cm of e] {$\mi{call}[Y]$};
  \node (g) [below right=0.2cm and -0.1cm of e] {$\mi{seq}[\nu]$};
  \node (h) [below=0.2cm of f] {$(\nu \otimes \aone) \dashcup \aone$};
  \node (i) [below=0.2cm of g] {$\ominus$};
  \node (j) [below left=0.2cm and -0.1cm of i] {$\mi{ndet}$};
  \node (k) [below right=0.2cm and -0.1cm of i] {$(\nu \otimes \aone) \dashcup \aone$};
  \node (l) [below left=0.2cm and 0.3cm of j] {$\oplus$};
  \node (m) [below right=0.2cm and 0.3cm of j] {$\oplus$};
  \node (n) [below left=0.2cm and -0.1cm of l] {$\nu \otimes \aone$};
  \node (o) [below right=0.2cm and -0.3cm of l] {$\mi{call}[Y]$};
  \node (p) [below left=0.2cm and -0.1cm of m] {$\aone$};
  \node (q) [below right=0.2cm and -0.1cm of m] {$\azero$};
  \node (r) [below=0.2cm of o] {$\aone$};

  \draw (a.south) ++ (-2pt,0) edge[->] (b.north);
  \draw (a.south) ++ (2pt,0) edge[->] (c.north);
  \draw (c.south) ++ (-5pt,0) edge[->] (d.north);
  \draw (c.south) edge[->] (e.north);
  \draw (e.south) ++ (-2pt,0) edge[->] (f.north);
  \draw (e.south) ++ (2pt,0) edge[->] (g.north);
  \draw (f.south) edge[->] (h.north);
  \draw (g.south) edge[->] (i.north);
  \draw (i.south) ++ (-2pt,0) edge[->] (j.north);
  \draw (i.south) ++ (2pt,0) edge[->] (k.north);
  \draw (j.south) ++ (-2pt,0) edge[->] (l.north);
  \draw (j.south) ++ (2pt,0) edge[->] (m.north);
  \draw (l.south) ++ (-2pt,0) edge[->] (n.north);
  \draw (l.south) ++ (2pt,0) edge[->] (o.north);
  \draw (m.south) ++ (-2pt,0) edge[->] (p.north);
  \draw (m.south) ++ (2pt,0) edge[->] (q.north);
  \draw (o.south) edge[->] (r.north);
\end{tikzpicture}
\caption{$\delta \oplus \calD f|_\nu(Y)$ as a tree.}
\label{Fi:RightHandSideOfTheLinearizedEquation}
\end{wrapfigure}%
Our \framework{} framework then uses an analysis-supplied strategy to solve linear equations like the one shown above.
Here, because the equation is simple enough, we can solve it analytically.
Observing probabilities are real numbers in the interval $[0,1]$, we reduce the equation shown above to the following equivalent linear equation:
$
Y = (\tfrac{1}{3} + \tfrac{2}{3} \cdot \nu^2 - \nu) + \tfrac{4}{3} \cdot \nu \cdot Y.
$
Therefore, on the $(i+1)^\textit{st}$ Newton round, we solve the equation above where $\nu$ is set to $\nu^{(i)}$ to
obtain the correction term  $\Delta^{(i)} \defeq \frac{1 + 2 (\nu^{(i)})^2 - 3\nu^{(i)}}{3 - 4\nu^{(i)}}$,
and perform the assignment $\nu^{(i+1)} {\gets} \nu^{(i)} + \Delta^{(i)} = \frac{ 1 - 2(\nu^{(i)})^2 }{3 - 4\nu^{(i)} } $.
The first four elements of the Newton sequence are
$
\nu^{(0)}  = \tfrac{1}{3}, \nu^{(1)} = \tfrac{7}{15},  \nu^{(2)}  = \tfrac{127}{255}, \nu^{(3)}  = \tfrac{32767}{65535} \approx 0.499992,
$
whereas the first four elements of the Kleene sequence are
$
\kappa^{(0)}  = 0, \kappa^{(1)}  = \tfrac{1}{3},  \kappa^{(2)}  = \tfrac{11}{27},  \kappa^{(3)} = \tfrac{971}{2187} \approx 0.443987,
$
and the Newton sequence converges faster than the Kleene sequence.

\paragraph{Unstructured control-flow}
So far we have discussed algebraic analysis of loop-free probabilistic programs.
Following the representation-interpretation-differentiation recipe,
we support algebraic analysis of unstructured probabilistic programs
via a suitable class of tree expressions.
We cannot apply \citet{POPL:RTP16}'s technique (reviewed in \cref{Se:PriorWorkNewtonianProgramAnalysis}),
because a regular algebra still contains only one confluence operation, whereas
we require three, as discussed in \cref{Exa:NPAGeneralCase}.
Instead,
we adopt a standard technique that represents possibly-infinite trees as finite
$\mu$-terms~\cite{JTCS:Courcelle83,kn:Courcelle90}, namely \emph{regular infinite-tree expressions}.
In this section, we briefly describe our development around such regular infinite-tree expressions;
details are given in \cref{Se:TechnicalDetails}.

To represent a loop, we may unroll the loop indefinitely and obtain an \emph{infinite} tree.
Each regular infinite-tree expression is a finite representation of a
possibly-infinite tree, in which every node still corresponds to a program command,
every root-to-leaf path corresponds to a finite program path, and
every rooted indefinite path corresponds to an infinite program path.
Compared with simple tree expressions, regular infinite-tree expressions have a new
form $\mu Z.\, E$, called a \emph{$\mu$-binder},
where the \emph{bound variable} $Z$ can occur multiple times as a leaf node in tree expression $E$,
to encode self-substituting $E$ for $Z$ an infinite number of times.
For example, the infinite tree corresponding to ``\kw{while} $\kw{prob}(\frac{3}{4})$ \kw{do} $x {\coloneqq} x{+}1$ \kw{od}''
can be encoded as
$
\mu Z.\, \mi{prob}[\frac{3}{4}]\bigl( \mi{seq}[x{\coloneqq}x{+}1](Z), \varepsilon\bigr)  .
$
Note that we use $\varepsilon$ instead of $\aone$ to denote a no-op and
$\mi{seq}[\m{act}]$ (where $\m{act}$ is an action such as an assignment) instead of 
$\mi{seq}[\underline{r}]$ (where $\underline{r}$ is an algebra constant) to
denote a sequencing command.
In this way, the representation of a program is independent from the analysis.
Applying Beki{\'c}'s theorem~\cite{kn:Bekic84}, we can compute a regular infinite-tree
  expression from an unstructured probabilistic program.

The introduction of $\mu$-binders poses challenges to the development of
interpretation and differentiation.
For interpretation, a common approach for computing the meaning of a $\mu$-binder
is through some fixed-point mechanism, which is usually iterative.
Our observation is that we can always reduce the interpretation of a
regular infinite-tree expression to solving a system of linear equations whose
right-hand-sides are simple tree expressions (see \cref{Se:SolvingLinearEquations});
thus, we can apply the analysis-supplied
strategy that solves linear equations, which we have used to solve linearized equations
during a Newton round.
For differentiation, we need to compute the differential of $\mu$-binders.
The major challenge is to differentiate those binding structures.
Our idea is to differentiate a $\mu$-binder as if we are differentiating its
  corresponding infinite tree, intuitively one would again obtain an infinite tree,
  and finally we fold the new tree back to a $\mu$-binder (see \cref{Se:AnalysisNewton}).

\paragraph{Summary}
\cref{Fi:Pipeline} illustrates the overall pipeline of interprocedural program analysis in \framework{}.
The inputs of our framework are (i) a probabilistic program $\{X_i \mapsto E_{X_i}\}_{i=1}^n$, where each $E_{X_i}$ is the regular-infinite-tree-expression representation of the procedure $X_i$, and (ii) an abstract domain for the program analysis.
The output of our framework is a vector $\vec{\nu}$ of elements in the abstract domain,
where
each $\nu_i$ is the analysis result for the procedure $X_i$.
\framework{} first
applies a differentiation process (oval 1) and a normalization transformation (oval 2) to extract a system of linear equations;
and, repeatedly, updates some constants in the equations and solves the linear equations (oval 3).

\begin{figure}
\centering
\includegraphics[width=0.9\textwidth]{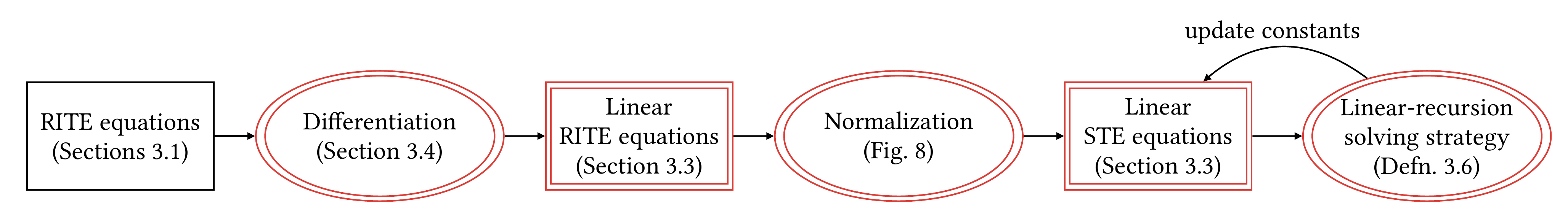}
\caption{The pipeline of \framework{}. ``RITE'' denotes ``regular infinite-tree expression.'' ``STE'' denotes ``simple tree expression.'' Red double borders indicate contributions of this work.
The Newtonian phase (the last oval) repeatedly updates some constants in the linear equations and solves the linear equations.
}
\label{Fi:Pipeline}
\end{figure}


\section{Technical Details}
\label{Se:TechnicalDetails}

We first review regular infinite-tree expressions and how we use them as a representation
of probabilistic programs (\cref{Se:RegularHyperPathExpressions}).
We then proceed to the novel
parts
of \framework{}.
We devise a new family of algebraic structures, called $\omega$PMAs,
for specifying abstract domains of program analyses (\cref{Se:OmegaPMAsAndAlgebraicExpressions}).
We then describe a method for solving a linear equation system over an $\omega$PMA (\cref{Se:SolvingLinearEquations}).
Next, we develop \framework{}'s mechanism
for solving interprocedural equation systems with regular-infinite-tree-expression right-hand sides (\cref{Se:AnalysisNewton}).
We end this section with a discussion for the soundness of \framework{} (\cref{Se:Soundness}).
\iflong
Proofs are included in \Cref{Se:Appendix:CFHG,Se:Appendix:TheoryOfHyperPaths,Se:Appendix:Proofs,Se:Appendix:Soundness}.
\else
The technical report~\cite{Techreport} includes the proofs of this section.
\fi

\subsection{Review: Regular Infinite-tree Expressions}
\label{Se:RegularHyperPathExpressions}

Let $\m{RegExp}^\infty(\calF,\calK)$ denote the set of \emph{regular infinite-tree expressions}
over a ranked alphabet $\calF$ and a set $\calK$ of free variables.
A \emph{ranked alphabet} is a pair $\tuple{\calF,\mi{Arity}}$ where $\calF$
is a nonempty set and $\mi{Arity}$ is a mapping from $\calF$ to $\bbN$.
The \emph{arity} of a symbol $f \in \calF$ is $\mi{Arity}(f)$.
The set of symbols of arity $n$ is denoted by $\calF_n$.
We use parentheses and commas to specify symbols with their arity,
e.g., $f(,)$ specifies a binary symbol $f$.
The set $\calK$ includes \emph{free variables} that can be used as leaves (i.e., symbols with arity zero) and would not be bound
by any $\mu$-binder.
We define $\m{RegExp}^\infty(\calF,\calK)$ to be inductively generated from the rules listed below, where $\uplus$ stands for disjoint union.
%
We use $\infty$ to emphasize that the expressions denote possibly-\emph{infinite} trees.
%
Each regular infinite-tree expression corresponds to exactly one possibly-infinite tree.
\begin{mathpar}\small
  \inferrule*[right=(Leaf)]
  { a \in \calF_0 \uplus \calK }
  { a \in \m{RegExp}^\infty(\calF,\calK) }
  \and
  \inferrule*[right=(Node)]
  { f \in \calF_n \\ n > 0 \\ \Forall{i=1,\cdots,n} E_i \in \m{RegExp}^\infty(\calF,\calK) }
  { f(E_1,\cdots,E_n) \in \m{RegExp}^\infty(\calF,\calK) }
  \and
  \inferrule*[right=(Concat)]
  { E_1 \in \m{RegExp}^\infty(\calF,\calK\uplus\{Z\}) \\ E_2 \in \m{RegExp}^\infty(\calF,\calK) }
  { E_1 \dplus_Z E_2 \in \m{RegExp}^\infty(\calF,\calK) }
  \and
  \inferrule*[right=(Mu)]
  { E \in \m{RegExp}^\infty(\calF,\calK\uplus\{Z\}) }
  { \mu Z.\, E \in \m{RegExp}^\infty(\calF,\calK) }
\end{mathpar}
We include the formal development of possibly-infinite trees in
\iflong
\Cref{Se:Appendix:TheoryOfHyperPaths},
\else
the technical report~\cite{Techreport},
\fi
following a well-studied theory of infinite trees~\cite{JTCS:Courcelle83,kn:Courcelle90}.

A free variable $Z \in \calK$ intuitively represents a \emph{substitution} placeholder,
for which we introduce the following notation.
$E_1 \dplus_Z E_2$ represents the tree obtained by substituting the tree encoded by $E_2$
at each occurrence of a leaf with variable $Z$ in the tree encoded by $E_1$.
A \emph{$\mu$-binder} $\mu Z.\, E$ represents the tree obtained by $(E \dplus_Z E \dplus_Z \cdots \dplus_Z E \dplus_Z \cdots)$,
i.e., self-substituting $E$ for the bound variable $Z$ an infinite number of times.

To represent probabilistic programs,
we define the following ranked alphabet, where $\calA$ is the set of \emph{data actions} and $\calL$ is the set of \emph{logical conditions}, and
$\mi{cond}[\cdot]$, $\mi{prob}[\cdot]$, and $\mi{ndet}$ stands for conditional-choice,
probabilistic-choice, and nondeterministic-choice commands, respectively:
\begin{align*}
\calF & \defeq \{\varepsilon \} \cup \{ \mi{seq}[\m{act}]() \mid \m{act} \in \calA \}
\cup \{ \mi{cond}[\varphi](,) \mid \varphi \in \calL \} \cup \{ \mi{prob}[p](,) \mid p \in [0,1] \} \cup \{ \mi{ndet}(,) \} \\
 & \quad 
\cup \{ \mi{call}[X_i]() \mid \text{$X_i$ is a procedure} \}.
\end{align*}
%
Data actions are atomic program actions that manipulate data, e.g., assignment and random sampling,
excluding control-flow actions like $\kw{break}$ and $\kw{continue}$.
%
For example, for probabilistic Boolean programs, we can define data actions and logical conditions as follows, where $x$ is a program variable and $p \in [0,1]$, and the data action ``$x \sim \cn{Ber}(p)$'' draws a random value from a Bernoulli distribution with mean $p$, i.e., the action
  assigns \kw{true} to $x$ with probability $p$, and otherwise assigns \kw{false}:
\begin{align*}
  \m{act} \in \calA & \Coloneqq x \coloneqq \varphi \mid x \sim \cn{Ber}(p) \mid \kw{skip} &
  \varphi \in \calL & \Coloneqq x \mid \kw{true} \mid \kw{false} \mid \neg \varphi \mid \varphi_1 \wedge \varphi_2 \mid \varphi_1 \vee \varphi_2
\end{align*}

We do not discuss how to transform an unstructured probabilistic program to its
corresponding regular infinite-tree expression here but include it in
\iflong
\Cref{Se:Appendix:CFHG},
\else
the technical report~\cite{Techreport},
\fi
because it is a standard routine
that applies Beki{\'c}'s theorem~\cite{kn:Bekic84} to the program's control-flow equations,
which can be extracted from its control-flow hyper-graph~\cite{PLDI:WHR18}---as
we mentioned in \cref{Se:Introduction}---a suitable representation of probabilistic programs
with multiple confluence operations.

\begin{example}
\begin{figure}
\centering
\begin{subfigure}{0.45\textwidth}
\centering
\begin{small}
\begin{pseudo*}
  \kw{while} \kw{true} \kw{do} \\+
    $b_1 \sim \cn{Ber}(0.5)$; \\
    \kw{while} $b_1 \vee b_2$ \kw{do} \\+
      \kw{if} \kw{prob}(0.1) \kw{then} \kw{return} \kw{fi}; \\
      \kw{if} $b_1$ \kw{then} $b_1 \sim \cn{Ber}(0.2)$ \\
      \kw{else} $b_2 \sim \cn{Ber}(0.8)$ \\
      \kw{fi} \\-
    \kw{od} \\-
  \kw{od}
\end{pseudo*}
\end{small}
\caption{An example program}
\end{subfigure}
\begin{subfigure}{0.45\textwidth}
\centering
\resizebox{0.7\textwidth}{!}{%
\begin{tikzpicture}[op/.style={rectangle,draw,inner sep=3pt},node distance=0.2cm,font=\small]
  \node[op] (outer) {$\mu Z_1$};
  \node (a) [below=of outer] {$\mi{seq}[b_1{\sim}\cn{Ber}(0.5)]$};
  \node[op] (inner) [below=of a] {$\mu Z_2$};
  \node (b) [below=of inner] {$\mi{cond}[b_1{\vee}b_2]$};
  \node (c) [below left=0.2cm and -1.2cm of b] {$\mi{prob}[0.1]$};
  \node (b_to_outer) [below right=0.2cm and -0.5cm of b] {$Z_1$};
  \node (d) [below right=0.2cm and -1.2cm of c] {$\mi{cond}[b_1]$};
  \node (c_to_exit) [below left=0.2cm and -0.2cm of c] {$\varepsilon$};
  \node (e) [below left=0.2cm and -1.4cm of d] {$\mi{seq}[b_1{\sim}\cn{Ber}(0.2)]$};
  \node (f) [below right=0.2cm and 0cm of d] {$\mi{seq}[b_2{\sim}\cn{Ber}(0.8)]$};
  \node (e_to_inner) [below=of e] {$Z_2$};
  \node (f_to_inner) [below=of f] {$Z_2$};
  
  \draw (outer.south) edge[->] (a.north);
  \draw (a.south) edge[->] (inner.north);
  \draw (inner.south) edge[->] (b.north);
  \draw (b.south) edge[->] (c.north);
  \draw (b.south) ++ (5pt,0) edge[->] (b_to_outer.north west);
  \draw (c.south) ++ (-5pt,0) edge[->] (c_to_exit.north east);
  \draw (c.south) edge[->] (d.north);
  \draw (d.south) ++ (-3pt,0) edge[->] (e.north);
  \draw (d.south) ++ (3pt,0) edge[->] (f.north);
  \draw (e.south) edge[->] (e_to_inner.north);
  \draw (f.south) edge[->] (f_to_inner.north);
\end{tikzpicture}}
\vspace{-0.5em}
\caption{Regular infinite-tree expression AST}
\end{subfigure}
\begin{subfigure}{\textwidth}
\centering
\begin{small}
\[
\mu Z_1.\,
\mi{seq}[\m{act}_1]\Biggl(
  \mu Z_2.\,
  \mi{cond}[b_1{\vee}b_2]\biggl(
    \mi{prob}[0.1]\Bigl(
      \varepsilon,
      \mi{cond}[b_1]\bigl(
        \mi{seq}[\m{act}_2](
          Z_2
        ),
        \mi{seq}[\m{act}_3](
          Z_2
        )
      \bigr)
    \Bigr),
    Z_1
  \biggr)
\Biggr)
\]
\end{small}
\vspace{-1em}
\caption{Regular infinite-tree expression (where $\m{act}_1 \defeq b_1{\sim}\cn{Ber}(0.5)$, $\m{act}_2 \defeq b_1{\sim}\cn{Ber}(0.2)$, $\m{act}_3 \defeq b_2{\sim}\cn{Ber}(0.8)$)}
\end{subfigure}
\caption{An example program with nested loops and unstructured control-flow.}
\label{Fi:ExampleProgramWithLoopsAndUnstructuredControlFlow}
\end{figure}
\cref{Fi:ExampleProgramWithLoopsAndUnstructuredControlFlow}(a) presents
an example probabilistic Boolean program with nested loops and unstructured
control-flow.
\cref{Fi:ExampleProgramWithLoopsAndUnstructuredControlFlow}(c) shows the program's
corresponding regular infinite-tree expression and
\cref{Fi:ExampleProgramWithLoopsAndUnstructuredControlFlow}(b) illustrates the AST of
the expression.
Branching (e.g., $\mi{cond}[b_1 \vee b_2]$ and $\mi{prob}[0.1]$) statements become
internal nodes with multiple children in the AST.
Every loop leads to a $\mu$-binder in the AST: $\mu Z_1$ and $\mu Z_2$ stand for the
outer and inner loops, respectively.
\end{example}


\subsection{$\omega$PMAs and Algebraic Regular Infinite-tree Expressions}
\label{Se:OmegaPMAsAndAlgebraicExpressions}

Throughout our development in the rest of \cref{Se:TechnicalDetails}, we fix a probabilistic program $P$ as an equation system $P \defeq \{X_i = E_{X_i}\}_{i=1}^n$,
where each $E_{X_i}$ is in $\m{RegExp}^\infty(\calF,\emptyset)$ and represents
different program analyses depending on how the data actions in the tree are interpreted.
Toward this end, we
devise $\omega$-continuous pre-Markov algebras, which can be
thought as a refinement of both pre-Markov algebras~\cite{PLDI:WHR18} and $\omega$-continuous semirings, to
specify program analyses.
\cref{Se:ThisWorkAConfluenceCentricAnalysisFramework}
described an $\omega$PMA over real numbers.
Recall from \cref{Se:ThisWorkAConfluenceCentricAnalysisFramework} that an $\omega$PMA has three binary confluence operators, $\gcho{\varphi}$ (parameterized by a logical condition $\varphi$), $\pcho{p}$ (parameterized by $p \in [0,1]$), and $\dashcup$, which correspond to conditional-choice, probabilistic-choice, and nondeterministic-choice, respectively.

\begin{definition}\label{De:OmegaContinuousPreMarkovAlgebra}
  An \emph{$\omega$-continuous pre-Markov algebra} ($\omega$PMA) over a set $\calL$ of logical conditions is an 8-tuple
  $\calM = \tuple{M, \oplus_M, \otimes_M, \gcho{\varphi}_M, \pcho{p}_M, \dashcup_M, \azero_M, \aone_M }$,
  where $\tuple{M, \oplus_M, \otimes_M, \azero_M, \aone_M}$ forms an $\omega$-continuous semiring (so it admits
  a partial order $\aord_M \defeq \{(a,b) \in M \times M \mid \Exists{d} a \oplus_M d = b \}$ and we can define
  $a \ominus_M b$ for $a \sqsupseteq_M b$ to be any element $d$ that satisfies $b \oplus_M d = a$);
  the operations $\gcho{\varphi}_M$, $\pcho{p}_M$, $\dashcup_M$ are $\omega$-continuous with respect to $\aord_M$;
  and for all $a,b,c,d \in M$ and $\varphi \in \calL, p \in [0,1]$, it holds that
  $(a \oplus_M b) \gcho{\varphi}_M (c \oplus_M d) = (a \gcho{\varphi}_M c) \oplus_M (b \gcho{\varphi}_M d)$
  and
  $(a \oplus_M b) \pcho{p}_M (c \oplus_M d) = (a \pcho{p}_M c) \oplus_M (b \pcho{p}_M d)$.

  An \emph{$\omega$PMA interpretation} is a pair $\scrM = \tuple{\calM,\interp{\cdot}^\scrM}$,
  where $\interp{\cdot}^\scrM$ maps data actions to $\calM$.
\end{definition}

\begin{example}\label{Exa:BayesianInferenceDomain}
We consider the \emph{Bayesian-inference} analysis of probabilistic Boolean programs
in the style of prior work by \citet{FSE:CRN13}; that is, we are interested in computing the
probability distribution on variable valuations at the end of the program, conditioned on termination of the program.
The inferred probability distribution is called the \emph{posterior distribution}.\footnote{Bayesian inference
usually uses \emph{conditioning} to specify posterior distributions. For the analysis in \cref{Exa:BayesianInferenceDomain} (and also \cref{Se:BayesianInferenceAnalysis}),
we can encode a conditioning failure as a non-terminating loop, i.e., $\kw{condition}(\varphi) \defeq \kw{while}~\neg\varphi~\kw{do}~\kw{skip}~\kw{od}$.}
For a Boolean program with a set $\m{Var}$ of variables, its probabilistic semantics can be modeled
as a transformer from states ($2^{\m{Var}}$) to state-distributions ($2^{\m{Var}} \to [0,1]$).
Therefore, we can represent distribution transformers as matrices of type $2^{\m{Var}} \times 2^{\m{Var}} \to [0,1]$.
\citet{PLDI:WHR18} formulated a PMA on such matrices
with the following definitions of operations and constants,
and here we extend it to be an $\omega$PMA
$\calB = \tuple{2^{\m{Var}} \times 2^{\m{Var}} \to \bbR_{\ge 0} \cup \{\infty\}, {\oplus}, {\otimes}, {\gcho{\varphi}}, {\pcho{p}}, {\dashcup}, \azero, \aone}$:
{\begin{align*}
  \azero & \defeq \mathbf{0}~\text{(all-zero matrix)}, & \mathbf{A} \otimes \mathbf{B} & \defeq \mathbf{A} \cdot \mathbf{B}~\text{(matrix multiplication)}, &
  \mathbf{A} \pcho{p} \mathbf{B} & \defeq p \cdot \mathbf{A} + (1-p) \cdot \mathbf{B}, \\
  \aone & \defeq \mathbf{I} ~\text{(identity matrix)}, &  \mathbf{A} \dashcup \mathbf{B} & \defeq  \min(\mathbf{A}, \mathbf{B})~\text{(pointwise minimum)}, & \mathbf{A} \gcho{\varphi} \mathbf{B} & \defeq \mathbf{\Gamma}_\varphi \cdot \mathbf{A} + \mathbf{\Gamma}_{\neg\varphi} \cdot \mathbf{B},
\end{align*}}%
where $\mathbf{\Gamma}_\varphi$ is a diagonal matrix such that $\mathbf{\Gamma}_\varphi(s,s) = 1$ if $\varphi$ holds in $s$, and otherwise $0$, for any state $s \in 2^{\m{Var}}$.
We further define $\mathbf{A} \oplus \mathbf{B} \defeq \mathbf{A} + \mathbf{B}$ to be pointwise addition.
Note that $\oplus$ does not correspond to any confluence operation;
instead,
it is used to define
a partial order $\mathbf{A} \aord_B \mathbf{B} \defeq \mathbf{A} \le \mathbf{B}$ (pointwise comparison)
and a subtraction operation $\mathbf{A} \ominus_B \mathbf{B} \defeq \mathbf{A} - \mathbf{B}$ (pointwise subtraction).
\citeauthor{PLDI:WHR18} defined a mapping $\interp{\cdot}^\scrB$ from data actions to the algebra $\calB$; for example,
let $\m{Var} \defeq \{b_1,b_2\}$, we can encode sampling statements as matrices as follows, where $p,q\in[0,1]$:
\begin{align*}
  \interp{ b_1{\sim}\cn{Ber}(p) }^\scrB & \defeq \enskip \begin{pNiceMatrix}[first-row,first-col,small]
    & \m{TT} & \m{TF} & \m{FT} & \m{FF} \\
    \m{TT} & p & 0 & 1-p & 0 \\
    \m{TF} & 0 & p & 0 & 1-p \\
    \m{FT} & p & 0 & 1-p & 0 \\
    \m{FF} & 0 & p & 0 & 1-p
  \end{pNiceMatrix}\,, &
  \interp{ b_2{\sim}\cn{Ber}(q) }^\scrB & \defeq \enskip \begin{pNiceMatrix}[first-row,first-col,small]
    & \m{TT} & \m{TF} & \m{FT} & \m{FF} \\
    \m{TT} & q & 1-q  & 0 & 0 \\
    \m{TF} & q & 1-q & 0 & 0 \\
    \m{FT} & 0 & 0 & q & 1-q \\
    \m{FF} & 0 & 0 & q & 1-q
 \end{pNiceMatrix}\,.
\end{align*}
We denote by ``$\sigma_1\sigma_2$'' the program state $\{b_1 \mapsto \sigma_1,b_2 \mapsto \sigma_2\}$, where
$\sigma_1,\sigma_2 \in \{\m{T}, \m{F}\}$, and rows and columns of the matrices stand for pre- and post-states,
respectively.  
\end{example}

Given a regular infinite-tree expression $E \in \m{RegExp}^\infty(\calF,\calK)$, an $\omega$PMA interpretation
$\scrM = \tuple{\calM,\interp{\cdot}^\scrM}$,
a \emph{valuation} $\gamma : \calK \to \calM$ that interprets bound variables,
and a \emph{procedure-summary vector} $\vec{\nu} \in \calM^n$,
we define the interpretation of $E$ under $\gamma$ and $\vec{\nu}$, denoted by $\scrM_\gamma\interp{E}(\vec{\nu})$,
as follows.
\begin{small}
\begin{alignat*}{2}
  \scrM_\gamma\interp{\varepsilon}(\vec{\nu}) \defeq \aone_M, & \quad \scrM_\gamma\interp{Z}(\vec{\nu})  \defeq \gamma(Z), &
  \scrM_\gamma\interp{\mi{cond}[\varphi](E_1,E_2)}(\vec{\nu}) & \defeq \scrM_\gamma\interp{E_1}(\vec{\nu}) \gcho{\varphi}_M \scrM_\gamma\interp{E_2}(\vec{\nu}), \\
  \scrM_\gamma\interp{\mi{seq}[\m{act}](E)}(\vec{\nu}) & \defeq \interp{\m{act}}^\scrM \otimes_M \scrM_\gamma\interp{E}(\vec{\nu}), &
  \scrM_\gamma\interp{\mi{prob}[p](E_1,E_2)}(\vec{\nu}) & \defeq \scrM_\gamma\interp{E_1}(\vec{\nu}) \pcho{p}_M \scrM_\gamma\interp{E_2}(\vec{\nu}), \\
  \scrM_\gamma\interp{\mi{call}[X_i](E)}(\vec{\nu}) & \defeq \nu_i \otimes_M \scrM_\gamma\interp{E}(\vec{\nu}), &
  \scrM_\gamma\interp{\mi{ndet}(E_1,E_2)}(\vec{\nu}) & \defeq \scrM_\gamma\interp{E_1}(\vec{\nu}) \dashcup_M \scrM_\gamma\interp{E_2}(\vec{\nu}), \\
  \scrM_\gamma\interp{\mu Z.\, E}(\vec{\nu}) & \defeq \lfp_{\azero_M}^{\aord_M} \lambda \theta.\; \scrM_{\gamma[Z \mapsto \theta]}\interp{E}(\vec{\nu}),  &
  \scrM_\gamma\interp{E_1 \dplus_Z E_2}(\vec{\nu}) & \defeq \scrM_{\gamma[Z \mapsto \scrM_\gamma\interp{E_2}(\vec{\nu})]}\interp{E_1}(\vec{\nu}).
\end{alignat*}
\end{small}%
For $E_1 \dplus_Z E_2$,
the variable $Z$ does \emph{not} appear in $E_2$,
but it might appear in $E_1$.
Thus, we can interpret $E_2$ under the valuation $\gamma$,
and to interpret $E_1$, we extend $\gamma$ by mapping $Z$ to the interpretation of $E_2$.
For $\mu Z.\, E$,
the bound variable $Z$ might appear in $E$.
In this case, we should find a value $\theta$ such that $\theta = \scrM_{\gamma[Z \mapsto \theta]}\interp{E}$.
Because $\calM$ is an $\omega$-continuous semiring (thus it admits an $\omega$-cpo),
we define the interpretation for $\mu$-binders using \emph{least fixed-points}.
For a procedure call $\mi{call}[X_i](E)$, we call $E$ the \emph{continuation} of this call to $X_i$,
and use the procedure summary $\nu_i$ to interpret the call to $X_i$.

For the probabilistic program $P \defeq \{ X_i = E_{X_i} \}_{i=1}^n$ represented by an
interprocedural equation system with regular-infinite-tree right-hand sides, we define its interpretation via a least fixed-point:
\[
\scrM\interp{P} \defeq \lfp_{\vec{\azero}_M}^{\aord_M} \lambda\vec{\theta}.\; \tuple{ \scrM_{\{\}}\interp{E_{X_1}}(\vec{\theta}), \cdots, \scrM_{\{\}}\interp{E_{X_n}}(\vec{\theta}) }.
\]

For easier manipulation of regular infinite-tree expressions (as well as the differentials we will develop in \cref{Se:AnalysisNewton}), we define \emph{algebraic regular infinite-tree
expressions} over the following ranked alphabet, where we essentially change data actions to algebra elements:
\begin{align*}
\calF^{\alpha} & \defeq \{c \mid c \in \calM \} \cup \{ \mi{seq}[c]() \mid c \in \calM \}
\cup \{ \mi{cond}[\varphi](,) \mid \varphi \in \calL \}
\cup \{ \mi{prob}[p](,) \mid p \in [0,1] \} \\
& 
\quad \cup \{ \mi{ndet}(,) \}
\cup \{ \mi{call}[X_i]() \mid i =1,\cdots,n \},
\end{align*}
Similar to the interpretation of regular infinite-tree expressions,
we can define the interpretation $\calM_\gamma\interp{E}(\vec{\nu})$ of an algebraic expression $E$ under $\gamma$ and $\vec{\nu}$
by induction on the structure of $E$.

\begin{lemma}\label{Lem:EquivalentToAlgebraicExpressions}
  For any expression $E \in \m{RegExp}^\infty(\calF,\calK)$, there exists an expression $E' \in \m{RegExp}^\infty(\calF^\alpha,\calK)$ ,
  such that for any $\gamma : \calK \to \calM$ and $\vec{\nu} \in \calM^n$, it holds that $\scrM_\gamma\interp{E}(\vec{\nu}) = \calM_\gamma\interp{E'}(\vec{\nu})$.
\end{lemma}

In the rest of this section, we assume that the right-hand sides of equations are all algebraic.

\subsection{Solving Linear Equations}
\label{Se:SolvingLinearEquations}

In this section, we develop a mechanism to solve \emph{linear} equation systems.
This material describes
box 2, oval 2, box 3, and oval 3 in \cref{Fi:Pipeline}.
Recall that in \cref{Se:ThisWorkAConfluenceCentricAnalysisFramework}, we informally said that a simple tree expression is
\emph{linear} if every root-to-leaf path contains \emph{at most} one $\mi{call}[\cdot]$ node.
We now formalize the idea by introducing \emph{linear} algebraic regular regular-tree expressions over the following ranked alphabet:
\begin{align*}
\calF^\alpha_\m{lin} & \defeq \{c \mid c \in \calM \} \cup \{ \mi{seq}[c]() \mid c \in \calM \}
\cup \{ \mi{cond}[\varphi](,) \mid \varphi \in \calL \}
\cup \{ \mi{prob}[p](,) \mid p \in [0,1] \} \\
& 
\quad \cup \{ \mi{ndet}(,), {\oplus}(,), {\ominus}(,) \}
\cup \{ \mi{call_\m{lin}}[Y_i; c] \mid i =1,\cdots,n, c \in \calM \}.
\end{align*}
In $\calF^\alpha_\m{lin}$,
we add two binary symbol $\oplus,\ominus$ that correspond to $\oplus_M, \ominus_M$, respectively,
and restrict general procedure calls to linear ones
of the form
$\mi{call}_\m{lin}[Y_i;c]$ by
restricting the continuation of a procedure call to $Y_i$ to be a value $c \in \calM$.\footnote{We strategically change the procedure names from $X_i$'s to $Y_i$'s to signify
that later in \cref{Se:AnalysisNewton},
linear expressions are derived by differentials of non-linear expressions---and thus a linearized
equation system should be understood as being defined over a different set of variables from the
original interprocedural equation system.}
We interpret linear algebraic regular infinite-tree expressions in $\m{RegExp}^{\infty}(\calF^\alpha_\m{lin},\calK)$ in the same way
that general algebraic regular infinite-tree expressions are interpreted,
except that we remove the rule for $\mi{call}[X_i](E)$, and add $\calM_\gamma\interp{\mi{call}_\m{lin}[Y_i; c]}(\vec{\nu}) \defeq \nu_i \otimes_M c$,
\begin{small}
\begin{align*}
  \calM_\gamma\interp{\oplus(E_1, E_2)}(\vec{\nu}) & \defeq \calM_\gamma\interp{E_1}(\vec{\nu}) \oplus_M \calM_\gamma\interp{E_2}(\vec{\nu}), &
  \calM_\gamma\interp{\ominus(E_1 , E_2)} & \defeq \calM_\gamma\interp{E_1}(\vec{\nu}) \ominus_M \calM_\gamma\interp{E_2}(\vec{\nu}).
\end{align*}
\end{small}

We can now state the problem of solving a system of linear equations (box 2 in \cref{Fi:Pipeline}):
\[\begin{array}{|l|}
\hline
\text{Solve $\{Y_i = E_{Y_i}\}_{i=1}^n$, where each $E_{Y_i}$ is in $\m{RegExp}^\infty(\calF^\alpha_\m{lin}, \emptyset)$; i.e., find the least $\vec{\theta} \in \calM^n$ that satisfies} \\
\text{$\calM_{\{\}}\interp{E_{Y_i}}(\vec{\theta}) = \theta_i$ for each $i$.} \\
\hline
\end{array}\]
\framework{} relies on an analysis-supplied strategy to solve linear equations.
To simplify matters, we work with equations whose right-hand sides are normalized to be
simple tree expressions (\cref{Se:ThisWorkAConfluenceCentricAnalysisFramework}).
(The normalization transformation is similar to the standard program transformation that replaces
loops with tail-recursion---see below---but is performed on the equation system.)
Let $\m{RegExp}^\m{cf}$ denote $\mu$-free regular infinite-tree expressions, i.e., simple tree expressions.
%
The superscript $\m{cf}$ indicates that the expressions are \emph{closure-free}, in the sense
that a $\mu$-binder introduces a closure.
%
The problem statement now becomes:
\[\begin{array}{|l|}
\hline
\text{Solve $\{Y_i = E_{Y_i}\}_{i=1}^n$, where $E_{Y_i} \in \m{RegExp}^\infty(\calF^\alpha_\m{lin}, \emptyset)$, given a strategy that solves equation systems} \\
\text{of the form $\{Y_i = F_{Y_i}\}_{i=1}^n \uplus \{ Z = F_Z \}_{Z \in \calZ}$, whose right-hand sides are in $\m{RegExp}^\m{cf}(\calF^\alpha_\m{lin}, \calZ)$.} \\
\hline
\end{array}\]
%
This formulation reflects that there are \emph{two} sources of
\emph{multiplicative factors}, i.e., expressions whose interpretation
involves the $\otimes$ operation and the value of a variable ($Y_i$ for some $i$ or $Z \in \calZ$).
The first source involves $Y_i$ for some $i$ and must take the form $\mi{call}_\m{lin}[Y_i;c]$,
whose interpretation is $\nu_i \otimes c$, where $\vec{\nu}$ is a procedure-summary vector.
The second source involves $Z \in \calZ$ and must take the form $\mi{seq}[c](Z)$,
whose interpretation is $c \otimes \gamma(Z)$, where $\gamma$ is a valuation.

The normalization transformation corresponds to oval 2 in \cref{Fi:Pipeline};
it removes $\mu$-binders by introducing
a new equation 
with
a fresh left-hand side
variable
for each $\mu$-binder.
\cref{Fi:ConstraintExtraction} presents some rules for a routine $\Gamma \vdash E \Downarrow F \mid \Theta$
that extracts from a linear regular infinite-tree expression $E$ a $\mu$-free expression $F$ and an equation system $\Theta$,
where $\Gamma$ is a mapping on variables.
For a $\mu$-binder $\mu Z.\, E$, the rule \textsc{(Mu)} introduces a fresh variable $Z'$ as
the result and adds a new equation $Z' = F$, where $F$ is the extracted $\mu$-free expression from $E$
with all occurrences of $Z$ mapped to $Z'$.

\begin{figure}
\centering  
\begin{mathpar}\footnotesize
  \inferrule*[right=(Call-Lin)]
  { c \in \calM
  }
  { \Gamma \vdash \mi{call}_\m{lin}[Y_i;c] \Downarrow \mi{call}_\m{lin}[Y_i;c] \mid \emptyset }
  \and
  \inferrule*[right=(Free-Var)]
  { }
  { \Gamma \vdash Z \Downarrow \Gamma(Z) \mid \emptyset }
  \and
  \inferrule*[right=(Concat)]
  { \Gamma[Z \mapsto Z] \vdash E_1 \Downarrow F_1 \mid \Theta_1 \\
    \Gamma \vdash E_2 \Downarrow F_2 \mid \Theta_2
  }
  { \Gamma \vdash E_1 \dplus_{Z} E_2 \Downarrow F_1 \dplus_{Z} F_2 \mid \Theta_1 \uplus \Theta_2 }
  \and
  \inferrule*[right=(Mu)]
  { Z'~\m{fresh} \\ \Gamma[ Z \mapsto Z' ] \vdash E \Downarrow F \mid \Theta }
  { \Gamma \vdash \mu Z.\, E \Downarrow Z' \mid \Theta \uplus \{ Z' = F \}  }
\end{mathpar}
\caption{Selected rules for extracting a $\mu$-free linear regular infinite-tree expression and an associated equation system from a linear regular infinite-tree expression.}
\label{Fi:ConstraintExtraction}
\end{figure}

\begin{example}\label{Exa:ExtractionAlgebraicEquations}
  Consider the following algebraic regular infinite-tree expression $E$:
\[
\mu Z_1.\,
\mi{seq}[\mathbf{C}_1]\Biggl(
  \mu Z_2.\,
  \mi{cond}[b_1{\vee}b_2]\biggl(
    \mi{prob}[0.1]\Bigl(
      \aone,
      \mi{cond}[b_1]\bigl(
        \mi{seq}[\mathbf{C}_2](
          Z_2
        ),
        \mi{seq}[\mathbf{C}_3](
          Z_2
        )
      \bigr)
    \Bigr),
    Z_1
  \biggr)
\Biggr),
\]
  where $\mathbf{C}_1$, $\mathbf{C}_2$, $\mathbf{C}_3$ are three algebra constants.
  We now show how the normalization transformation works by establishing $ \{\} \vdash E \Downarrow F \mid \Theta$ for some $F$ and $\Theta$.
  After applying the rule \textsc{(Mu)} twice, we derive the following normalization for a $\mu$-free subexpression:
  \[
  \begin{split}
  & \{ Z_1 \mapsto Z_1', Z_2 \mapsto Z_2'\} \vdash 
  \mi{cond}[b_1{\vee}b_2](
    \mi{prob}[0.1](
      \aone,
      \mi{cond}[b_1](
        \mi{seq}[\mathbf{C}_2](
          Z_2
        ),
        \mi{seq}[\mathbf{C}_3](
          Z_2
        )
      )
    ),
    Z_1
  ) \Downarrow \\
  & \qquad \mi{cond}[b_1{\vee}b_2]( \mi{prob}[0.1](\aone, \mi{cond}[b_1](\mi{seq}[\mathbf{C}_2](Z_2'), \mi{seq}[\mathbf{C}_3](Z_2'))), Z_1') \mid \{ \} .
  \end{split}
  \]
  We then add a new equation to $\Theta$ for each application of the rule \textsc{(Mu)}:
  \begin{equation}\label{Eq:ExampleExtraction}
  \begin{split}
  \{ \} \vdash E \Downarrow Z_1' \mid \{ & Z_1' = \mi{seq}[\mathbf{C}_1]( Z_2'),  \\
   & Z_2' = \mi{cond}[b_1{\vee}b_2](\mi{prob}[0.1](\aone_B, \mi{cond}[b_1](\mi{seq}[\mathbf{C}_2](Z_2'), \mi{seq}[\mathbf{C}_3](Z_2') ) ), Z_1') \} .
  \end{split}
  \end{equation}
\end{example}

We now formalize \emph{linear-recursion-solving strategies} (oval 3 in \cref{Fi:Pipeline}).

\begin{definition}\label{De:LinearRecursionSolvingStrategy}
  An algorithm $\m{solve}$ (parameterized by $\calK$) is said to be a \emph{linear-recursion-solving strategy} for an $\omega$PMA
  $\calM$,
  if for any equation system $\{ Y_i = E_{Y_i}\}_{i=1}^n \uplus \{ Z = E_Z \}_{Z \in \calZ}$ where $\calZ$ is a finite set of variables that is
  disjoint from $(\calF_\m{lin}^\alpha)_0 \cup \calK$
  (where $(\calF_\m{lin}^\alpha)_0$ represents the set of zero-arity symbols in $\calF_\m{lin}^\alpha$)
  and right-hand sides of the equations are in $\m{RegExp}^\m{cf}(\calF^\alpha_\m{lin}, \calK \uplus \calZ)$,
  and any valuation $\gamma : \calK \to \calM$,
  the routine $\m{solve}_\calK( \{ Y_i = E_{Y_i} \}_{i=1}^n , \{Z=E_Z\}_{Z \in \calZ }, \gamma)$ returns 
  the \emph{least} mapping $\iota : \calZ \to \calM$ and vector $\vec{\nu}$, such that
  $\iota(Z) = \calM_{\gamma \uplus \iota}\interp{ E_Z }(\vec{\nu})$ for each $Z \in \calZ$
  and $\nu_i = \calM_{\gamma \uplus \iota}\interp{E_{Y_i}}(\vec{\nu})$ for each $i=1,\cdots,n$,
  with respect to the order $\aord_M$.
\end{definition}

Our method for solving an equation system $\{Y_i = E_{Y_i}\}_{i=1}^n$
whose right-hand sides are in $\m{RegExp}^\infty(\calF_\m{lin}^\alpha,\calK)$---given a valuation $\gamma : \calK \to \calM$---proceeds in three steps:
\begin{enumerate}[nosep,leftmargin=*]
  \item For each equation $Y_i = E_{Y_i}$, derive $\{ Z \mapsto Z \}_{Z \in \calK} \vdash E_{Y_i} \Downarrow F_{Y_i} \mid \Theta_i$.
  \item Run $\m{solve}_\calK( \{Y_i = F_{Y_i} \}_{i=1}^n, \biguplus_{i=1}^n \Theta_i, \gamma)$ to obtain a solution mapping $\iota : \dom(\biguplus_{i=1}^n \Theta_i) \to \calM$ and a procedure-summary vector $\vec{\nu} \in \calM^n$.
  \item Compute $\calM_{\gamma \uplus \iota}\interp{F_{Y_i}}(\vec{\nu})$ as the solution for $Y_i$ for each $i=1,\cdots,n$.
\end{enumerate}

\begin{theorem}\label{The:LinearRecursionSolvingSound}
  Given a procedure $\m{solve}$ that meets the requirements of \cref{De:LinearRecursionSolvingStrategy},
  the method presented above computes $\lfp_{\vec{\azero}_M}^{\aord_M} \lambda \vec{\theta}.\; \tuple{\calM_\gamma\interp{E_{Y_1}}(\vec{\theta}), \cdots, \calM_\gamma\interp{E_{Y_n}}(\vec{\theta})}$.
\end{theorem}

\begin{example}\label{Exa:PMAForBooleanPrograms}
  Consider the $\omega$PMA $\calB$ from \cref{Exa:BayesianInferenceDomain} that encodes Bayesian-inference analysis
  and the regular infinite-tree expression $E$ from \cref{Exa:ExtractionAlgebraicEquations} and an equation system $\{Y = E\}$.
  Here $\mathbf{C}_1,\mathbf{C}_2,\mathbf{C}_3$ are three matrices encoding three data actions.
  In step (1), we apply 
  the rules in \cref{Fi:ConstraintExtraction} to $E$
  and then derive the normalization judgment shown in \cref{Eq:ExampleExtraction}.
  In step (2), we run the analysis-supplied strategy, i.e., $\m{solve}_{\{\}}( { Y = Z_1' }, \Theta, \{\})$, to obtain a solution mapping $\iota : \{ Z_1', Z_2' \} \to \calB$ and a procedure-summary $\nu \in \calB$ for $Y$.
  We will describe two such strategies for Bayesian-inference analysis in \cref{Se:BayesianInferenceAnalysis}.
  To demonstrate the idea, under the interpretation with respect to $\calB$,
  the extracted equations are \emph{linear} matrix equations in $Z_1'$ and $Z_2'$, thus
  a strategy needs to solve those linear equations:
{\begin{equation*}
\begin{split}
  Z_1' & = \mathbf{C}_1 \cdot Z_{2}', \\
  Z_2' & = \mathbf{\Gamma}_{b_1{\vee}b_2} \cdot ( 0.1 \cdot \mathbf{I} + 0.9 \cdot ( \mathbf{\Gamma}_{b_1} \cdot \mathbf{C}_2 \cdot Z_{2}' + \mathbf{\Gamma}_{\neg b_1} \cdot \mathbf{C}_3 \cdot Z_{2}' ) ) + \mathbf{\Gamma}_{\neg(b_1{\vee}b_2)} \cdot Z_{1}' .
\end{split}
\end{equation*}}
  Finally, in step (3), we compute $\calB_{\iota}\interp{Z_1'}(\nu) = \iota(Z_1')$ as the solution for $Y$.
\end{example}

\subsection{Solving Equations via Newton's Method}
\label{Se:AnalysisNewton}

In \cref{Se:SolvingLinearEquations}, we discussed solving linear equations.
We now consider the problem of solving equation systems in the general case by extending Newtonian Program Analysis~\cite{ICALP:EKL08,JACM:EKL10}, i.e.,
we develop a differentiation process for regular infinite-tree expressions (oval 1 in \cref{Fi:Pipeline}).

Let $f \in \m{RegExp}^\infty(\calF^\alpha,\calK)$. For brevity, in this section, we use $f_\gamma(\vec{\nu})$ to denote
$\calM_\gamma\interp{f}(\vec{\nu})$ for a valuation $\gamma : \calK \to \calM$ and a procedure-summary vector $\vec{\nu}$.
We also write $f(\vec{X})$ from time to time to indicate that the procedure-call symbols in the alphabet are $\{\mi{call}[X_i]() \mid i=1,\cdots,n\}$.

\begin{definition}\label{De:RegularHyperPathDifferential}
  Let $f(\vec{X}) \in \m{RegExp}^\infty(\calF^\alpha,\calK)$.
  The \emph{differential} of $f(\vec{X})$ with respect to $X_j$ at $\vec{\nu} \in \calM^n$ under a valuation $\gamma : \calK \to \calM$,
  denoted by $\calD_{X_j} f_\gamma|_{\vec{\nu}}(\vec{Y})$, is defined as follows:
  {\small\[\begin{array}{|l|} \hline
  \calD_{X_j} f_\gamma|_{\vec{\nu}}(\vec{Y}) \defeq
  \begin{dcases*}
    \mi{seq}[\nu_k]( \calD_{X_j} g_\gamma|_{\vec{\nu}}( \vec{Y} ) ) & if $f = \mi{call}[X_k](g)$, $k \neq j$ \\[-3pt]
    \oplus( \mi{call}_\m{lin}[Y_j;g_\gamma(\vec{\nu})] , \mi{seq}[\nu_j]( \calD_{X_j} g_\gamma|_{\vec{\nu}}(\vec{Y}) ) ) & if $f = \mi{call}[X_j](g)$ \\[-3pt]
    \ominus( \mi{ndet}( {\oplus }(g_\gamma(\vec{\nu}), \calD_{X_j} g_\gamma|_{\vec{\nu}}(\vec{Y}) ), {\oplus}(h_\gamma(\vec{\nu})  , \calD_{X_j} h_\gamma|_{\vec{\nu}}(\vec{Y}) ) ), f_\gamma(\vec{\nu}) ) & if $f = \mi{ndet}(g,h)$ \\[-3pt]
    Z & if $f = Z \in \calK$ \\[-3pt]
    \calD_{X_j} g_{\gamma[Z \mapsto h_\gamma(\vec{\nu})]}|_{\vec{\nu}}(\vec{Y}) \dplus_{Z} \calD_{X_j} h_\gamma|_{\vec{\nu}}(\vec{Y}) & if $f = g \dplus_{Z} h$ \\[-3pt]
    \mu Z.\, \calD_{X_j} g_{\gamma[Z \mapsto f_\gamma(\vec{\nu})]}|_{\vec{\nu}}(\vec{Y})  & if $f = \mu Z.\, g$
  \end{dcases*} \\ \hline \end{array}
  \]}%
  The complete definition is included in
  \iflong
  \Cref{Se:Appendix:Proofs}.
  \else
  the technical report~\cite{Techreport}.
  \fi
  Let $\vec{f}$ be a vector of algebraic regular infinite-tree expressions.
  The \emph{multivariate differential} of $\vec{f}$ at $\vec{\nu}$ under $\gamma$,
  denoted by $\calD \vec{f}_\gamma|_{\vec{\nu}}(\vec{Y})$, is defined as
  \[
  \calD \vec{f}_\gamma|_{\vec{\nu}}(\vec{Y}) \defeq \left\langle 
    \calD_{X_1} (f_1)_\gamma |_{\vec{\nu}}(\vec{Y}) \oplus \cdots \oplus \calD_{X_n} (f_1)_\gamma |_{\vec{\nu}}(\vec{Y}),
    \cdots,
    \calD_{X_1} (f_n)_\gamma |_{\vec{\nu}}(\vec{Y}) \oplus \cdots \oplus \calD_{X_n} (f_n)_\gamma |_{\vec{\nu}}(\vec{Y})
   \right\rangle,
  \]
  where we use $\oplus$ as an infix operator.
  We use $\calD (f_i)_\gamma |_{\vec{\nu}}(\vec{Y})$ to denote the $i^\textit{th}$ component of $\calD \vec{f}_\gamma|_{\vec{\nu}}(\vec{Y})$.
\end{definition}

\begin{wrapfigure}{r}{0.34\textwidth}
\vspace{-0.5em}
\centering
\begin{tabular}{@{\hspace{0ex}}c@{\hspace{.2ex}}c@{\hspace{0ex}}}
\begin{minipage}{0.18\textwidth}
\centering
\begin{tikzpicture}[op/.style={rectangle,draw,inner sep=3pt},node distance=0.2cm,font=\footnotesize]
  \node[op] (inner) {$\mu Z$};
  \node (b) [below=of inner] {$\mi{cond}[\varphi]$};
  \node (c) [below left=0.2cm and -0.8cm of b] {$\mi{call}[X]$};
  \node (b_to_outer) [below right=0.2cm and -0.2cm of b] {$\aone$};
  \node (e) [below=of c] {$Z$};
  
  \draw (inner.south) edge[->] (b.north);
  \draw (b.south) ++ (-2pt,0) edge[->] (c.north);
  \draw (b.south) ++ (2pt,0) edge[->] (b_to_outer.north west);
  \draw (c.south) edge[->] (e.north);
\end{tikzpicture}
\vspace{-1em}
\end{minipage}
&
\begin{minipage}{0.18\textwidth}
\centering
\begin{tikzpicture}[op/.style={rectangle,draw,inner sep=3pt},node distance=0.2cm,font=\footnotesize]
  \node (a) {$\mi{cond}[\varphi]$};
  \node (b) [below left=0.2cm and -0.8cm of a] {$\mi{call}[X]$};
  \node (c) [below right=0.2cm and -0.2cm of a] {$\aone$};
  \node (d) [below=of b] {$\mi{cond}[\varphi]$};
  \node (e) [below left=0.2cm and -0.8cm of d] {$\mi{call}[X]$};
  \node (f) [below right=0.2cm and -0.2cm of d] {$\aone$};
  \node (g) [below=of e] {$\cdots$};

  \draw (a.south) ++ (-2pt,0) edge[->] (b.north);
  \draw (a.south) ++ (2pt,0) edge[->] (c.north);
  \draw (b.south) edge[->] (d.north);
  \draw (d.south) ++ (-2pt,0) edge[->] (e.north);
  \draw (d.south) ++ (2pt,0) edge[->] (f.north);
  \draw (e.south) edge[->] (g.north);
\end{tikzpicture}
\vspace{-1em}
\end{minipage}
\\
\\
{\small (a) $f(X)$} & {\small (b) Unfold $f(X)$}
\\
\\
\begin{minipage}{0.18\textwidth}
\centering
\begin{tikzpicture}[op/.style={rectangle,draw,inner sep=3pt},node distance=0.2cm,font=\footnotesize]
  \node[op] (inner) {$\mu Z$};
  \node (b) [below=of inner] {$\mi{cond}[\varphi]$};
  \node (plus) [below left=0.2cm and -0.4cm of b] {$\oplus$};
  \node (c) [below left=0.2cm and -0.3cm of plus] {$\mi{call}[Y]$};
  \node (b_to_outer) [below right=0.2cm and -0.2cm of b] {$\azero$};
  \node (d) [below=of c] {$f(\nu)$};
  \node (e) [below right=0.2cm and -0.2cm of plus] {$\mi{seq}[\nu]$};
  \node (g) [below=of e] {$Z$};
  
  \draw (inner.south) edge[->] (b.north);
  \draw (b.south) ++ (-2pt,0) edge[->] (plus.north);
  \draw (plus.south) ++ (-2pt,0) edge[->] (c.north);
  \draw (b.south) ++ (2pt,0) edge[->] (b_to_outer.north west);
  \draw (c.south) edge[->] (d.north);
  \draw (plus.south) ++ (2pt,0) edge[->] (e.north);
  \draw (e.south) edge[->] (g.north);
\end{tikzpicture}
\vspace{-1em}
\end{minipage}
&
\begin{minipage}{0.18\textwidth}
\centering
\begin{tikzpicture}[op/.style={rectangle,draw,inner sep=3pt},node distance=0.15cm,font=\footnotesize]
  \node (a) {$\mi{cond}[\varphi]$};
  \node (aa) [below left=0.15cm and -0.4cm of a] {$\oplus$};
  \node (b) [below left=0.15cm and -0.3cm of aa] {$\mi{call}[Y]$};
  \node (b_alter) [below right=0.15cm and -0.2cm of aa] {$\mi{seq}[\nu]$};
  \node (c) [below right=0.15cm and -0.2cm of a] {$\azero$};
  \node (b_cont) [below=of b] {$f(\nu)$};
  \node (d) [below=of b_alter] {$\mi{cond}[\varphi]$};
  \node (e) [below left=0.15cm and -0.4cm of d] {$\oplus$};
  \node (f) [below right=0.15cm and -0.2cm of d] {$\azero$};
  \node (g) [below left=0.15cm and -0.3cm of e] {$\mi{call}[Y]$};
  \node (h) [below=of g] {$f(\nu)$};
  \node (i) [below right=0.15cm and -0.2cm of e] {$\mi{seq}[\nu]$};
  \node (k) [below=of i] {$\cdots$};

  \draw (a.south) ++ (-2pt,0) edge[->] (aa.north);
  \draw (aa.south) ++ (-2pt,0) edge[->] (b.north);
  \draw (a.south) ++ (2pt,0) edge[->] (c.north);
  \draw (b.south) edge[->] (b_cont.north);
  \draw (d.south) ++ (-2pt,0) edge[->] (e.north);
  \draw (d.south) ++ (2pt,0) edge[->] (f.north);
  \draw (aa.south) ++ (2pt,0) edge[->] (b_alter.north);
  \draw (b_alter.south) edge[->] (d.north);
  \draw (e.south) ++ (-2pt,0) edge[->] (g.north);
  \draw (g.south) edge[->] (h.north);
  \draw (e.south) ++ (2pt,0) edge[->] (i.north);
  \draw (i.south) edge[->] (k.north);
\end{tikzpicture}
\vspace{-1em}
\end{minipage}
\\
\\
{\small (c) $\calD f|_\nu(Y)$} & {\small (d) Unfold $\calD f|_\nu(Y)$}
\end{tabular}
\caption{An example of differentiating $\mu$-binders.
}
\label{Fi:DifferentialOfClosureExamples}
\end{wrapfigure}

We have shown in \cref{Se:ThisWorkAConfluenceCentricAnalysisFramework}
how to derive most of those differentiation rules.
One missing non-trivial case is the rule for $\mu$-binders.
Consider an example $X = f(X) \defeq \mu Z.\,  \mi{cond}[\varphi]\bigl(\mi{call}[X]( Z), \aone\bigr) $
for some condition $\varphi$,
whose right-hand side encodes the program ``\kw{while} $\varphi$ \kw{do} $X()$ \kw{od}.''
\cref{Fi:DifferentialOfClosureExamples}(a) shows the AST of $f(X)$ and \cref{Fi:DifferentialOfClosureExamples}(b)
demonstrates its corresponding infinite tree.
Note that
\cref{Fi:DifferentialOfClosureExamples}(b) shows that $f(X)$  in our example is definitely not linear:
the leftmost rooted path contains more than one $\mi{call}[X]$ node.
We develop the differential of $\mu$-binders as if we are differentiating their corresponding
infinite trees.
One would obtain another infinite tree similar to the one depicted in \cref{Fi:DifferentialOfClosureExamples}(d),
where one can easily ``fold'' this infinite tree and get
the regular infinite-tree expression in \cref{Fi:DifferentialOfClosureExamples}(c).

\begin{example}
 Consider the equation $X = f(X)$ for a single-procedure program with
 $f(X) \defeq \mu Z.\,  \mi{cond}[\varphi]( \mi{call}[X]( \mi{seq}[c_1](Z))) , c_2 )$, which represents the program ``\kw{while} $\varphi$ \kw{do} $X$(); $c_1$ \kw{od}; $c_2$,''
 for some $\varphi \in \calL$ and $c_1,c_2 \in \calM$.
 Suppose we want
 the differential
 of $f$ at $\nu$.
 By \cref{De:RegularHyperPathDifferential}, we
 derive that $\calD f|_{\nu}^{\{\}}(Y)$ is
 \begin{wrapped}{equation*}
  \mu Z.\,  \mi{cond}[\varphi](  {\oplus}( \mi{call}_\m{lin}[Y; c_1 \otimes_M f_{\{\}}(\nu) ] , \mi{seq}[\nu] ( \mi{seq}[c_1]( Z ) ) ) , \azero_M ) ,
 \end{wrapped}
 which
 is a \emph{linear} regular infinite-tree expression in $\m{RegExp}^\infty(\calF_\m{lin}^\alpha,\emptyset)$.
\end{example}

\begin{lemma}\label{Lem:DiffIsLinear}
  If $\vec{f}$ is a vector of regular infinite-tree expressions in $\m{RegExp}^\infty(\calF^\alpha,\calK)$,
  it holds that
  $\calD \vec{f}_\gamma|_{\vec{\nu}}(\vec{Y})$ is a vector of regular infinite-tree expressions in $\m{RegExp}^\infty(\calF_\m{lin}^\alpha, \calK)$.
\end{lemma}

We now describe Newton's method for $\omega$PMAs by defining a Newton sequence of Newton approximants.
In this sense, \framework{} is still \emph{iterative} when solving general recursion.

\begin{definition}\label{De:PMANewtonSequence}
  Let $\vec{f}$ be a vector of regular infinite-tree expressions in $\m{RegExp}^\infty(\calF^\alpha,\emptyset)$.
  %
  A \emph{Newton sequence} $\{\vec{\nu}^{(i)}\}_{i \in \bbN}$ is given by
  $\vec{\nu}^{(0)} \defeq \vec{f}_{\{\}}(\vec{\azero}_M)$ and
  $ \vec{\nu}^{(i+1)}  \defeq \vec{\nu}^{(i)} \oplus_M \vec{\Delta}^{(i)}$ for $i \ge 0$,
  where $\vec{\Delta}^{(i)}$ is the least solution of the following linearized equation system
  \[
  \vec{Y} = {\oplus}(  \vec{f}_{\{\}}(\vec{\nu}^{(i)})  \ominus_M \vec{\nu}^{(i)}  , \calD \vec{f}_{\{\}}|_{\vec{\nu}^{(i)}} (\vec{Y}) ),
  \]
  using the method for solving linear recursion developed in \cref{Se:SolvingLinearEquations,The:LinearRecursionSolvingSound}.
\end{definition}

We show that the Newton sequence converges to the least fixed-point of $\vec{f}$.

\begin{theorem}\label{The:NewtonConvergence}
  Let $\vec{f}$ be a vector of regular infinite-tree expressions in $\m{RegExp}^\infty(\calF^\alpha,\emptyset)$.
  Then the Newton sequence is monotonically increasing, and it converges to the least fixed-point
  as least as fast as the Kleene sequence, i.e., for all $i \in \bbN$, we have
  \[
  \vec{\kappa}^{(i)} \aord_M \vec{\nu}^{(i)} \aord_M \vec{f}_{\{\}}(\vec{\nu}^{(i)}) \aord_M \vec{\nu}^{(i+1)} \aord_M \lfp_{\vec{\azero}_M}^{\aord_M} \vec{f}_{\{\}} = \textstyle\bigsqcup^{\uparrow}_{j \in \bbN} \vec{\kappa}^{(j)},
  \]
  where the Kleene sequence is defined as $\vec{\kappa}^{(j)} \defeq \vec{f}_{\{\}}^j ( \vec{\azero}_M)$ for $j \in \bbN$.
  Specifically, the Newton sequence and the Kleene sequence converge to the same least fixed-point.
\end{theorem}

Finally, we connect the Newton sequence back to our original
equation-solving problem.

\begin{corollary}\label{Cor:NewtonConvergence}
  Consider the equation system $\{X_i = E_{X_i}\}_{i=1}^n$ where $E_{X_i} \in \m{RegExp}^\infty(\calF^\alpha, \emptyset)$ for each $i$.
  Let $\vec{f} \defeq \tuple{E_{X_i}}_{i=1,\cdots,n}$.
  Then by \cref{The:NewtonConvergence}, the Newton sequence converges to the least $\vec{\theta} \in \calM^n$ that
  satisfies $\vec{f}_{\{\}}(\vec{\theta}) = \vec{\theta}$, i.e., $\calM_{\{\}}\interp{E_{X_i}}(\vec{\theta}) = \theta_i$ for each $i$.
\end{corollary}


\subsection{Soundness}
\label{Se:Soundness}

In this section, we sketch our approach for proving the soundness of \framework{}, but include the details
in
\iflong
\Cref{Se:Appendix:Soundness}.
\else
the technical report~\cite{Techreport}.
\fi
We use a recently proposed family of algebraic structures,
namely \emph{Markov algebras} (MAs)~\cite{ENTCS:WHR19}, to specify concrete semantics of probabilistic programs.
We then introduce \emph{soundness relations} between MAs and $\omega$PMAs and show that these
relations guarantee the soundness of program analyses in \framework{}.
%

\paragraph{Semantic foundations}
%
%
We use the interpretation of regular infinite-tree expressions over MAs to define concrete semantics
of probabilistic programs.
%
%
  A \emph{Markov algebra} (MA) $\calM = \tuple{M,\aord_M, \otimes_M,\gcho{\varphi}_M, \pcho{p}_M, \dashcup_M, \azero_M, \aone_M}$ over a set $\calL$ of logical conditions is 
  almost the same as an $\omega$PMA,
  except that it has an explicit partial order $\aord_M$ (instead of the $\oplus_M$ operation), and satisfies a different set of algebraic laws:
  $\tuple{M,\aord_M}$ forms a directed-complete partial order with $\azero_M$ as its least element;
  $\tuple{M,\otimes_M,\aone_M}$ forms a monoid;
  $\dashcup_M$ is idempotent, commutative, and associative, and for all $a,b \in M$ and $\varphi \in \calL, p\in[0,1]$
  it holds that $a \gcho{\varphi}_M b, a \pcho{p}_M b \le_M a \dashcup_M b$,
  where $\le_M$ is the semilattice ordering induced by $\dashcup_M$ (i.e., $a \le_M b$ if $a \dashcup_M b = b$);
  and $\otimes_M,\gcho{\varphi}_M, \pcho{p}_M, \dashcup_M$ are Scott-continuous. 
%
  An \emph{MA interpretation} is then a pair $\scrM = \tuple{\calM,\interp{\cdot}^\scrM}$
  where $\interp{\cdot}^\scrM$ maps data actions
  to $\calM$.
%
%
Given a regular infinite-tree expression $E \in \m{RegExp}^\infty(\calF,\calK)$,
an MA interpretation $\scrM = \tuple{\calM,\interp{\cdot}^\scrM}$,
a valuation $\gamma {:} \calK {\to} \calM$,
and a procedure-summary vector $\vec{\nu} \in \calM^n$,
the interpretation of $E$ under $\gamma$ and $\vec{\nu}$, denoted by $\scrM_\gamma\interp{E}(\vec{\nu})$,
can be defined in the same way as the $\omega$PMA interpretations presented in \cref{Se:OmegaPMAsAndAlgebraicExpressions}.
%

%

\begin{example}\label{Exa:MAForBooleanPrograms}
  Consider probabilistic Boolean programs with a set $\m{Var}$ of program variables.
  Let $S \defeq 2^{\m{Var}}$ denote the state space of such programs.
  We can formulate a demonic denotational semantics~\cite{book:MM05} by defining an MA interpretation
  $\scrC = \tuple{\calC, \interp{\cdot}^\scrC}$ where $\calC$ is an MA on
  $S \to \wp(S \to [0,1])$, i.e., mappings from states to sets of state distributions.
  %
\end{example}

\begin{remark}
It is admissible to use $\omega$PMAs for both the abstract and the concrete semantics.
On the other hand, the concrete semantics does not have to be an $\omega$PMA, which requires
some additional properties (e.g., the $\oplus$ operation and the semiring structure).
Intuitively, Markov algebras are used to characterize a small but reasonable set of properties for
the concrete semantics.
\end{remark}  

\paragraph{Soundness relations}
We adapt abstract interpretation~\cite{POPL:CC77} to justify \framework{}
by establishing a relation between the concrete and abstract semantics.
We express the relation via a \emph{soundness relation}, which is a binary relation between an MA interpretation
and an $\omega$PMA interpretation that is preserved by the algebraic operations.
Intuitively, a (concrete, abstract) pair in the relation should be read as ``the concrete element is
approximated by the abstract element.''


\begin{example}\label{Exa:SoundnessRelationForBayesianInference}
  Let $\scrC = \tuple{\calC,\interp{\cdot}^\scrC }$ be an MA interpretation where $\calC$ is an MA for
  probabilistic Boolean programs, which is introduced in \cref{Exa:MAForBooleanPrograms}.
  Let $\scrB = \tuple{\calB, \interp{\cdot}^\scrB }$ be an $\omega$PMA interpretation where $\calB$ is the $\omega$PMA over matrices
  described in \cref{Exa:BayesianInferenceDomain}.
  We define the following approximation relation to indicate that the program analysis reasons about lower bounds:
  \[
  r \Vdash \mathbf{A} \iff \Forall{s \in S} \Forall{\mu \in r(s)} \Forall{s' \in S} \mu(s') \ge \mathbf{A}(s,s').
  \]
  %
\end{example}

The correctness of an interpretation
is then justified by the following soundness theorem, which followed by induction on the structure of
regular infinite-tree expressions.

\begin{theorem}
\label{The:InterSoundness}
  Let $\scrC = \tuple{\calC, \interp{\cdot}^\scrC}$ be an MA interpretation.
  Let $\scrM = \tuple{\calM, \interp{\cdot}^\scrM}$ be an $\omega$PMA interpretation.
  Let ${\Vdash} \subseteq \calC \times \calM$ be a soundness relation.
  Then for any regular infinite-tree expression $E \in \m{RegExp}^\infty(\calF,\calK)$,
  $\gamma : \calK \to \calC$, $\gamma^\sharp : \calK \to \calM$ such that
  $\gamma(Z) \Vdash \gamma^\sharp(Z)$ for all $Z \in \calK$,
  $\vec{\nu} : \calC^n$, and $\vec{\nu}^\sharp : \calM^n$ such that $\nu_i \Vdash \nu^\sharp_i$ for $i=1,\cdots,n$,
  we have $\scrC_{\gamma}\interp{E}(\vec{\nu}) \Vdash \scrM_{\gamma^\sharp}\interp{E}(\vec{\nu}^\sharp)$.
\end{theorem}

%


\section{Case Studies}
\label{Se:CaseStudies}

We implemented a prototype of \framework{} in OCaml; the core framework consists of around 1,000 lines of code.
%
We conducted four case studies in which we instantiated \framework{} to perform different analyses of probabilistic programs.
This section presents preliminary experimental results, and shows the generality of \framework{} by sketching instantiations with existing abstract domains for
(i) finding probability distributions of variable valuations (\cref{Se:BayesianInferenceAnalysis} \&
\iflong
\Cref{Se:BayesianInferenceADD})
\else
the technical report~\cite{Techreport})
\fi
and (ii) upper bounds on moments
\iflong
(\Cref{Se:HigherMomentAnalysis}),
\else
(the technical report~\cite{Techreport})
\fi
as well as
with new abstract domains for (iii) finding linear expectation invariants (\cref{Se:ExpectationInvariantAnalysis}) and (iv) non-linear expectation invariants
(\cref{Se:ExpectationRecurrenceAnalysis}).
For some of those case studies, the subproblems generated from Newton's method are reduced to
\emph{linear programming} (LP);
in the
implementations of those instantiations,
we used the off-the-shelf LP solver CoinOr CLP~\cite{misc:CLP22}.
%
The case studies were performed on a machine with an Apple M2 processor (3.50 GHz, 8 cores) and 24 GB of RAM under macOS Sonoma 14.3.1.

\subsection{Bayesian-inference Analysis}
\label{Se:BayesianInferenceAnalysis}

\paragraph{Matrix-based domain}
In \cref{Exa:BayesianInferenceDomain}, we introduced the Bayesian-inference analysis of probabilistic Boolean programs and its corresponding $\omega$PMA $\calB$.
The analysis is in the style of prior work~\cite{FSE:CRN13,PLDI:WHR18,OOPSLA:HBM20,JACM:EY15};
it computes the probability distribution of variable valuations conditioned on program termination.
When the analyzed program contains nondeterminism, the analysis computes a lower bound
on the probability distribution.
In \cref{Exa:SoundnessRelationForBayesianInference}, we formulated a soundness relation between
a concrete interpretation $\scrC$ for probabilistic Boolean programs \Omit{(\cref{Exa:MAForBooleanPrograms})} and the
abstract interpretation $\scrB$.

We now describe a linear-recursion-solving strategy $\m{solve}$ for this analysis.
Consider $\m{solve}_\calK( \{Y_i = E_{Y_i} \}_{i=1}^n, \{ Z = E_Z\}_{Z \in \calZ}, \gamma)$
where $E_{Y_i}$'s and $E_Z$'s are linear regular infinite-tree expressions in $\m{RegExp}^\m{cf}(\calF^\alpha_\m{lin},\calK \uplus \calZ)$,
and $\gamma : \calK \to \calB$ is a valuation.
We then reduce the equation-solving problem to \emph{linear programming} (LP).
We transform a min-linear matrix equation system to an LP instance by introducing a fresh variable for each $\min$-subexpression,
and thus obtain a \emph{normalized} matrix equation system $\vec{Z} = \vec{f}(\vec{Z})$ in which each $f_i(\vec{Z})$ is either
a linear matrix expression over $\vec{Z}$ or the minimum of two variables $\min(Z_j,Z_k)$, for some $j,k$.
The least solution of $\vec{Z} = \vec{f}(\vec{Z})$ can then be obtained via the following LP instance:
{\small\begin{alignat*}{6}
  \textbf{maximize} \enskip & \textstyle\sum_i Z_i, &
  \enskip \textbf{s.t.} \enskip & Z_i = \mathbf{C}_i + \textstyle\sum_{j,k} \mathbf{A}_{i,j,k} \cdot Z_j \cdot \mathbf{B}_{i,j,k} & & \enskip \text{for $i$ such that $f_i(\vec{Z}) = \mathbf{C}_i + \textstyle \sum_{j,k} \mathbf{A}_{i,j,k} \cdot Z_j \cdot \mathbf{B}_{i,j,k}$}, \\[-3pt]
  & & & Z_i \le Z_j, Z_i \le Z_k & & \enskip \text{for $i$ such that $f_i(\vec{Z}) = \min(Z_j,Z_k)$}, \\[-3pt]
  & & & Z_i \ge \mathbf{0}   & & \enskip \text{for each $i$},
\end{alignat*}}%
where each $Z_i$ can be represented by a matrix of numeric variables in an LP solver.

Our implementation of this domain consists of about 400 lines of code.
We evaluated the performance of \framework{} against Kleene iteration on 100 randomly generated probabilistic programs,
each of which has two Boolean-valued variables and consists of 100 procedures, each of which takes one of the following forms:
(i) $\mi{prob}[p]( \mi{seq}[x{\coloneqq}a]( \mi{call}[X_i]( \varepsilon ) ), \mi{seq}[y{\coloneqq}b]( \mi{call}[X_j]( \varepsilon ) ) )$
for some probability $p$, variables $x,y$, Boolean constants $a,b$, and procedures $X_i,X_j$;
(ii) $\mi{prob}[p]( \mi{cond}[x]( \mi{call}[X_i]( \varepsilon ), \mi{call}[X_j]( \varepsilon ) ), \varepsilon )$
for some probability $p$, variable $x$, and procedures $X_i,X_j$; or
(iii) $\mi{call}[X_i]( \mi{call}[X_j] (\varepsilon))$ for some procedures $X_i,X_j$.
On the benchmark suite of 100 programs, \framework{} and Kleene iteration derive the same
analysis results (up to some negligible floating-point error). Such an experimental result
is unsurprising because \cref{The:NewtonConvergence} ensures that Newton iteration and
Kleene iteration converge to the same fixed-point.
On this benchmark suite,
the average number of Newton rounds performed by \framework{} is 9.15,
whereas the average number of Kleene rounds is 3,070.42.
\cref{Fi:RuntimePlotBayesianInference} (left) presents a scatter plot that compares the \emph{running time} (in seconds) of Kleene iteration ($x$-axis)
against \framework{} ($y$-axis), where a point in the lower-right triangle indicates than \framework{} has
better performance on that benchmark.
Overall, \framework{} outperforms Kleene iteration:
the geometric mean of $x/y$ is 4.08.
Although a Newton round is usually
more expensive than a Kleene round,
\framework{} greatly reduces the number of rounds (3,070.42 $\to$ 9.15), resulting in an average speedup of 4.08x.

\begin{figure}
\centering
\begin{subfigure}[b]{0.4\textwidth}
\centering
\vspace{-0.1em}
\includegraphics[width=\textwidth]{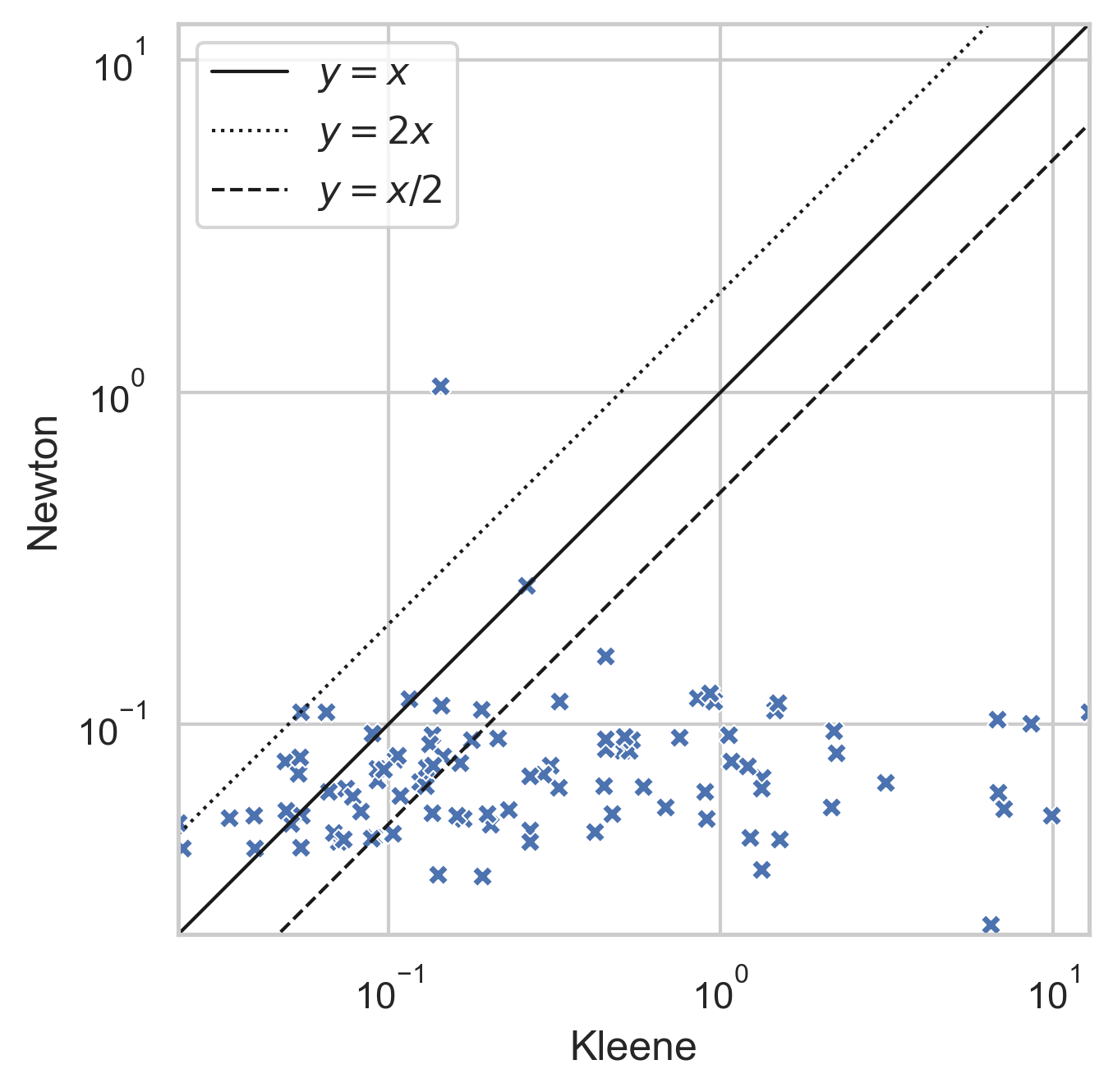}
\vspace{-0.7em}
\end{subfigure}
\hfil
\begin{subfigure}[b]{0.4\textwidth}  
\centering
\vspace{-0.1em}
\includegraphics[width=\textwidth]{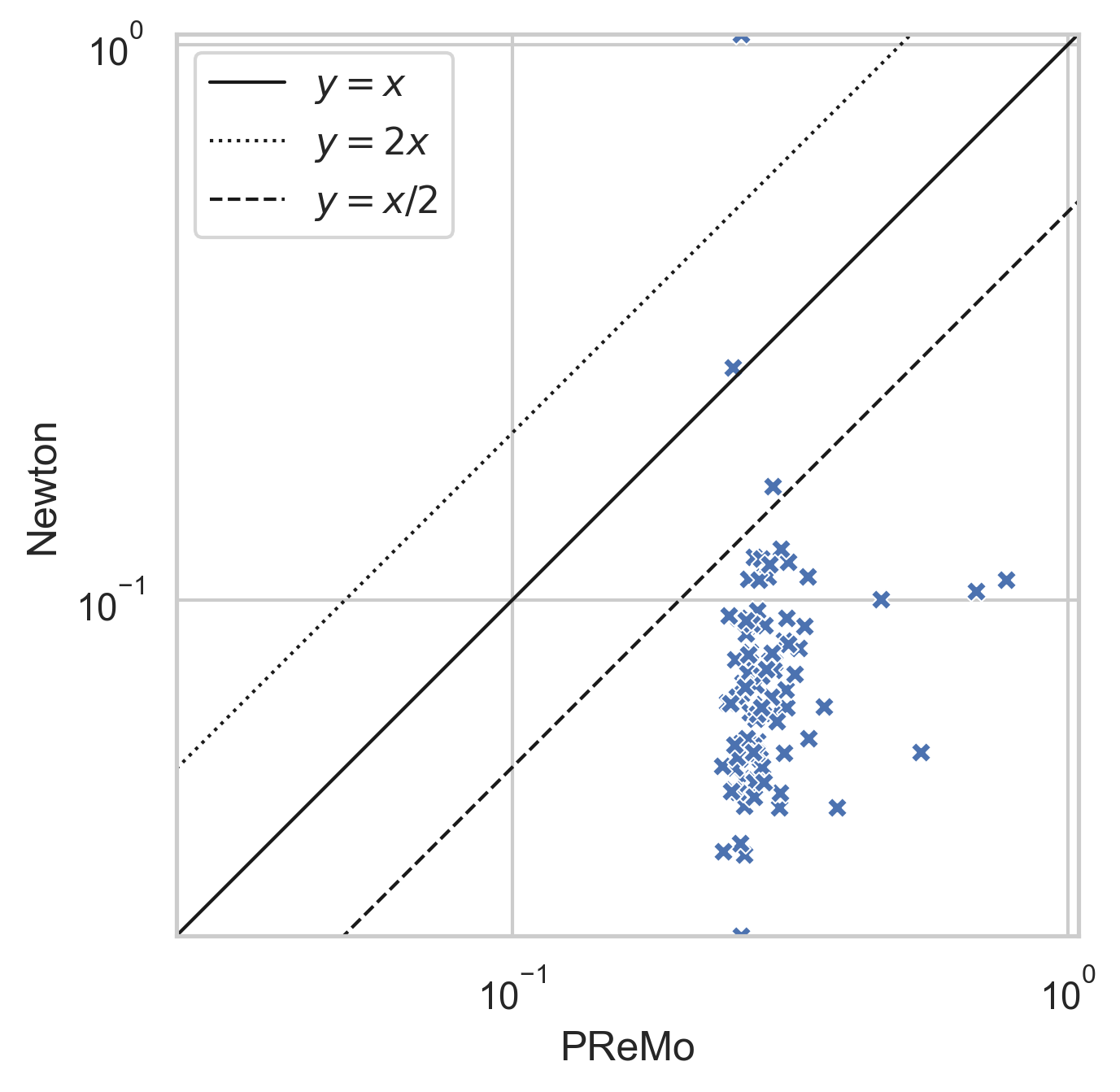}
\vspace{-0.7em}
\end{subfigure}
\caption{Log-log scatter plots
of running times (in seconds) of \framework{} (``Newton'') versus Kleene-iteration (left plot) and \framework{} (``Newton'') versus PReMo (right plot),
where the solid line, dotted line, and dashed line indicate equal performance, 2x slowdown, and 2x speedup, respectively.
}
\label{Fi:RuntimePlotBayesianInference}
\end{figure}

We also compared our prototype implementation with PReMo~\cite{TACAS:WE07}, a tool for
analyzing probabilistic recursive models such as recursive MDPs.
PReMo also applies a generalized form of Newton's method, specialized to polynomial
equation systems over the real semiring.
\cref{Fi:RuntimePlotBayesianInference} (right) shows a scatter plot that compares the
running time of PReMo ($x$-axis) against \framework{} ($y$-axis).
From the statistics, \framework{} outperforms PReMo
(the geometric mean of $x/y$ is 4.07),
but it should be noted that PReMo is implemented in Java, which might incur more overhead than our implementation.

\paragraph{Algebraic-decision-diagram-based domain} 
%
In
\iflong
\Cref{Se:BayesianInferenceADD},
\else
the technical report~\cite{Techreport},
\fi
we develop an instantiation that uses \emph{Algebraic Decision Diagrams}~\cite{FMSD:BFG97} to 
represent distribution transformers and
carry out Bayesian-inference analysis of probabilistic programs \emph{without} nondeterminism.
We evaluated the performance of this domain on the benchmark suite mentioned earlier.
\framework{} achieves an overall 1.54x speedup compared with Kleene iteration
(computed as the geometric mean).
%
We adopted the technique of \citet{POPL:RTP16} for solving linear equation systems using tensor products.
%
%
In their experiments, NPA-TP was
slower than chaotic iteration on a BDD-based predicate-abstraction domain, whereas our result is much better than that.
We think it is the quantitative nature of probabilistic programs that unleashes the potential of Newton's method.

\subsection{Expectation-invariant Analysis}
\label{Se:ExpectationInvariantAnalysis}

For a probabilistic program $P$ with a set $\m{Var}$ of numeric program variables,
$\bbE[\calE_2] \bowtie \calE_1$ is called an \emph{expectation invariant}~\cite{PLDI:WHR18,SAS:KMM10,SAS:CS14},
where (i)
$\calE_1$ (resp. $\calE_2$) is an
arithmetic expression over $\m{Var}$ (resp., $\m{Var}'$, primed copies of the variables in $\m{Var}$),
(ii)
$\bowtie$ is a comparison operator (e.g., $\le$),
and
(iii)
the expectation of $\calE_2$ in the final valuation (i.e., the primed copies in $\m{Var}'$ represent the values of variables at the termination of $P$) is related to the value of $\calE_1$ in the initial valuation of $P$ by $\bowtie$.
We studied two kinds of expectation-invariant analysis:
in this section, we consider linear expectation invariants of the form $\bbE[x'] \le \calE$, where $x$ is a variable, for programs where all variables are nonnegative.
In \cref{Se:ExpectationRecurrenceAnalysis}, we will describe an analysis for obtaining non-linear expectation invariants, but considers a more restricted setting in which randomness does \emph{not} influence the flow of control in a program.

We consider programs with nonnegative program variables
and we want to derive upper bounds on the expected value $\bbE[x']$ of every program variable $x$.
In particular, we consider expectation invariants of the form $\bbE[x'] \le \calE$ where $\calE$ is a linear expression.
Expected-value analysis is useful in many applications; for example, a recent work has used expected-value analysis
as a subroutine for performing expected-cost analysis~\cite{OOPSLA:AMS20}.
Suppose that a program uses $k$ variables and they are represented by a vector $\vec{x}$.
We can encode upper bounds on expected values of $\vec{x}$ as a nonnegative matrix $\mathbf{T}$
of the form $\begin{pNiceArray}[small]{c|c} p & \vec{b} \\ \hline \vec{0}^\m{T} & \mathbf{M}  \end{pNiceArray}$,
where $p$ is a number in the unit interval $[0,1]$,
$\vec{b}$ is a
$k$-dimensional vector, and $\mathbf{M}$ is a $k$-by-$k$ matrix, in the sense that
$[ \bbE[ 1]  \mid \bbE[\vec{x}'] ] \le [1 \mid \vec{x}] \cdot \mathbf{T}$.
Let $V$ denote the space of such matrices.

We now define an $\omega$PMA $\calV = \tuple{V, \oplus_V, \otimes_V, \gcho{\varphi}_V, \pcho{p}_V, \dashcup_V, \azero_V, \aone_V}$ as follows.
Note that we ignore $\gcho{\varphi}_V$ and $\dashcup_V$ temporarily because they are more involved;
we will revisit them later.
{\small\begin{align*}
  \azero_V & \defeq \begin{pNiceArray}[small]{c|c} 0 & \vec{0} \\ \hline \vec{0}^\m{T} & \mathbf{0} \end{pNiceArray}, &
  \aone_V & \defeq \begin{pNiceArray}[small]{c|c} 1 & \vec{0} \\ \hline \vec{0}^\m{T} & \mathbf{I} \end{pNiceArray} , &
  \mathbf{T}_1 \oplus_V \mathbf{T}_2 & \defeq \mathbf{T}_1 + \mathbf{T}_2, &
  \mathbf{T}_1 \otimes_V \mathbf{T}_2 & \defeq \mathbf{T}_1 \cdot \mathbf{T}_2, 
  &
  \mathbf{T}_1 \pcho{p}_V \mathbf{T}_2 & \defeq p \cdot \mathbf{T}_1 + (1-p) \cdot \mathbf{T}_2.
\end{align*}}%
The $\omega$PMA $\calV$ admits
the partial order $\mathbf{T}_1 \aord_V \mathbf{T}_2 \defeq \mathbf{T}_1 \le \mathbf{T}_2$ (i.e., pointwise comparison)
and the subtraction operation $\mathbf{T}_1 \ominus_V \mathbf{T}_2 \defeq \mathbf{T}_1 - \mathbf{T}_2$ (i.e., pointwise subtraction).
Observing that the linear-recursion-solving strategy for $\calV$ needs to essentially
solve a system of linear matrix equations (similar to the setting in \cref{Se:BayesianInferenceAnalysis}), we again apply LP to solve linear recursion. 

\begin{figure}
\centering
\begin{subfigure}{0.25\textwidth}
\centering\footnotesize
\begin{spacing}{0.8}
\begin{pseudo}
  \kw{proc} $X$() \kw{begin} \\+
    \kw{if} \kw{prob}($\frac{2}{3}$) \kw{then} \\+
      \kw{skip} \\-
    \kw{else} \\+
      $t \coloneqq t + x$; \\
      \kw{if} \kw{prob}($\frac{1}{2}$) \kw{then} \\+
        $x \coloneqq 1$ \\-
      \kw{else} $x \coloneqq 0$ \kw{fi}; \\
      $X$(); $t \coloneqq t + 1$; \\
      $X$(); $t \coloneqq t + 1$ \\-
    \kw{fi} \\-
  \kw{end}
\end{pseudo}
\end{spacing}
\caption{A recursive program}
\end{subfigure}
\begin{subfigure}{0.48\textwidth}
\centering\footnotesize
\begin{align*}
  \nu^{(0)} & : \bbE[1] {\le} 1, \bbE[x'] {\le} 0, \bbE[t'] {\le} t \\[-5pt]
  f(\nu^{(0)}) & : \bbE[1] {\le} 1, \bbE[x'] {\le} \textstyle\frac{2}{3} x, \bbE[t'] {\le}  t + \frac{1}{3} x + \frac{2}{3} \\[-5pt]
  \delta^{(0)} & : \bbE[1] {\le} 0, \bbE[x'] {\le} \textstyle\frac{2}{3} x, \bbE[t'] {\le} \frac{1}{3} x + \frac{2}{3} \\[-5pt]
  \Delta^{(0)} & : \bbE[1] {\le} 0, \bbE[x'] {\le} \textstyle\frac{2}{3} x, \bbE[t'] {\le} \frac{1}{3} x + \frac{13}{6} \\ \hline
  \nu^{(1)} & : \bbE[1] {\le} 1, \bbE[x'] {\le} \textstyle\frac{2}{3} x, \bbE[t'] {\le} t + \frac{1}{3} x + \frac{13}{6} \\[-5pt]
  f(\nu^{(1)}) & : \bbE[1] {\le} 1, \bbE[x'] {\le} \textstyle\frac{2}{3} x + \frac{2}{27}, \bbE[t'] {\le} t +  \frac{1}{3}x + \frac{119}{54} \\[-5pt]
  \delta^{(1)} & : \bbE[1] {\le} 0, \bbE[x'] {\le} \textstyle\frac{2}{27}, \bbE[t'] {\le} \frac{1}{27} \\[-5pt]
  \Delta^{(1)} & : \bbE[1] {\le} 0, \bbE[x'] {\le} \textstyle\frac{1}{6}, \bbE[t'] {\le} \frac{1}{6} \\ \hline
  \nu^{(2)} & : \bbE[1] {\le} 1, \bbE[x'] {\le} \textstyle\frac{2}{3} x + \frac{1}{6}, \bbE[t'] {\le} t + \frac{1}{3}x + \frac{7}{3} \\[-5pt]
  f(\nu^{(2)}) & : \bbE[1] {\le} 1, \bbE[x'] {\le} \textstyle\frac{2}{3} x + \frac{1}{6},  \bbE[t'] {\le} t + \frac{1}{3}x + \frac{7}{3} 
\end{align*}
\caption{The Newton sequence for the program in (a)}
\end{subfigure}
\begin{subfigure}{0.25\textwidth}
\centering\footnotesize
\begin{spacing}{0.8}
\begin{pseudo}
  $(x,y) \coloneqq (0,0)$; \\
  \kw{while} $x-y < n$ \kw{do} \\+
    \kw{if} \kw{prob}($\frac{3}{4}$) \kw{then} \\+
      $x \coloneqq x + 1$ \\-
    \kw{else} \\+
      $y \coloneqq y + 1$ \\-
    \kw{fi}; \\
    $t \coloneqq t + 1$ \\-
  \kw{od}; \\
  \kw{skip}
\end{pseudo}
\end{spacing}
\caption{A random-walk program}
\end{subfigure}
\caption{Example programs for expected-value analysis.}\label{Fi:ExampleExpectedValueAnalysis}
\vspace{-0.5em}
\end{figure}

  Consider the recursive procedure $X$ shown in \cref{Fi:ExampleExpectedValueAnalysis}(a)
  with two program variables $x$ and $t$.
  The variable $t$ can be seen as a reward accumulator and thus its expected value corresponds
  to the expected accumulated reward.
  \cref{Fi:ExampleExpectedValueAnalysis}(b) presents the Newton sequence obtained when analyzing
  procedure $X$ with respect to $\omega$PMA $\calV$.
  We observe that the Newton sequence converges in three iterations and the analysis derives expectation invariants $\bbE[x'] \le \frac{2}{3} x + \frac{1}{6}$ and $\bbE[t'] \le t + \frac{1}{3}x + \frac{7}{3}$.
  Note that the sequence shown in \cref{Fi:ExampleExpectedValueAnalysis}(b) does \emph{not} start the iteration from $\azero_V$;
  instead, we set $\nu^{(0)}$ to be
$\begin{pNiceArray}[small]{c|c c}[first-col,first-row]
     & 1& x' & t' \\
  1 & 1 & 0 & 0 \\
  \hline
  x & 0 & 0 & 0 \\
  t & 0 & 0 & 1
\end{pNiceArray}$
  because (i) we are interested in \emph{partial correctness} (thus, we assume the program terminates with probability one by having
  $\bbE[1] \le 1$ as an expectation invariant),
 and (ii) the program shown in \cref{Fi:ExampleExpectedValueAnalysis}(a) will only
  increment $t$ during its execution
  (thus $\bbE[t'] \le t$ is certainly an under-approximation of the final solution
  and we can safely set $\bbE[t'] \le t$ in $\nu^{(0)}$).

A straightforward way to define the conditional-choice operation $\gcho{\varphi}_V$
is to use pointwise maximum of the operand matrices, i.e., $\mathbf{A} \gcho{\varphi}_V \mathbf{B} \defeq \max(\mathbf{A}, \mathbf{B})$.
Alternatively, we can implement a more precise conditional-choice operation by adopting the idea of \emph{rewrite functions}~\cite{CAV:CHR17,PLDI:NCH18}, which provide a mechanism to incorporate information about invariant conditions into an analysis.
A rewrite function at a given location $l$ in the program is based on a term $t$ over program variables.
When $t \ge 0$ is an invariant at $l$, $t$ can be used as a \emph{nonnegativity rewrite function}.
When $t = 0$ is an invariant at $l$, $t$ can be used as a \emph{vanishing rewrite function}.
The rewrite functions we use are linear expressions over program variables (and such an expression can contain negative coefficients).
For example, consider the probabilistic program in \cref{Fi:ExampleExpectedValueAnalysis}(c).
which implements a biased one-dimensional random walk, where the current position of the walk is represented by $x-y$.
By a simple non-probabilistic program analysis, we can derive that (i) at the beginning of each loop iteration,
the expression $n - x + y - 1$ is always nonnegative (and hence can be used as a nonnegativity rewrite function), and (ii) at the end of the loop, the expression $n - x + y$ is always zero (and hence can be used as a vanishing rewrite function).
Consider computing $\mathbf{A} \gcho{x{-}y<n}_V \mathbf{B}$, where $\mathbf{A}$ and $\mathbf{B}$ summarize
the property of the computations starting from line 3 and line 10 of the program, respectively.
Before performing the $\max$ operator, we can add $c_1 \cdot (n-x+y-1)$ and $c_2 \cdot (n-x+y)$ to the upper bounds indicated by $\mathbf{A}$ and $\mathbf{B}$, respectively,
where $c_1 \ge 0$ and $c_2$ is an arbitrary number.
For example, if $\mathbf{A}$ implies $\bbE[t'] \le t + 2n-2x+2y$ and $\mathbf{B}$ implies $\bbE[t'] \le t$,
the direct maximum would yield $\bbE[t'] \le t +2n+2y$,
but we can add $2 \cdot (n-x+y)$ to the right-hand side of $\bbE[t'] \le t$ and obtain the more precise upper bound $\bbE[t'] \le t+2n-2x+2y$.
Using the technique of rewriting functions, we can again reduce the problem of computing $\gcho{\varphi}_V$
to LP.
In practice, we assume that each $\mi{cond}[\varphi](,)$ is associated with two sets of rewriting functions $\Gamma_\varphi,\Gamma_{\neg\varphi} \subseteq \bbR[\vec{x}]$ of linear expressions, in the sense that the expressions in $\Gamma_\varphi$
(resp., $\Gamma_{\neg\varphi}$) are guaranteed to be always nonnegative if the conditional choice takes the ``then'' (resp., ``else'') branch.\footnote{
  A vanishing rewrite function $t$ can be represented by two nonnegativity rewrite functions $t$ and ${-}t$.
}
Our expected-value analysis derives the following expectation invariants for the random-walk program:
$
\bbE[t'] \le t+2n, \bbE[x'] \le \textstyle \frac{3}{2} n, \bbE[y'] \le \frac{1}{2} n.
$

\begin{table}
\centering
\footnotesize
\caption{Selected evaluation results of expectation-invariant analysis. (Time is in seconds.)}
\label{Ta:ExpectationInvariantAnalysis}
\begin{tabular}{c c c l}
\hline
\# & Program & Time & Invariants (selected) \\ \hline
8 & eg-tail & 0.008772 & $\bbE[x'] \le x + 3, \bbE[y'] \le y + 3, \bbE[z'] \le 0.25 z + 0.75$ \\[-2pt]
12 & recursive & 0.001473 & $\bbE[x'] \le x + 9$ \\[-2pt]
15 & coupon-five-fsm & 0.005259 & $\bbE[t'] \le t + 11.4167$ \\[-2pt]
16 & dice & 0.001575 & $\bbE[r'] \le 3.5, \bbE[t'] \le t + 1.33333$ \\[-2pt]
17 & exponential & 0.040215 & $\bbE[n'] \le 1, \bbE[r'] \le 0.2$ \\[-2pt]
18 & geometric & 0.001840 & $\bbE[i'] \le n, \bbE[t'] \le -2560 i + 2560 n + t$ \\[-2pt]
19 & non-linear-recursion & 0.001642 & $\bbE[t'] \le t + 0.333333 x + 2.33333, \bbE[x'] \le 0.666667 x + 0.166667$ \\[-2pt]
22 & unbiased & 0.000624 & $\bbE[r'] \le 1.5, \bbE[t'] \le t + 11.1111$ \\
\hline
\end{tabular}  
\end{table}

Our implementation of the abstract domain for linear expected-value analysis consists of about 500 lines of code.
To evaluate its effectiveness, we collected 22 benchmark programs---13 of which come from prior work~\cite{PLDI:WHR18}.
\cref{Ta:ExpectationInvariantAnalysis} presents selected evaluation results (the full table is included in
\iflong
\Cref{Se:Appendix:ExpectationInvariantAnalysis}),
\else
the technical report~\cite{Techreport}),
\fi
including the derived
expectation invariants and the time taken by the implementation.
Note that the domain used here captures a more restrictive class of invariants than the one used in prior work~\cite{PLDI:WHR18}.
Thus, a comparison of running times is not informative.
Nevertheless, the analysis created via \framework{} still efficiently obtains non-trivial results.
For this domain, Newton's method is slower than Kleene iteration on most of the benchmarks.
%
The reasons are that (i) the benchmark programs are quite small;
(ii) only 8 of the 22 programs involve recursion,
and (iii) the LP solving used by our instantiation has some runtime overhead.

\subsection{Non-linear Expectation Invariants via Expectation-recurrence Analysis}
\label{Se:ExpectationRecurrenceAnalysis}

\begin{wrapfigure}{r}{0.38\textwidth}
\vspace{-0.5em}
\centering
\begin{subfigure}{0.38\textwidth}
\centering\small
\begin{spacing}{0.8}
\begin{pseudo*}
  \kw{assume}($n > 0$); $i \coloneqq 0$; \\
  \kw{while} $i<n$ \kw{do} \\+
    $y \coloneqq y + 0.1 * x$; \\
    $z \sim \cn{Unif}(-0.1,0.1)$; \\
    $x \coloneqq 0.8 * x + z$; \\
    $i \coloneqq i + 1$ \\-
  \kw{od}
\end{pseudo*}
\end{spacing}
\vspace{-0.7em}
\caption{A lane-keeping model}
\end{subfigure}

\vspace{1.2em}

\begin{subfigure}{0.38\textwidth}
\centering\small
\begin{spacing}{0.8}
\begin{pseudo*}
  \kw{assume}($n{>}0$); $(i,x,y,z) \coloneqq (0,1,1,1)$; \\
  \kw{while} $i<n$ \kw{do} \\+
    $x \sim \cn{Cat}( (x+1, 0.5), (x+2, 0.5) )$; \\
    $y \sim \cn{Cat}( (y{+}z{+}x{*}x, \frac{1}{3}), (y{-}z{-}x, \frac{2}{3}) )$; \\
    $z \sim \cn{Cat}( (z+y, \frac{1}{4}), (z-y, \frac{3}{4}) )$; \\
    $i \coloneqq i + 1$ \\-
  \kw{od}
\end{pseudo*}
\end{spacing}
\vspace{-0.7em}
\caption{A program with non-linear arithmetic}
\end{subfigure}
\caption{Example programs.}
\label{Fi:ExampleForExpectationRecurrence}
\end{wrapfigure}

Prior work~\cite{ATVA:BKS19,TACAS:BGP16} has studied program analysis of probabilistic programs where
the control flow of a program does \emph{not} depend on randomness;
i.e., the program does not branch on the values of
random variables.
Such probabilistic programs are not uncommon in applications such as robotics and control theory~\cite{TACAS:BGP16}.
\cref{Fi:ExampleForExpectationRecurrence}(a) shows a probabilistic program that
models a simple lane-keeping controller:
variable $y$ tracks the distance between the controlled vehicle
and the middle of the lane;
during each iteration, the vehicle changes its position with respect to an angle stored
in variable $x$, and then updates the angle by a random amount (see the random variable $z$).
Expectation invariants that involve $\bbE[y']$ can be used to derive tail bounds on the
final position $y'$ of the vehicle via \emph{concentration-of-measure} inequalities~\cite{book:DP09}, such as the Markov inequality.

When the control flow of analyzed programs is not probabilistic,
we can categorize program variables into deterministic ones $\m{Var}_{\m{det}}$ and random ones $\m{Var}_{\m{rnd}}$.
Randomness can only flow to variables in $\m{Var}_{\m{rnd}}$,
and branch conditions can only depend on variables in $\m{Var}_{\m{det}}$.
We then apply the abstract domain of Compositional Recurrence Analysis (CRA)~\cite{FMCAD:FK15,POPL:KCB18}---reviewed
below---to analyze expectation recurrences.

\paragraph{Transition-formula analysis}
  We consider \emph{transition formulas} over integer arithmetic.
  A transition formula $\varphi(\m{Var},\m{Var}')$
  is a formula over the variables $\m{Var}$ (e.g., $\{i,x,y\}$) and their primed copies $\m{Var}'$ (e.g., $\{i',x',y'\}$),
  representing the values of the variables in the pre- and post-state of an execution.
  For example, the transition formula of ``$i{\coloneqq}i{-}1$'' is $i'=i-1 \wedge x'=x \wedge y' = y$.
  The sequencing operation is defined as relational composition and choice is defined as disjunction:
    $\varphi_1 \otimes \varphi_2 \defeq \Exists{\m{Var}''} \varphi_1(\m{Var},\m{Var}'') \wedge \varphi_2(\m{Var}'',\m{Var}')$, 
    $\varphi_1 \oplus \varphi_2 \defeq \varphi_1 \vee \varphi_2.$
  
  The key innovation of CRA
  is a method for interpreting the iteration operation $^\circledast$ in the transition-formula domain.
  Given a formula $\varphi_{\mi{body}}$ that encodes the property of a loop body, CRA computes a formula
  $\varphi_{\mi{body}}^\circledast$ that over-approximates any number of iterations of $\varphi_{\mi{body}}$,
  by extracting recurrence relations from $\varphi_\mi{body}$
  and computing their closed forms.
  For example, for
  ``\kw{while} $i>0$ \kw{do} \kw{if} $\star$ \kw{then} $x{\coloneqq}x{+}i$ \kw{else} $y{\coloneqq}y{+}i$ \kw{fi}; $i{\coloneqq}i{-}1$ \kw{od},'' CRA
  summarizes the loop body by the
  transition formula
  \[
  \varphi_{\mi{body}}: i > 0 \wedge \bigl( (x'=x+i \wedge y'=y) \vee (x'=x \wedge y'=y+i) \bigr) \wedge i' = i - 1 .
  \]
  CRA then extracts recurrences from $\varphi_\mi{body}$ and computes their closed forms; some instances are presented below,
  where $i^{(k)}$ denotes the value of $i$ at the end of the $k^{\textit{th}}$ iteration of the loop.
  \begin{center}
  \begin{small}
    \begin{tabular}{l|l@{\hspace{4em}}l|l}
      Recurrence & Closed form & Recurrence & Closed form \\ \hline
      $i' = i - 1$ & $i^{(k)} = i^{(0)} - k$ & $x'+y'=x+y+i$ & $x^{(k)}+y^{(k)} = x^{(0)} + y^{(0)}+k(k-1)/2+k i^{(0)}$
    \end{tabular}
  \end{small}
  \end{center}
  Finally, CRA introduces an existentially quantified nonnegative variable $k$ as the number of iterations and conjoins the closed forms to obtain a summary of the loop:
  \[
  \varphi_\mi{body}^\circledast : \Exists{k} k \ge 0 \wedge i' = i-k \wedge x'+y'=x+y +k(k-1)/2+k i .
  \]

\paragraph{Expectation-recurrence analysis}
The abstract domain consists of transition formulas of the form
$\varphi(\m{Var}_{\m{det}} \cup \m{Var}_{\m{rnd}}, \m{Var}_{\m{det}}' \cup \bbE[\m{Var}_{\m{rnd}}'])$,
which describe non-probabilistic invariants for $\m{Var}_{\m{det}}$ and expectation invariants for $\m{Var}_{\m{rnd}}$.
For example, the transition formula of ``$z{\sim}\cn{Unif}(-0.1,0.1)$'' is $i'=i \wedge n'=n \wedge \bbE[x'] = x \wedge \bbE[y'] = y \wedge \bbE[z'] = 0$.
Let $\tuple{ I, \oplus_I, \otimes_I, ^{\circledast_I}, \azero_I, \aone_I }$ be the regular algebra for CRA
where $I$ is the set of transition formulas for both non-probabilistic and expectation invariants.

We define a PMA $\calI = \tuple{ I, \aord_I, \otimes_I, \gcho{\varphi}_I, \pcho{p}_I, \dashcup_I, \azero_I, \aone_I }$
with the following definitions of operations:
\begin{gather*}
  \begin{aligned}
  \varphi_1 \aord_I \varphi_2 & \defeq (\varphi_1 {\implies} \varphi_2), &
  \varphi_1 \gcho{\psi} \varphi_2 & \defeq (\psi \otimes_I \varphi_1) \oplus_I (\neg\psi \otimes_I \varphi_2), &
  \varphi_1 \dashcup_I \varphi_2 & \defeq \varphi_1 \oplus_I \varphi_2,
  \end{aligned}
\end{gather*}
where we do not define $\pcho{p}_I$ because we assume that the control flow of programs is not random.

We now describe a loop-solving strategy $\m{solve}$ for this analysis.
Consider $\m{solve}_\calK(\{X = E_X\}_{X \in \calX}, \gamma )$ where $E_X$'s are regular infinite-tree expressions
in $\m{RegExp}^{\m{cf}}(\calF_\m{intra}, \calK \uplus \calX)$.
Because the ranked alphabet $\calF_\m{intra}$ for intraprocedural analysis does not contain procedure-call symbols,
the equation system $\{X = E_X\}_{X \in \calX}$ can be interpreted as a system of \emph{regular} equations over
the algebra $\tuple{I,\oplus_I,\otimes_I,\azero_I,\aone_I}$.
Thus, we can use the approach for algebraic program analysis of non-probabilistic programs~\cite{CAV:KRC21,JACM:Tarjan81,arxiv:FK13},
i.e., apply Tarjan's path-expression algorithm~\cite{JACM:Tarjan81:Alg} to solve the regular equations.
  We do not describe a linear-recursion-solving strategy because it is complex and
  not central to this article.
  Interprocedural CRA is developed by \citet{PLDI:KBB17}.

  For the lane-keeping program from \cref{Fi:ExampleForExpectationRecurrence}(a), our analysis obtains the following non-linear expectation invariants:
  $
  \bbE[z'] = 0 , \bbE[x'] = (\textstyle\frac{4}{5})^n \cdot x , \bbE[y'] = y + \frac{1}{2} \cdot x - \frac{1}{2}  \cdot (\frac{4}{5})^n \cdot x.
  $
  \cref{Fi:ExampleForExpectationRecurrence}(b) presents another example program 
  where $\cn{Cat}(\cdot)$ stands for categorical distributions.
  Our analysis derives
  \[
  \bbE[x'] = \textstyle\frac{3}{2} \cdot n + 1 , \bbE[y'] = -9 \cdot n - \frac{37}{5} \cdot (\frac{2}{3})^n + \frac{67}{5} \cdot (\frac{3}{2})^n - 5 , \bbE[z'] = \frac{9}{4} \cdot n^2 + \frac{19}{4} \cdot n - \frac{201}{10} \cdot (\frac{3}{2})^{n} - \frac{37}{5} \cdot (\frac{2}{3})^n + \frac{57}{2} .
  \]

To implement the abstract domain for expectation-recurrence analysis, we extended CRA with about 200 lines of code.
To evaluate its effectiveness,
we collected 14 benchmark programs---12 of which come from prior work on \textsc{Polar}~\cite{OOPSLA:MSB22}.\footnote{%
We excluded a few benchmark programs from \textsc{Polar}, each of which involves non-trivial stochastic control flow.
As we discussed earlier in \cref{Se:ExpectationRecurrenceAnalysis}, our domain for expectation-recurrence analysis does not support control flow that depends on the values of random variables.
Extending this domain to support stochastic control flow is interesting future research.
}
\cref{Ta:ExpectationRecurrenceAnalysis} presents the results of the evaluation, including the derived expectation
invariants and the time taken by the implementation.
The analysis instantiated via \framework{} is capable of deriving complex non-linear expectation recurrences
and achieves comparable performance with \textsc{Polar} on many benchmark programs.
Another take-away is that \framework{} is flexible enough to allow techniques originally 
developed for analyzing non-probabilistic programs to be extended to probabilistic programs.

\newcommand{\pfrac}[2]{#1/#2}
\begin{table}
\centering
\footnotesize
\caption{
Selected results from
expectation-recurrence analysis. (Time is in seconds.)}
\label{Ta:ExpectationRecurrenceAnalysis}
\begin{tabular}{c c c l}
\hline
\multirow{1}{*}{\#} & \multirow{1}{*}{Program} & Time (our work / \textsc{Polar}) & \multirow{1}{*}{Invariants (selected) derived by our work} \\ \hline
1 & 50coinflips & 108.920 / \textbf{0.448} & $\bbE[\mi{r0}'] = 1 - (\pfrac{1}{2})^{n}, \bbE[\mi{total}'] = 50(1 -  (\pfrac{1}{2})^{n-1} )$ \\[-2pt]
3 & dbn-component-health & \textbf{0.109} / 0.181 & $\bbE[\mi{obs}'] = (\pfrac{2683}{1000}) \cdot (\pfrac{9}{100})^n + \pfrac{701}{1000}$ \\[-2pt]
8 & hermann3 & T/O / \textbf{0.294} & - \\[-2pt]
9 & las-vegas-search & \textbf{0.059} / 0.288 & $\bbE[\mi{found}'] = 1 -  (\pfrac{20}{21})^n, \bbE[\mi{attempts}'] = \pfrac{20 n}{21}  $ \\[-2pt]
\multirow{1}{*}{11} & \multirow{1}{*}{randomized-response} & \multirow{1}{*}{\textbf{0.093} / 0.168} & $\bbE[\mi{n1}'] = \pfrac{5n}{4} - p n, \bbE[\mi{p1}'] = \pfrac{n}{4} + \pfrac{p n}{2}, \bbE[\mi{ntrue}'] = \pfrac{3n}{4}$ \\[-2pt]
\multirow{1}{*}{12} & \multirow{1}{*}{rock-paper-scissors} & \multirow{1}{*}{\textbf{0.110} / 0.215} & $\bbE[\mi{p1bal}'] = n(\mi{p2}(-\mi{q2}-2\mi{q3}+1) + \mi{p3}(\mi{q2}-\mi{q3})+\mi{q3})$ \\
\hline
\end{tabular}
\end{table}


\vspace{-2pt}
\section{Related Work}
\label{Se:RelatedWork}

\Omit{
We have discussed prior work that is most closely related to our work in previous sections.
%
%
We now discuss related work on automata techniques and the two program analyses
described in \cref{Se:CaseStudies}.
}

\paragraph{Markov decision processes (MDPs)}
MDPs provide a well-studied mechanism to model both probabilistic and nondeterministic choices~\cite{book:Puterman94}.
\citet{ICALP:EY05} proposed recursive Markov decision processes (RMDPs), which augment MDPs with recursive calls.
RMDPs can serve as models of nondeterministic probabilistic programs, but as discussed in \cref{Se:TheGap}, they assume a finite space for program states, whereas
our \framework{} framework aims at general state spaces.
Nevertheless, for the case study on Bayesian inference of Boolean programs described
in \cref{Se:BayesianInferenceAnalysis}, the state space is finite and the inference
problem can be reduced to reaching-probability analysis of RMDPs.
On the algorithmic side, several studies~\cite{ICALP:EGK08,ICALP:ESY12,ICALP:ESY15} have
generalized Newton's method to analyze the reaching probability of RMDPs.
The linearization technique proposed in prior work coincides with our differentiation
technique for the particular case of the real semiring, whereas our technique applies
more broadly to $\omega$PMAs.
The prior work used specialized techniques, such as strategy iteration, for solving
linearized equations over the real semiring, which can be cast as
linear-recursion-solving strategies in our framework.

\paragraph{Automata techniques}
\Omit{
Tree-automata theory \cite{misc:TATA07} has a long history.
Fundamental results include:
(i) a tree language $L$ (a set of finite trees over a ranked alphabet) is \emph{recognizable} by some finite tree automaton if and only if
$L$ is a \emph{regular} tree language, and
(ii) a tree language $L$ is \emph{regular} iff $L$ can be finitely represented by a \emph{regular tree expression}.
Therefore, a finite tree automaton that recognizes $L$ can be converted to a regular tree expression for $L$, and vice versa.
%
%
In our work, we need to consider possibly-infinite control-flow hyper-paths;
thus, we adopted a well-studied theory of regular infinite-tree expressions~\cite{JTCS:Courcelle83,kn:Courcelle90}.
}
%
The theory of regular word languages has been lifted to languages that
can contain infinite words, and several different kinds of automata
have been defined \cite{book:PP04}.
%
\Omit{
(Interested readers can consult \citet{book:PP04}, or other resources, for more details.)
}
%
%
%
%
\citet{PLDI:ZK21} observed that the capability of describing sets of infinite words makes \emph{$\omega$-regular expressions}
suitable for reasoning about program \emph{termination}, and proposed an algebraic termination analysis (for non-probabilistic programs)
that uses $\omega$-regular path expressions to represent sets of infinite execution paths in a control-flow graph.
In our work, we have to consider multiple kinds of confluences and thus develop a theory of trees rather than words.

\paragraph{Bayesian inference of Boolean programs (\cref{Se:BayesianInferenceAnalysis})}
%
\citet{FSE:CRN13}
developed an (intraprocedural) dataflow analysis for computing posterior distributions\Omit{ and
used Algebraic Decision Diagrams \cite{FMSD:BFG97} to represent probability distributions compactly}.
\citet{OOPSLA:HBM20} proposed a technique that reduces Bayesian inference of discrete programs to weighted model counting\Omit{ and uses
Binary Decision Diagrams \cite{TC:Bryant86} to encode and manipulate Boolean formulas efficiently}.
\citet{OOPSLA:CMS23} developed a technique that reduces inference of recursive programs with recursive
datatypes to solving polynomial equations via Newton's method.
There has been a line of work on inference in probabilistic logics and logic programs~\cite{IJCAI:DKT07,TPLP:FDG15,TPLP:RS11,JTCS:CIN05}.
Probabilistic model checking~\cite{CAV:KNP11,CAV:DJK17} can also be used to perform Bayesian inference of finite-state programs.
It might be possible to obtain more efficient Newtonian analyzers for Bayesian inference
by formulating the aforementioned techniques in the framework of \framework{}.

\Omit{
\noindent
\textbf{Cost analysis of probabilistic programs (\cref{Se:HigherMomentAnalysis}).}
%
%
Expected-accumulated-reward analysis is a classic problem in the area of Markov decision processes (MDPs)~\cite{book:Puterman94}.
\citet{ICALP:EWY08} studied recursive MDPs with rewards.
In the context of program analysis, those MDP models can easily be implemented as probabilistic programs,
and the notion of rewards can also be generalized.
%
Running time is one of the most popular kind of costs and there have been many techniques for expected runtime analysis,
e.g.,~\cite{ESOP:KKM16,LICS:OKK16,ESOP:BKK18,POPL:CFN16,CAV:CFG16,POPL:FH15,PLDI:NCH18,OOPSLA:AMS20,ICFP:WKH20,TACAS:MHG21,TOPLAS:TOU21,CAV:CFM17}.
Non-monotone costs (such as cash flow) have also been studied~\cite{PLDI:WFG19,PLDI:WHR21A}.
%
It would be interesting future research to instantiate \framework{} to carry out interprocedural cost analysis
of probabilistic programs.
}

\paragraph{Expectation-invariant analysis (\cref{Se:ExpectationInvariantAnalysis} \& \cref{Se:ExpectationRecurrenceAnalysis})}
%
%
There has been a line of work on generating expectation invariants (sometimes called stochastic invariants or probabilistic invariants),
e.g.~\cite{CAV:CS13,SAS:CS14,CAV:BEF16,POPL:CNZ17,SAS:KMM10,ATVA:BKS19,arxiv:SO19,PLDI:WHR18}.
One observation is that
for a certain class of probabilistic programs,
we can apply existing techniques of generating invariants for non-probabilistic programs to discover expectation invariants.\
%
Our future research would investigate using non-probabilistic recurrence-solving techniques (e.g.,~\cite{PLDI:KBB17,POPL:KBC19,POPL:KCB18,PLDI:BCK20,FMCAD:FK15})
to derive higher-moment invariants of probabilistic programs.


\section{Conclusion}
\label{Se:Conclusion}

\framework{} features a new family of $\omega$-continuous pre-Markov algebras for
formulating abstract domains, an adoption of regular infinite-tree expressions for
encoding unstructured control-flow, and a new differentiation method for linearizing
interprocedural equation systems.
Experiments on four domains have demonstrated
that \framework{} is both efficient and flexible.
On the other hand, \framework{} aims at algebraic dataflow analysis, so
it does not support control-flow analysis or 
have access to contextual information from surrounding constructs when analyzing a program fragment.
In the future, we plan to develop more advanced and powerful instantiations of \framework{}.


\section*{Data-Availability Statement}

The artifact that includes a prototype implementation of \framework{}
and supports the case studies described in \cref{Se:CaseStudies} is available on Zenodo~\cite{Artifact}.

\begin{acks}
We thank the referees for suggestions of how to clarify several items in the presentation.
The work was supported, in part,
by a gift from Rajiv and Ritu Batra;
by ONR under grant N00014-17-1-2889;
and by NSF under grants CCF-\{2211968,2212558\}.
Any opinions, findings, and conclusions or recommendations
expressed in this publication are those of the authors,
and do not necessarily reflect the views of the sponsoring
entities.
\end{acks}

\bibliographystyle{ACM-Reference-Format}
\bibliography{db,misc}

\iflong
\newpage
\appendix
\section{A Theory of Regular Hyper-paths}
\label{Se:Appendix:TheoryOfHyperPaths}

Our development in this section is inspired by the theory on
regular tree languages and regular tree expressions~\cite{misc:TATA07}.
However, we have to tackle a technical challenge that a hyper-path can correspond
to a tree with infinite depth.
Induction is a well-established proof principle for reasoning about
inductively defined datatypes (e.g., finite lists and finite trees),
but does not work for infinite datatypes in general.
Thus, in this section, we rely on a principle of \emph{coinduction},
which, as shown in prior work (e.g.~\cite{FI:JKS17,JMSCS:KS17}),
is indeed a well-founded principle for reasoning about coinductive datatypes
(e.g., infinite streams and infinite trees).

\subsection{Possibly-infinite Trees}
\label{Se:Appendix:PossiblyInfiniteTrees}

A \emph{ranked alphabet} is a pair $(\calF,\mi{Arity})$ where $\calF$
is a nonempty set and $\mi{Arity}$ is a mapping from $\calF$ to $\bbN$.
The \emph{arity} of a symbol $f \in \calF$ is $\mi{Arity}(f)$.
The set of symbols of arity $n$ is denoted by $\calF_n$.
For simplicity, we use parentheses and commas to specify symbols with their arity;
for example, $f(,)$ specifies a binary symbol $f$.

Let $\calK$ be a set of constants (i.e., symbols with arity zero) called \emph{holes}.
We assume that the sets $\calK$ and $\calF_0$ are disjoint,
and there is a distinguished symbol $\infholder \not\in \calF_0 \cup \calK$ with arity zero.
Intuitively, a hole symbol in $\calK$ represents a placeholder for later substitution with trees,
and the $\infholder$ symbol indicates a ``yet unknown subtree.''
A \emph{possibly-infinite tree} $t \in \m{Tree}^\infty(\calF,\calK)$ over
the ranked alphabet $\calF$ and the set of variables $\calK$
is a partial map $t : \bbN_{>0}^* \rightharpoonup \calF \cup \calK \cup \{ \infholder \}$ with domain
$\dom(t) \subseteq \bbN_{>0}^*$ satisfying the
following properties:
\begin{itemize}[nosep,leftmargin=*]
  \item $\dom(t)$ is non-empty and prefix-closed (thus $\epsilon \in \dom(t)$);
  \item for any $p \in \dom(t)$, if $t(p) \in \calF_n$ for some $n > 0$, then
  $\{ j \mid p j \in \dom(t) \} = \{ 1, \cdots, n \}$;
  \item for any $p \in \dom(t)$, if $t(p) \in \calF_0 \cup \calK \cup \{\infholder\}$, then
  $\{ j \mid p j \in \dom(t) \} = \emptyset$.
\end{itemize}

For brevity, we denote by $a$ the finite tree $\{ \epsilon \mapsto a \}$, for any $a \in \calF_0 \cup \calK \cup \{ \infholder \}$.
We also denote by $f(s_1,\cdots,s_n)$ the possibly-infinite tree
\[
\{ \epsilon \mapsto f \} \cup \textstyle\bigcup_{j \in \{1,\cdots,n\}} \{ jp \mapsto s_j(p) \mid p \in \dom(s_j) \},
\]
for any $f \in \calF_n$ with $n > 0$ where $s_1,\cdots,s_n$ are possibly-infinite trees.

Let $t \in \m{Tree}^\infty(\calF,\calK)$ be a possibly-infinite tree.
Every element in $\dom(t)$ is called a \emph{position}.
A \emph{leaf position} is a position $p$ such that $p j \not\in \dom(t)$ for any $j \in \bbN$.
We denote by $\mi{root}(t)$ the \emph{root symbol} of $t$, defined by $\mi{root}(t) \defeq t(\epsilon)$.
A \emph{subtree} $t|_p$ of a tree $t$ at position $p$ is the tree defined by the following:
\begin{itemize}[nosep,leftmargin=*]
  \item $\dom(t|_p) = \{q \mid p q \in \dom(t) \}$, and
  \item $\Forall{q \in \dom(t|_p)} t|_p(q) = t(p q)$. 
\end{itemize}
A tree $t$ can then be decomposed as $f(t|_1,\cdots,t|_n)$, where $f = \mi{root}(t)$ with arity $n$.

To compare two possibly-infinite trees, we define a \emph{refinement relation} $\sqsubseteq_\scrT \subseteq \m{Tree}^\infty(\calF,\calK) \times \m{Tree}^\infty(\calF,\calK)$ as the \emph{maximum} binary relation satisfying the following property:
if $t_1 \sqsubseteq_\scrT t_2$, then either (i) $t_1 = \infholder$, or (ii)
$\mi{root}(t_1) = f$ and  $\mi{root}(t_2) = f$ for some $f \in \calF \cup \calK$,
and for all $i \in \{1,\cdots,\mi{Arity}(f) \}$, it holds that $t_1|_i \sqsubseteq_\scrT t_2|_i$.
This property formulates the idea that the symbol $\infholder$ indicates a ``yet unknown subtree.''
Because the relation $\sqsubseteq_\scrT$ is maximum, it can be characterized by the
\emph{greatest} fixed point of the operator below.
\[
\begin{split}
& T_{\sqsubseteq_\scrT}( R ) \defeq \{ \tuple{\infholder, t} \mid t \in \m{Tree}^\infty(\calF,\calK) \} \cup {} \\
& \quad  \{ \tuple{ f(s_1,\!\cdots\!,s_n), f(t_1,\!\cdots\!,t_n) } \mid f \in \calF \!\cup\! \calK, \mi{Arity}(f) \!=\! n, \Forall{j \in \{1,\!\cdots\!,n\}} \tuple{s_j,t_j} \in R \}.
\end{split}
\]

\begin{lemma}\label{Lem:Appendix:RefinementOrderPartial}
  The relation $\sqsubseteq_\scrT$ is a partial order on $\m{Tree}^\infty(\calF,\calK)$.
\end{lemma}
\begin{proof}
  The relation $\sqsubseteq_\scrT$ is reflexive because $\{ \tuple{t,t} \mid t \in \m{Tree}^\infty(\calF,\calK) \}$ satisfies properties of the refinement relation, and the relation $\sqsubseteq_\scrT$ is maximum by definition.
  
  We claim that $\sqsubseteq_\scrT$ is antisymmetric, i.e., if $t_1 \sqsubseteq_\scrT t_2$ and $t_2 \sqsubseteq_\scrT t_1$, then $t_1 = t_2$.
  If $\mi{root}(t_1) = \infholder$, then by $t_2 \sqsubseteq_\scrT t_1$ we also have $\mi{root}(t_2) = \infholder$, thus $t_1 = t_2$.
  Otherwise, we know that $\mi{root}(t_1) = f$ for some $f \in \calF \cup \calK$, then by $t_1 \sqsubseteq_\scrT t_2$ we have $\mi{root}(t_2) = f$, and for all $i \in \{1,\cdots,\mi{Arity}(f) \}$, it holds that $t_1|_i \sqsubseteq_\scrT t_2|_i$.
  By $t_2 \sqsubseteq_\scrT t_1$, we also have $t_2|_i \sqsubseteq_\scrT t_1|_i$ for all $i \in \{1,\cdots,\mi{Arity}(f) \}$.
  Therefore, by coinduction hypothesis\footnote{As pointed out by \citet{JMSCS:KS17}, we can use the coinduction hypothesis as long as there is progress in observing the roots of trees, and there is no further investigation of the children of the trees.}, we have $t_1|_i = t_2|_i$ for all $i \in \{1,\cdots,\mi{Arity}(f)\}$, thus $t_1 = t_2$.
  
  We claim that $\sqsubseteq_\scrT$ is transitive, i.e., if $t_1 \sqsubseteq_\scrT t_2$ and $t_2 \sqsubseteq_\scrT t_3$, then $t_1 \sqsubseteq_\scrT t_3$.
  If $\mi{root}(t_1) = \infholder$, then we have $t_1 \sqsubseteq_\scrT t_3$ by definition.
  Otherwise, we know that $\mi{root}(t_1) = f$ for some $f \in \calF \cup \calK$.
  By $t_1 \sqsubseteq_\scrT t_2$, we know that $\mi{root}(t_2) = f$ and for all $i \in \{1,\cdots,\mi{Arity}(f) \}$, it holds that $t_1|_i \sqsubseteq_\scrT t_2|_i$.
  By $t_2 \sqsubseteq_\scrT t_3$, we know that $\mi{root}(t_3) = f$ and for all $i \in \{1,\cdots,\mi{Arity}(f) \}$, it holds that $t_2|_i \sqsubseteq_\scrT t_3|_i$.
  Thus, by coinduction hypothesis, for all $i \in \{1,\cdots,\mi{Arity}(f)\}$, it holds that $t_1|_i \sqsubseteq_\scrT t_3|_i$.
  Therefore, by definition, we conclude that $t_1 \sqsubseteq_\scrT t_3$.
\end{proof}

The lemma below justifies the relation $\sqsubseteq_\scrT$ in the sense that if $t_1 \sqsubseteq_\scrT t_2$, it must hold that $t_2$ is obtained by substituting some $\infholder$ symbols in $t_1$ with trees.

\begin{lemma}\label{Lem:Appendix:RefinementOrderRefine}
  If $t_1 \sqsubseteq_\scrT t_2$, then $\dom(t_1) \subseteq \dom(t_2)$ and for all $p \in \dom(t_1)$, $t_1(p) {=} \infholder$ or $t_1(p) {=} t_2(p)$.
\end{lemma}
\begin{proof}
  If $\mi{root}(t_1) = \infholder$, then $\dom(t_1) = \{ \epsilon \}$ and the lemma holds obviously.
  Otherwise, we know that $\mi{root}(t_1) = f$ for some $f \in \calF \cup \calK$,
  and by $t_1 \sqsubseteq_\scrT t_2$, we have $\mi{root}(t_2) = f$ and for all $i \in \{1,\cdots,\mi{Arity}(f)\}$, it holds that $t_1|_i \sqsubseteq_\scrT t_2|_i$.
  By coninduction hypothesis, for all $i \in \{1,\cdots,\mi{Arity}(f)\}$, we know that
  the lemma holds for $t_1|_i \sqsubseteq_\scrT t_2|_i$.
  Thus, we have 
  \begin{align*}
  \dom(t_1) & = \textstyle\bigcup_{i \in \{1,\cdots,\mi{Arity}(f)\}} \{ i p \mid p \in \dom(t_1|_i) \} \\
  & \subseteq \textstyle\bigcup_{i \in \{1,\cdots,\mi{Arity}(f)\}} \{ i p \mid p \in \dom(t_2|_i) \} \\
  & = \dom(t_2).
  \end{align*}
  For any $p \in \dom(t_1)$, either $p = \epsilon$, or there exists $i$ such that $p = i p'$ and $p' \in \dom(t_1|_i)$.
  In the former case, we have $t_1(\epsilon) = f = t_2(\epsilon)$.
  In the latter case,
  we know that either $t_1|_i(p') = \infholder$ or $t_1|_i(p') = t_2|_i(p')$.
  Therefore, we conclude that either $t_1(p) = t_1(i p') = t_1|_i(p') = \infholder$,
  or $t_1(p) = t_1(i p') = t_1|_i(p') = t_2|_i(p') = t_2(i p') = t_2(p)$.
\end{proof}

The set $\m{Tree}^\infty(\calF,\calK)$ of possibly-infinite trees forms an $\omega$-cpo with respect to $\sqsubseteq_\scrT$.

\begin{lemma}\label{Lem:Appendix:RefinementOrderContinuous}
  The relation $\sqsubseteq_\scrT$ is an $\omega$-continuous partial order with $\infholder$ as the least element.
\end{lemma}
\begin{proof}
  Fix an $\omega$-chain $\{t_k\}_{k \in \bbN}$ of possibly-infinite trees with respect to the partial order $\sqsubseteq_\scrT$ (\cref{Lem:Appendix:RefinementOrderPartial}).
  We then try to construct the least upper bound of the chain, and call it $t'$.
  We set $\dom(t')$ to be $\bigcup_{k \in \bbN} \dom(t_k)$.
  By \cref{Lem:Appendix:RefinementOrderRefine}, we know that $\{\dom(t_k)\}_{k \in \bbN}$ is a $\subseteq$-chain.
  For a position $p \in \dom(t')$, we set $t'(p)$ to be
  \begin{itemize}
    \item $\infholder$, if there exists $K \in \bbN$ such that $t_k(p) = \infholder$ for all $k \ge K$;
    \item $t_{k_p}(p)$, if there exists $k_p \in \bbN$ satisfying $p \in \dom(t_{k_p})$ and $t_{k_p}(p) \neq \infholder$.
  \end{itemize}
  The well-definedness of $t'$ is guaranteed by \cref{Lem:Appendix:RefinementOrderRefine}.
  If $t'$ is an upper bound on $\{t_k\}_{k \in \bbN}$, then $t'$ is the least one as $\dom(t')$ is the least upper bound on $\{\dom(t_k)\}_{k \in \bbN}$.
  Thus, it remains to prove that $t'$ is indeed an upper bound.
  
  We then proceed the proof by coinduction. For any $\ell \in \bbN$,
  if $\mi{root}(t_\ell) = \infholder$, then $t_\ell \sqsubseteq_\scrT t'$.
  Thus, without loss of generality, we can assume that $t_1 \neq \infholder$.
  Let $f = \mi{root}(t_1) \in \calF \cup \calK$.
  Thus, by definition, we know that $\mi{root}(t_k) = f$ for all $k \in \bbN$, and also $\mi{root}(t') = f$.
  Let $n = \mi{Arity}(f)$.
  Then for each $j \in \{1,\cdots,n\}$, it holds that $\{ t_k|_j \}_{k \in \bbN}$ is an $\omega$-chain, and by coinduction hypothesis, we know that $t'|_j$ is an upper bound on $\{t_k|_j \}_{k \ge \bbN}$.
  Therefore, $f(t'|_1,\cdots,t'|_n)$ is an upper bound of
  $\{ f(t_k|_1,\cdots, t_k|_n) \}_{k \in \bbN}$.
\end{proof}

We now formulate a mechanism to define functions on possibly-infinite trees.
Note that the standard definition of inductive functions does \emph{not} work because of
the existence of infinite trees.
Let $(D,\sqsubseteq_D)$ be an $\omega$-cpo with a least element $\bot_D$.
To define a function from possibly-infinite trees in $\m{Tree}^\infty(\calF,\calK)$ to $D$, we require a mapping
$\m{base} : \calF_0 \cup \calK \to D$ and a family of mappings into $\omega$-continuous functions $\m{ind}_n : \calF_n \to [D^n \to D]$ for all $n > 0$ such that $\calF_n \neq \emptyset$.
We then introduce an $\omega$-chain of functions $\{h_i\}_{i \ge 0}$ as follows, and define the target function by $h \defeq \bigsqcup_{i \ge 0} h_i$ with respect to the pointwise extension of $\sqsubseteq_D$.
\begin{align*}
  h_0 & \defeq \lambda t.\; \bot_D, \\
  h_{i+1} & \defeq \lambda t.\; \begin{dcases*}
    \bot_D & if $\mi{root}(t) = \infholder$ \\
    \m{base}(\mi{root}(t)) & if $\mi{root}(t) \in \calF_0 \cup \calK$ \\
    \m{ind}_n(\mi{root}(t))( h_i(t|_1), \cdots, h_i(t|_n ) ) & if $\mi{root}(t) \in \calF_n$ for some $n > 0$
  \end{dcases*} .
\end{align*}

\begin{lemma}\label{Lem:Appendix:CoinductionPreserveBottom}
  Let $h : \m{Tree}^\infty(\calF,\calK) \to D$ be a function from
  possibly-infinite trees to an $\omega$-cpo $(D,\sqsubseteq_D)$ with a least element $\bot_D$.
  Then $h(\infholder) = \bot_D$.
\end{lemma}
\begin{proof}
  By definition, we know that there exists an $\omega$-chain of functions $\{h_j\}_{j \ge 0}$
  defined as above such that $h = \bigsqcup_{j \ge 0} h_j$.
  Also, it is obvious that for each $j \ge 0$, $h_j(\infholder) = \bot_D$.
  Thus, we conclude that $h(\infholder) = \bigsqcup_{j \ge 0} h_j(\infholder) = \bigsqcup_{j \ge 0} \bot_D = \bot_D$.
\end{proof}

\begin{lemma}\label{Lem:Appendix:CoinductionContinuous}
  Let $h : \m{Tree}^\infty(\calF,\calK) \to D$ be a function from
  possibly-infinite trees to an $\omega$-cpo $(D,\sqsubseteq_D)$ with a least element $\bot_D$,
  induced by mappings $\m{base}$ and $\m{ind}_n$ for $n > 0$.
  Then $h$ is $\omega$-continuous with respect to the pointwise extension of $\sqsubseteq_D$.
\end{lemma}
\begin{proof}
  We know that there exists an $\omega$-chain of functions
  $\{ h_j \}_{j \ge 0}$ defined as above such that $h = \bigsqcup_{j \ge 0} h_j$.
  We claim that for each $j \ge 0$, the function $h_j$ is $\omega$-continuous.
  We proceed by induction on $j$.
  \begin{description}[labelindent=\parindent]
    \item[When $j=0$:]\
    
    By definition, we have $h_0 \defeq \lambda t.\; \bot$, which is obviously $\omega$-continuous.
    
    \item[When $j=k+1$:]\
        
    By definition, we have
    \[
    h_j \defeq \lambda t.\; \begin{dcases*}
      \bot & if $\mi{root}(t) = \infholder$, \\
      \m{base}(\mi{root}(t)) & if $\mi{root}(t) \in \calF_0 \cup \calK$, \\
      \m{ind}_n(\mi{root}(t))(h_k(t|_1),\cdots,h_k(t|_n)) & if $\mi{root}(t) \in \calF_n$ for some $n > 0$.
    \end{dcases*}.
    \]
    
    Consider an $\omega$-chain of trees $\{t_i\}_{i \ge 0}$.
    Without loss of generality, we assume that $\mi{root}(t_0) \neq \infholder$ (because $h_j(\{ \epsilon \mapsto \infholder \}) = \bot$).
    
    Then by the definition of the refinement order $\sqsubseteq_\scrT$, we know that
    $\{ \mi{root}(t_i) \mid i \ge 0\}$ is a singleton set.
    
    If the singleton set contains a symbol $a \in \calF_0 \cup \calK$,
    we know that $h_j(t_i) = \m{base}(a)$ for all $i \ge 0$, thus $h_j(\bigsqcup_{i \ge 0} t_i)= h_j(a) = \bigsqcup_{i \ge 0} h_j(a) = \bigsqcup_{i \ge 0} h_j(t_i)$.
    
    If the singleton set contains a symbol $f \in \calF_n$ for some $n > 0$,
    then by the assumption that $\m{ind}_n(f)$ is $\omega$-continuous and by induction hypothesis that $h_k$ is $\omega$-continuous, we can derive
    \begin{align*}
      h_j(\textstyle\bigsqcup_{i \ge 0} t_i) & = \m{ind}_n(f)(h_k( (\textstyle\bigsqcup_{i \ge 0} t_i)|_1), \cdots, h_k( (\textstyle\bigsqcup_{i \ge 0} t_i)|_n)) \\
      & = \textstyle \m{ind}_n(f)(h_k(\bigsqcup_{i \ge 0} (t_i|_1)), \cdots, h_k(\bigsqcup_{i \ge 0}(t_i |_ n))) \\
      & = \textstyle \m{ind}_n(f)( \bigsqcup_{i \ge 0} h_k(t_i|_1), \cdots, \bigsqcup_{i \ge 0} h_k(t_i|_n)) \\
      & = \textstyle \bigsqcup_{i_1 \ge 0} \cdots \bigsqcup_{i_n \ge 0} \m{ind}_n(f)(h_k(t_{i_1}|_1), \cdots, h_k(t_{i_n}|_n)) \\
      & \hphantom{{}={}} \text{($\sqsupseteq$: obvious)} \\
      & \hphantom{{}={}} \text{($\sqsubseteq$: $\forall i_1, \cdots i_n$, the LHS item is bounded by the $\max(i_1,\cdots,i_n)$-th RHS item)} \\
      & = \textstyle \bigsqcup_{i \ge 0} \m{ind}_n(f)(h_k(t_i|_1), \cdots, h_k(t_i|_n)) \\
      & = \textstyle \bigsqcup_{i \ge 0} h_j(t_i).
    \end{align*}
    Thus, we conclude the proof of the claim that $h_j$ is $\omega$-continuous for all $j \ge 0$.
  \end{description}
  
  Let us now consider an $\omega$-chain of trees $\{t_i\}_{i \ge 0}$. Then we can conclude the proof by
  \begin{align*}
    h(\textstyle\bigsqcup_{i \ge 0} t_i) & = (\textstyle\bigsqcup_{j \ge 0} h_j)(\textstyle\bigsqcup_{i \ge 0} t_i ) \\
    & = \textstyle\bigsqcup_{j \ge 0} h_j(\textstyle\bigsqcup_{i \ge 0} t_i) \\
    & = \textstyle\bigsqcup_{j \ge 0} \bigsqcup_{i \ge 0} h_j(t_i) \\
    & = \textstyle\bigsqcup_{i \ge 0} \bigsqcup_{j \ge 0} h_j(t_i) \\
    & = \textstyle\bigsqcup_{i \ge 0} (\textstyle \bigsqcup_{j \ge 0} h_j) (t_i) \\
    & = \textstyle\bigsqcup_{i \ge 0} h(t_i).
  \end{align*}
\end{proof}

We end this section by demonstrating a function that maps possibly-infinite trees
to collections of \emph{rooted paths} on trees.
We define $\mi{paths}(t) \subseteq \m{Path}^\infty(\calF,\calK) \defeq (\calF \cup \calK) \cdot (\bbN \cdot (\calF \cup \calK))^\infty$
with the following base and induction steps:
\begin{align*}
  \mi{paths} & : \m{Tree}^\infty(\calF,\calK) \to \m{Path}^\infty(\calF,\calK) \\
  \mi{paths}(a) & \defeq \{ a \} \quad \text{for $a \in \calF_0 \cup \calK$}, \\
  \mi{paths}(f(s_1,\cdots,s_n)) & \defeq \{ f  j  w \mid j=1,\dots,n \wedge w \in \mi{paths}(s_j) \} \quad \text{for $f \in \calF_n$ for some $n > 0$}.
\end{align*}
Note that the set $\m{Path}^\infty(\calF,\calK)$ of possibly-infinite paths admits
an $\omega$-cpo with the ordering
\[
A \sqsubseteq_{\m{P}} B \defeq (A \cap \m{Path}^+(\calF,\calK) \subseteq B \cap \m{Path}^+(\calF,\calK)) \wedge (A \cap \m{Path}^\omega(\calF,\calK) \supseteq B \cap \m{Path}^\omega(\calF,\calK)),
\]
and the least element $\bot_\m{P} \defeq \m{Path}^\omega(\calF,\calK)$,
where $\m{Path}^+(\calF,\calK) \defeq (\calF \cup \calK) \cdot (\bbN \cdot (\calF \cup \calK))^*$ is the set of finite paths,
and $\m{Path}^\omega(\calF,\calK) \defeq (\calF \cup \calK) \cdot (\bbN \cdot (\calF \cup \calK))^\omega$ is the set of infinite paths.
We can again see that the $\infholder$ symbol indicates a ``yet unknown tree,''
because $\mi{paths}(\infholder)$, by \cref{Lem:Appendix:CoinductionPreserveBottom}, equals $\bot_\m{P}$, i.e., the set of all infinite paths
(and no finite paths).

\subsection{Hyper-path Interpretation of Regular Infinite-tree Expressions}
\label{Se:Appendix:InterpRegularHyperPathExpressions}

To formally interpret regular infinite-tree expressions as possibly-infinite trees,
we first introduce a \emph{substitution} operator, defined as a function on possibly-infinite trees with the following base and induction steps,
where ${-}_1$ and ${-}_2$ indicate the two arguments of the substitution operator,
and the function is defined coinductively on the first argument:
\begin{align*}
  ({-}_1) \{ \square \leftleadsto {-}_2 \} & : \m{Tree}^\infty(\calF,\calK \cup \{ \square \}) \times \m{Tree}^\infty(\calF,\calK ) \to \m{Tree}^\infty(\calF,\calK  ) \\
  \square \{ \square \leftleadsto t \} & \defeq t, \\
  a \{ \square \leftleadsto t \} & \defeq a \quad \text{for $a \neq \square$}, \\
  f(s_1,\cdots,s_n)\{ \square \leftleadsto t \} & \defeq f(s_1\{ \square \leftleadsto t \} ,\cdots,s_n\{ \square \leftleadsto t \}) .
\end{align*}

\begin{example}\label{Exa:Appendix:TreeSubstitution}
  Let $\calF = \{ \mi{seq}[x{\sim}\cn{Ber}(0.5)](), \mi{cond}[x](,), \varepsilon \}$
  and $\calK = \{ \square_1, \square_2 \}$.
  Let $t = \mi{cond}[x](\square_1, \square_2)$ and $s = \mi{seq}[x{\sim}\cn{Ber}(0.5)](\varepsilon)$.
  Then
  \begin{center}\small
  $t\{ \square_1 \leftleadsto s \} =$ \begin{tabular}{l} \Tree[.{$\mi{cond}[x]$} [.{$\mi{seq}[x{\sim}\cn{Ber}(0.5)]$} {$\varepsilon$} ] {$\square_2$} ] \end{tabular}
  ; \quad
  $t\{ \square_2 \leftleadsto s \} =$ \begin{tabular}{l} \Tree[.{$\mi{cond}[x]$} {$\square_1$} [.{$\mi{seq}[x{\sim}\cn{Ber}(0.5)]$} {$\varepsilon$} ] ] \end{tabular}
  .
  \end{center}
\end{example}

\begin{proposition}\label{Prop:Appendix:TreeSubstContinuous}
  The substitution operator is $\omega$-continuous in its second argument.
\end{proposition}

We now define an interpretation from regular infinite-tree expressions to possibly-infinite
trees via the mapping $\scrT_\calK\interp{{-}}$, which is inductively defined on the structure of regular infinite-tree expressions as follows, where ``$\lfp$'' denotes the least-fixed-point operator.
\begin{align*}
  \scrT_\calK\interp{{-}} & : \m{RegExp}^\infty(\calF,\calK) \to \m{Tree}^\infty(\calF,\calK) \\
  \scrT_\calK\interp{a} & \defeq a,  \\
  \scrT_\calK\interp{f(E_1,\cdots,E_n)} & \defeq f(\scrT_\calK\interp{ s_1},\cdots, \scrT_\calK\interp{s_n}),  \\
  \scrT_\calK\interp{E_1 \cdot_\square E_2} & \defeq \scrT_{\calK \cup \{\square\}}\interp{E_1} \{ \square \leftleadsto \scrT_\calK\interp{E_2} \}, \\
  \scrT_\calK\interp{E^{\infty_\square}} & \defeq \lfp_{  \infholder }^{\sqsubseteq_\scrT} \lambda t.\; (\scrT_{\calK\cup\{\square\}}\interp{E} \{ \square \leftleadsto t \}).
\end{align*}

\begin{example}\label{Exa:Appendix:RegExpInterp}
  The regular infinite-tree expression $E \defeq (\mi{cond}[x](\mi{seq}[x{\sim}\cn{Ber}(0.5)](\square), \varepsilon))^{\infty_\square}$ gives a finite representation of the infinite hyper-path of the loop ``\kw{while} $x$ \kw{do} $x\sim\cn{Ber}(0.5)$ \kw{od}.''
  Intuitively, the map $\scrT_\emptyset\interp{E}$ constructs the following
  $\omega$-chain of trees and obtains the least upper bound as the interpretation of $E$:
  \begin{center}\small
    \begin{tabular}{l} $\infholder$ \end{tabular}
    $\sqsubseteq_\scrT$
    \begin{tabular}{l} \Tree[.{$\mi{cond}[x<0]$} [.{$\mi{seq}[x \coloneqq x + 1]$} {$\infholder$} ] {$\varepsilon$} ] \end{tabular}
    $\sqsubseteq_\scrT$
    \begin{tabular}{l} \Tree[.{$\mi{cond}[x<0]$} [.{$\mi{seq}[x \coloneqq x + 1]$} [.{$\mi{cond}[x<0]$} [.{$\mi{seq}[x \coloneqq x + 1]$} {$\infholder$} ] {$\varepsilon$} ] ] {$\varepsilon$} ] \end{tabular}
    $\sqsubseteq_\scrT$
    $\cdots$
    .
  \end{center}
\end{example}

\subsection{Hyper-path Interpretation is Sound for MA Interpretation}
\label{Se:Appendix:HyperPathAbstractMA}

The key theoretical result we will establish in this section is that the tree-based interpretation (developed in \Cref{Se:Appendix:InterpRegularHyperPathExpressions}) is \emph{sound} for any MA interpretation, i.e., regular infinite-tree expressions with the same tree-based interpretation also yield the same interpretation under any MA.
(Recall the definition of Markov algebras reviewed in \cref{Se:Soundness}.)

\begin{theorem}\label{The:Appendix:Sound}
  Let $E,F \in \m{RegExp}^\infty(\calF,\calK)$.
  If
  $\scrT_\calK\interp{E} = \scrT_\calK\interp{F}$,
  then for any MA interpretation $\scrM = (\calM,\interp{\cdot}^\scrM)$
  and any hole-valuation $\gamma : \calK \to \calM$, it holds that $\scrM_{\gamma}\interp{E} = \scrM_\gamma\interp{F}$.
\end{theorem}

Before proving the theorem, we define an interpretation map from possibly-infinite trees $\m{Tree}^\infty(\calF,\calK)$ into a semantic algebra $\calM$ equipped with an interpretation $\scrM = (\calM,\interp{\cdot}^\scrM)$.
Because $\calM$ admits a dcpo, thus it also admits an $\omega$-cpo, I can follow the coinductive principle for function definitions on trees (developed in \Cref{Se:Appendix:PossiblyInfiniteTrees}).
Below gives the base and induction steps for an interpretation map $r^\scrM_\gamma$, parameterized by a hole-valuation $\gamma : \calK \to \calM$.
The interpretation is well-defined because the operators $\otimes_M$, $\gcho{\varphi}_M$, $\pcho{p}_M$ are all Scott-continuous, thus they are also $\omega$-continuous.
\begin{align*}
  r^\scrM_\gamma(\varepsilon) & \defeq 1_M \\
  r^\scrM_\gamma(\mi{seq}[\m{act}](s_1)) & \defeq \interp{\m{act}}^\scrM \otimes_M r^\scrM_\gamma(s_1) \\
  r^\scrM_\gamma(\mi{cond}[\varphi](s_1,s_2)) & \defeq r^\scrM_\gamma(s_1) \gcho{\varphi}_M r^\scrM_\gamma(s_2) \\
  r^\scrM_\gamma(\mi{prob}[p](s_1,s_2)) & \defeq r^\scrM_\gamma(s_1) \pcho{p}_M r^\scrM_\gamma(s_2) \\
  r^\scrM_\gamma(\mi{ndet}(s_1,s_2)) & \defeq r^\scrM_\gamma(s_1) \dashcup_M r^\scrM_\gamma(s_2) \\
  r^\scrM_\gamma(\square) & \defeq \gamma(\square)  
\end{align*}

We then prove the following property of $r^\scrM_\gamma$ about substitutions.

\begin{lemma}\label{Lem:Appendix:Subst}
  For any $t \in \m{Tree}^\infty(\calF, \calK \cup \{\square\})$, $u \in \m{Tree}^\infty(\calF, \calK)$,
  and $\gamma : \calK \to \calM$,
  it holds that $r^\scrM_{\gamma[\square \mapsto r^\scrM_\gamma(u)]}(t) = r^\scrM_{\gamma}( t \{ \square \leftleadsto u \} )$.
\end{lemma}
\begin{proof}
  Let $\gamma' \defeq \gamma[\square \mapsto r^\scrM_\gamma(u)]$.
  Then $\gamma' : (\calK \cup \{ \square \}) \to \calM$.
  
  We proceed by induction on the structure of $t$.
  \begin{description}[labelindent=\parindent]
    \item[Case $t = \varepsilon$:]\
    
    By definition, we have $r^\scrM_{\gamma'}(\varepsilon) = 1_M$.
    
    By definition, we have $\varepsilon\{\square \leftleadsto u\} = \varepsilon$, and also $r^\scrM_\gamma(\varepsilon) = 1_M$.
    
    Thus, we conclude that $r^\scrM_{\gamma'}(\varepsilon) = r^\scrM_\gamma(\varepsilon \{ \square \leftleadsto u\} )$.
    
    \item[Case {$t = \mi{seq}[\m{act}](s)$}:]\
    
    By induction hypothesis, we have $r^\scrM_{\gamma'}(s) = r^\scrM_\gamma(s \{ \square \leftleadsto u \})$.
    
    By definition, we have $r^\scrM_{\gamma'}(\mi{seq}[\m{act}](s)) = \interp{\m{act}}^\scrM \otimes_M r^\scrM_{\gamma'}(s)$.
    
    By definition, we have $\mi{seq}[\m{act}](s) \{ \square \leftleadsto u \} = \mi{seq}[\m{act}]( s\{ \square \leftleadsto u \} )$,
    and also $r^\scrM_\gamma(\mi{seq}[\m{act}]( s \{ \square \leftleadsto u \})) = \interp{\m{act}}^\scrM \otimes_M r^\scrM_\gamma( s \{\square \leftleadsto u \} )$.
    
    Thus, we conclude that $r^\scrM_{\gamma'}(\mi{seq}[\m{act}](s)) = r^\scrM_\gamma(\mi{seq}[\m{act}](s) \{ \square \leftleadsto u \})$.
    
    \item[Case {$t = \mi{cond}[\varphi](s_1,s_2)$}:]\
    
    By induction hypothesis on $s_1$ and $s_2$, respectively, we have $r^\scrM_{\gamma'}(s_1) = r^\scrM_\gamma(s_1\{\square \leftleadsto u\})$ and $r^\scrM_{\gamma'}(s_2) = r^\scrM_\gamma(s_2\{\square \leftleadsto u\})$, respectively.
    
    By definition, we have $r^\scrM_{\gamma'}(\mi{cond}[\varphi](s_1,s_2)) = r^\scrM_{\gamma'}(s_1) \gcho{\varphi}_M r^\scrM_{\gamma'}(s_2)$.
    
    By definition, we have $\mi{cond}[\varphi](s_1,s_2)\{\square \leftleadsto u\} = \mi{cond}[\varphi](s_1\{ \square \leftleadsto u \}, s_2\{\square \leftleadsto u \})$, and also
    \[
    r^\scrM_\gamma(\mi{cond}[\varphi](s_1\{ \square \leftleadsto u\}, s_2\{\square\leftleadsto u\}) ) = r^\scrM_\gamma( s_1\{ \square \leftleadsto u\} ) \gcho{\varphi}_M r^\scrM_\gamma( s_2\{\square\leftleadsto u\} ).
    \]
    
    Thus, we conclude that $r^\scrM_{\gamma'}(\mi{cond}[\varphi](s_1,s_2)) = r^\scrM_\gamma(\mi{cond}[\varphi](s_1,s_2)\{\square \leftleadsto u\})$.
    
    \item[Case {$t = \mi{prob}[p](s_1,s_2)$}:]\
    
    By induction hypothesis on $s_1$ and $s_2$, respectively, we have $r^\scrM_{\gamma'}(s_1) = r^\scrM_\gamma(s_1\{\square \leftleadsto u\})$ and $r^\scrM_{\gamma'}(s_2) = r^\scrM_\gamma(s_2\{\square \leftleadsto u\})$, respectively.
    
    By definition, we have $r^\scrM_{\gamma'}(\mi{prob}[p](s_1,s_2)) = r^\scrM_{\gamma'}(s_1) \pcho{p}_M r^\scrM_{\gamma'}(s_2)$.
    
    By definition, we have $\mi{prob}[p](s_1,s_2)\{\square \leftleadsto u\} = \mi{prob}[p](s_1\{ \square \leftleadsto u \}, s_2\{\square \leftleadsto u \})$, and also
    \[
    r^\scrM_\gamma(\mi{prob}[p](s_1\{ \square \leftleadsto u\}, s_2\{\square\leftleadsto u\}) ) = r^\scrM_\gamma( s_1\{ \square \leftleadsto u\} ) \pcho{p}_M r^\scrM_\gamma( s_2\{\square\leftleadsto u\} ).
    \]
    
    Thus, we conclude that $r^\scrM_{\gamma'}(\mi{prob}[p](s_1,s_2)) = r^\scrM_\gamma(\mi{prob}[p](s_1,s_2)\{\square \leftleadsto u\})$.
    
    \item[Case {$t = \mi{ndet}(s_1,s_2)$}:]\
    
    By induction hypothesis on $s_1$ and $s_2$, respectively, we have $r^\scrM_{\gamma'}(s_1) = r^\scrM_\gamma(s_1\{\square \leftleadsto u\})$ and $r^\scrM_{\gamma'}(s_2) = r^\scrM_\gamma(s_2\{\square \leftleadsto u\})$, respectively.
    
    By definition, we have $r^\scrM_{\gamma'}(\mi{ndet}(s_1,s_2)) = r^\scrM_{\gamma'}(s_1) \dashcup_M r^\scrM_{\gamma'}(s_2)$.
    
    By definition, we have $\mi{ndet}(s_1,s_2)\{\square \leftleadsto u\} = \mi{ndet}(s_1\{ \square \leftleadsto u \}, s_2\{\square \leftleadsto u \})$, and also
    \[
    r^\scrM_\gamma(\mi{ndet}(s_1\{ \square \leftleadsto u\}, s_2\{\square\leftleadsto u\}) ) = r^\scrM_\gamma( s_1\{ \square \leftleadsto u\} ) \dashcup_M r^\scrM_\gamma( s_2\{\square\leftleadsto u\} ).
    \]
    
    Thus, we conclude that $r^\scrM_{\gamma'}(\mi{ndet}(s_1,s_2)) = r^\scrM_\gamma(\mi{ndet}(s_1,s_2)\{\square \leftleadsto u\})$.
    
    \item[Case $t = \square$:]\
    
    By definition, we have $r^\scrM_{\gamma'}(\square) = \gamma'(\square) = r^\scrM_\gamma(u)$.
    
    By definition, we have $\square\{\square \leftleadsto u\} = u$.
    
    Thus, we conclude that $r^\scrM_{\gamma'}(\square) = r^\scrM_{\gamma}(\square\{\square \leftleadsto u \})$.
    
    \item[Case $t = \square'$ \textnormal{where $\square' \neq \square$}:]\
    
    By definition, we have $r^\scrM_{\gamma'}(\square') = \gamma'(\square') = \gamma(\square')$.
    
    By definition, we have $\square'\{\square \leftleadsto u\} = \square'$, and also $r^\scrM_\gamma(\square') = \gamma(\square')$.
    
    Thus, we conclude that $r^\scrM_{\gamma'}(\square') = r^\scrM_{\gamma}(\square'\{\square \leftleadsto u \})$.
  \end{description}
\end{proof}

We can now present the proof of \cref{The:Appendix:Sound}.

\begin{proof}[Proof of \cref{The:Appendix:Sound}]
  Fix $E,F \in \m{RegExp}^\infty(\calF,\calK)$ such that $\scrT_\calK\interp{E} = \scrT_\calK\interp{F}$.
  Fix an interpretation $\scrM = (\calM, \interp{\cdot}^\scrM)$
  and a hole-valuation $\gamma : \calK \to \calM$.
  The theorem requires us to show that $\scrM_\gamma\interp{E} = \scrM_\gamma\interp{F}$.
  We claim that for any expression $E$, it holds that $\scrM_\gamma\interp{E} = r^\scrM_\gamma( \scrT_\calK\interp{E} )$.
  If this is true, we can conclude the proof by
  \[
  \scrM_\gamma\interp{E} = r^\scrM_\gamma( \scrT_\calK\interp{E} ) = r^\scrM_\gamma( \scrT_\calK\interp{F} ) = \scrM_\gamma\interp{F}.
  \]
  
  To prove the claim, we fix an expression $E$ and proceed by induction on the structure of $E$.
    \begin{description}[labelindent=\parindent]
    \item[Case $E = \varepsilon$:]\
    
    By definition, we have $\scrM_\gamma\interp{\varepsilon} = 1_M$.
    
    By definition, we have $\scrT_\calK\interp{\varepsilon} = \varepsilon$, and also $r^\scrM_\gamma( \varepsilon ) = 1_M$.
    
    Thus, we conclude that $\scrM_\gamma\interp{\varepsilon} = r^\scrM_\gamma(\scrT_\calK\interp{\varepsilon})$.
    
    \item[Case {$E = \mi{seq}[\m{act}](E_1)$}:]\
    
    By induction hypothesis, we have $\scrM_\gamma\interp{E_1} = r^\scrM_\gamma(\scrT_\calK\interp{E_1})$.
    
    By definition, we have $\scrM_\gamma\interp{\mi{seq}[\m{act}](E_1)} = \interp{\m{act}}^\scrM \otimes_M \scrM_\gamma\interp{E_1}$.
    
    By definition, we have $\scrT_\calK\interp{\mi{seq}[\m{act}](E_1)} =  \mi{seq}[\m{act}](s_1) $, where $s_1 \defeq \scrT_\calK\interp{E_1}$,
    and also $r^\scrM_\gamma(\mi{seq}[\m{act}]( s_1 )) =  \interp{\m{act}}^\scrM \otimes_M r^\scrM_\gamma(s_1)$.
    
    Thus, we conclude that $\scrM_\gamma\interp{\mi{seq}[\m{act}](E_1)} = r^\scrM_\gamma(\scrT_\calK\interp{\mi{seq}[\m{act}](E_1)})$.
    
    \item[Case {$E = \mi{cond}[\varphi](E_1,E_2)$}:]\
    
    By induction hypothesis on $E_1$ and $E_2$, respectively, we have $\scrM_\gamma\interp{E_1} = r^\scrM_\gamma(\scrT_\calK\interp{E_1})$ and $\scrM_\gamma\interp{E_2} = r^\scrM_\gamma(\scrT_\calK\interp{E_2})$, respectively.
    
    By definition, we have $\scrM_\gamma\interp{\mi{cond}[\varphi](E_1,E_2)} = \scrM_\gamma\interp{E_1} \gcho{\varphi}_M \scrM_\gamma\interp{E_2}$.
    
    By definition, we have $\scrT_\calK\interp{\mi{cond}[\varphi](E_1,E_2)} =  \mi{cond}[\varphi](s_1,  s_2)$,  where $s_1 \defeq \scrT_\calK\interp{E_1}$, $s_2 \defeq \scrT_\calK\interp{E_1}$,  and also
    $r^\scrM_\gamma(\mi{cond}[\varphi](s_1,s_2)) = r^\scrM_\gamma(s_1) \gcho{\varphi}_M r^\scrM_\gamma(s_2)$.
    
    Thus, we conclude that $\scrM_\gamma\interp{\mi{cond}[\varphi](E_1,E_2)} = r^\scrM_\gamma(\scrT_\calK\interp{\mi{cond}[\varphi](E_1,E_2)})$.
    
    \item[Case {$E = \mi{prob}[p](E_1,E_2)$}:]\
    
    By induction hypothesis on $E_1$ and $E_2$, respectively, we have $\scrM_\gamma\interp{E_1} = r^\scrM_\gamma(\scrT_\calK\interp{E_1})$ and $\scrM_\gamma\interp{E_2} = r^\scrM_\gamma(\scrT_\calK\interp{E_2})$, respectively.
    
    By definition, we have $\scrM_\gamma\interp{\mi{prob}[p](E_1,E_2)} = \scrM_\gamma\interp{E_1} \pcho{p}_M \scrM_\gamma\interp{E_2}$.
    
    By definition, we have $\scrT_\calK\interp{\mi{prob}[p](E_1,E_2)} =  \mi{prob}[p](s_1,  s_2)$,  where $s_1 \defeq \scrT_\calK\interp{E_1}$, $s_2 \defeq \scrT_\calK\interp{E_1}$,  and also
    $r^\scrM_\gamma(\mi{prob}[p](s_1,s_2)) = r^\scrM_\gamma(s_1) \pcho{p}_M r^\scrM_\gamma(s_2)$.
    
    Thus, we conclude that $\scrM_\gamma\interp{\mi{prob}[p](E_1,E_2)} = r^\scrM_\gamma(\scrT_\calK\interp{\mi{prob}[p](E_1,E_2)})$.
    
    \item[Case {$E = \mi{ndet}(E_1,E_2)$}:]\
    
    By induction hypothesis on $E_1$ and $E_2$, respectively, we have $\scrM_\gamma\interp{E_1} = r^\scrM_\gamma(\scrT_\calK\interp{E_1})$ and $\scrM_\gamma\interp{E_2} = r^\scrM_\gamma(\scrT_\calK\interp{E_2})$, respectively.
    
    By definition, we have $\scrM_\gamma\interp{\mi{ndet}(E_1,E_2)} = \scrM_\gamma\interp{E_1} \dashcup_M \scrM_\gamma\interp{E_2}$.
    
    By definition, we have $\scrT_\calK\interp{\mi{ndet}(E_1,E_2)} =  \mi{ndet}(s_1,  s_2)$,  where $s_1 \defeq \scrT_\calK\interp{E_1}$, $s_2 \defeq \scrT_\calK\interp{E_1}$,  and also
    $r^\scrM_\gamma(\mi{ndet}(s_1,s_2)) = r^\scrM_\gamma(s_1) \dashcup_M r^\scrM_\gamma(s_2)$.
    
    Thus, we conclude that $\scrM_\gamma\interp{\mi{ndet}(E_1,E_2)} = r^\scrM_\gamma(\scrT_\calK\interp{\mi{ndet}(E_1,E_2)})$.

    \item[Case $E = \square$ \textnormal{for some $\square \in \calK$}:]\
    
    By definition, we have $\scrM_\gamma\interp{\square} = \gamma(\square)$.
    
    By definition, we have $\scrT_\calK\interp{\square} = \square$, and also $r^\scrM_\gamma(\square) = \gamma(\square)$.
    
    Thus, we conclude that $\scrM_\gamma\interp{\square} = r^\scrM_\gamma(\scrT_\calK\interp{\square})$.
    
    \item[Case $E = E_1 \cdot_\square E_2$:]\
    
    By induction hypothesis on $E_2$, we have $\scrM_\gamma\interp{E_2} = r^\scrM_\gamma(\scrT_\calK\interp{E_2})$.
    
    Let $\gamma' \defeq \gamma[\square \mapsto \scrM_\gamma\interp{E_2}]$.
    Then $\gamma' : (\calK \cup \{\square\}) \to \calM$.
    
    By induction hypothesis on $E_1$ (which is an expression in $\m{RegExp}^\infty(\calF,\calK \cup \{ \square \})$), we have $\scrM_{\gamma'}\interp{E_1} = r^\scrM_{\gamma'}(\scrT_{\calK \cup \{\square\}}\interp{E_1})$.
    
    By definition, we have $\scrM_\gamma\interp{E_1 \cdot_\square E_2} = \scrM_{\gamma'}\interp{E_1}$.
    
    By definition, we have $\scrT_\calK\interp{E_1 \cdot_\square E_2} = \scrT_{\calK \cup \{\square\}}\interp{E_1} \{ \square \leftleadsto \scrT_\calK\interp{E_2} \}$.
    
    By \cref{Lem:Appendix:Subst}, we know that for any $s \in \m{Tree}^\infty(\calF,\calK\cup\{\square\})$, it holds that $r^\scrM_{\gamma'}(s) = r_\gamma(s\{ \square \leftleadsto \scrT_\calK\interp{E_2} \})$.
    Therefore, we have
    \[
    r^\scrM_\gamma( \scrT_{\calK \cup \{\square\}}\interp{E_1} \{ \square \leftleadsto \scrT_\calK\interp{E_2} \} ) = r^\scrM_{\gamma'}( \scrT_{\calK \cup \{\square\}}\interp{E_1} ).
    \]
    
    Thus, we conclude that $\scrM_\gamma\interp{E_1 \cdot_\square E_2} = r^\scrM_\gamma(\scrT_\calK\interp{E_1 \cdot_\square E_2})$.
    
    \item[Case $E = (E_1)^{\infty_\square}$:]\
    
    By definition, we have $\scrM_\gamma\interp{(E_1)^{\infty_\square}} = \lfp_{\bot_M}^{\aord_M} \lambda X.\; \scrM_{\gamma[\square \mapsto X]} \interp{E_1}$.
    
    By definition, we have $\scrT_\calK\interp{(E_1)^{\infty_\square}} = \lfp_{  \infholder  }^{\sqsubseteq_\scrT} \lambda X.\; ( \scrT_{\calK \cup \{\square\}}\interp{E_1} \{ \square \leftleadsto X \} )$.
    
    Let $F \defeq \lambda X.\; \scrM_{\gamma[\square \mapsto X]} \interp{E_1}$
    and $G \defeq \lambda X.\; ( \scrT_{\calK \cup \{\square\}}\interp{E_1} \{ \square \leftleadsto X \} )$.
    By the Kleene fixed-point theorem,
    it suffices to show
    that for any $n \ge 0$, $F^n(\bot_M) = r^\scrM_\gamma(G^n(\infholder))$.
    We prove the claim by induction on $n$.
    \begin{description}
      \item[When $n=0$:]\
      
      By definition, we have $F^0(\bot_M) = \bot_M$.
      
      By definition, we have $G^0(\infholder) = \infholder $, and also $r^\scrM_\gamma(\infholder) = \bot_M$.
      
      Thus, we conclude this that $F^0(\bot_M) = r^\scrM_\gamma(G^0(\infholder))$.
      
      \item[When $n=k+1$:]\
      
      By induction hypothesis (on $n$), we have $F^k(\bot_M) = r^\scrM_\gamma(G^k(\infholder))$.
      
      Let $\gamma' \defeq \gamma[\square \mapsto F^k(\bot_M)]$.
      Then $\gamma' : (\calK \cup \{ \square \}) \to \calM$.
      
      By definition, we have $F^n(\bot_M) = \scrM_{\gamma'}\interp{E_1}$.
      
      By definition, we have $G^n(\infholder) = \scrT_{\calK \cup \{\square\}}\interp{E_1} \{ \square \gets G^k(\infholder) \}$, and also by \cref{Lem:Appendix:Subst},
      we have
      $
      r^\scrM_\gamma( \scrT_{\calK \cup \{\square\}}\interp{E_1} \{ \square \gets G^k(\infholder) \} ) = r^\scrM_{\gamma'}( \scrT_{\calK \cup \{\square\}}\interp{E_1} ).
      $
      
      By induction hypothesis (on $E_1$), we have $\scrM_{\gamma'}\interp{E_1} = r^\scrM_{\gamma'}(\scrT_{\calK\cup\{\square\}}\interp{E_1})$.
      
      Thus, we conclude that $F^n(\bot_M) = r^\scrM_\gamma(G^n(\infholder))$.
    \end{description}
  \end{description}
\end{proof}


\section{Probabilistic Programs and Control-flow Hyper-graphs}
\label{Se:Appendix:CFHG}

In this article, we use an imperative probabilistic programming
language developed by \citet{PLDI:WHR18}.
The semantics of the language is parameterized by a set $\calA$ of \emph{data actions} (such as assignment)
and a set $\calL$ of \emph{logical conditions} (such as comparison).
The language uses control-flow hyper-graphs (CFHGs) to encode unstructured control-flow
and multiple kinds of confluences.
A \emph{hyper-graph} $H$ is a quadruple $\tuple{V,E,v^\m{entry},v^\m{exit}}$, where
$V$ is a set of nodes, $E$ is a set of hyper-edges, and $v^\m{entry},v^\m{exit} \in V$ are
distinguished entry and exit nodes, respectively.
Each \emph{hyper-edge} in $E$ is a pair $\tuple{x,(y_1,\cdots,y_k)}$, where $k \ge 1$ and $x,y_1,\cdots,y_k \in V$
(and the order of $y_i$'s is significant).
We assume that $v^\m{entry}$ (resp., $v^\m{exit}$) does not have incoming (resp., outgoing) hyper-edges.

A \emph{probabilistic program} is then defined as a sequence of procedures $\{X_i \mapsto H_i\}_{i=1}^n$,
where the $i^\textit{th}$ procedure has name $X_i$ and CFHG $H_i = \tuple{V_i,E_i,v^\m{entry}_i,v^\m{exit}_i}$,
in which each node in $V_i \setminus \{ v^\m{exit}_i \}$ has exactly one outgoing hyper-edge.
We assume that $V_1,\cdots,V_n$ are pairwise disjoint.
Commands are attached to hyper-edges via a mapping $\mi{Cmd} : \bigcup_{i=1}^n E_i \to \m{Cmd}$,
where each command in $\m{Cmd}$ takes one of the following forms: $\mi{seq}[\m{act}]$ for a data action $\m{act} \in \calA$,
$\mi{cond}[\varphi]$ for a logical condition $\varphi \in \calL$, $\mi{prob}[p]$ for a number $p \in [0,1]$,
$\mi{ndet}$, or $\mi{call}[X_i]$ for a procedure $X_i$.

\begin{example}[Probabilistic Boolean programs]\label{Exa:Appendix:ProbabilisticBooleanPrograms}
  \cref{Fi:Appendix:ExampleProgramWithLoopsAndUnstructuredControlFlow}(a) demonstrates a probabilistic Boolean program, i.e., a probabilistic program
  with Boolean-valued variables.
  The grammar below defines data actions and logical conditions for Boolean programs, where $x$ is a program variable and $p \in [0,1]$.
  The data action ``$x \sim \cn{Ber}(p)$'' draws a random value from a Bernoulli distribution with mean $p$, i.e., the action
  assigns \kw{true} to $x$ with probability $p$, and otherwise assigns \kw{false}.
  {\begin{align*}
    \m{act} \in \calA & \Coloneqq x \coloneqq \varphi \mid x \sim \cn{Ber}(p) \mid \kw{skip} &
    \varphi \in \calL & \Coloneqq x \mid \kw{true} \mid \kw{false} \mid \neg \varphi \mid \varphi_1 \wedge \varphi_2 \mid \varphi_1 \vee \varphi_2
  \end{align*}}%
  \cref{Fi:Appendix:ExampleProgramWithLoopsAndUnstructuredControlFlow}(b) shows the CFHG of the program in \cref{Fi:Appendix:ExampleProgramWithLoopsAndUnstructuredControlFlow}(a), where
  $v_1$ is the entry node and $v_7$ is the exit node.
  There are three hyper-edges, each with two destination nodes in the CFHG:
  edge $\tuple{v_2, (v_3,v_1)}$ is associated with a conditional-branching command $\mi{cond}[b_1 \vee b_2]$,
  edge $\tuple{v_3, (v_7,v_4)}$ is associated with a probabilistic-branching command $\mi{prob}[0.1]$,
  and edge $\tuple{v4, (v_5,v_6)}$ is associated with a conditional-branching command $\mi{cond}[b_1]$.
  Note that \kw{return} (and \kw{goto}, etc.) are not data actions, and
  they are encoded as control-flow hyper-edges with a singleton-valued destination-node set.
\end{example}

\begin{figure}
\centering
\begin{subfigure}{0.27\textwidth}
\centering
\begin{small}
\begin{pseudo*}
  \kw{while} \kw{true} \kw{do} \\+
    $b_1 \sim \cn{Ber}(0.5)$; \\
    \kw{while} $b_1 \vee b_2$ \kw{do} \\+
      \kw{if} \kw{prob}(0.1) \kw{then} \kw{return} \kw{fi}; \\
      \kw{if} $b_1$ \kw{then} $b_1 \sim \cn{Ber}(0.2)$ \\
      \kw{else} $b_2 \sim \cn{Ber}(0.8)$ \\
      \kw{fi} \\-
    \kw{od} \\-
  \kw{od}
\end{pseudo*}
\end{small}
\caption{An example program}
\end{subfigure}
\begin{subfigure}{0.35\textwidth}
\centering
\resizebox{0.8\textwidth}{!}{%
\begin{tikzpicture}[node/.style={circle,draw,minimum size=20pt,font=\small},conf/.style={circle,fill,inner sep=0pt,outer sep=0pt,minimum size=0pt},every edge quotes/.style={font=\small},node distance=0.5cm]
  \node[node] (a) {$v_1$};
  \node[node] (b) [below=of a] {$v_2$};
  \node[conf] (cond) [below=of b] {};
  \node[node] (c) [below=of cond] {$v_3$};
  \node[conf] (prob) [below=of c] {};
  \node[node] (g) [left=of prob] {$v_7$};
  \node[node] (d) [below=of prob] {$v_4$};
  \node[conf] (phi) [right=1.5cm of d] {};
  \node[node] (e) [above right=0.6cm and 0.8cm of phi] {$v_5$};
  \node[node] (f) [below right=0.1cm and 0.8cm of phi] {$v_6$};
  
  \draw (a.south) edge["$\!\mi{seq}\sq{b_1{\sim}\cn{Ber}(0.5)}$",right,->] (b.north);
  \draw (b.south) edge["$\!\mi{cond}\sq{b_1{\vee}b_2}$",right] (cond.north);
  \draw (cond.west) edge["F",->,out=180,in=180] (a.west);
  \draw (cond.south) edge["T",->] (c.north);
  \draw (c.south) edge["$\!\mi{prob}\sq{0.1}$",right] (prob.north);
  \draw (prob.west) edge["T",->] (g.east);
  \draw (prob.south) edge["F",->] (d.north);
  \draw (d.east) edge["$\mi{cond}\sq{b_1}$",above] (phi.west);
  \draw (phi.east) edge["T",->] (e.west);
  \draw (phi.east) edge["F",->] (f.west);
  \draw (e.north) edge["$\mi{seq}\sq{b_1{\sim}\cn{Ber}(0.2)}$",->,in=30,out=45,sloped,below] (b.east);
  \draw (f.east) edge["$\mi{seq}\sq{b_2{\sim}\cn{Ber}(0.8)}$",->,sloped,above,in=30,out=40] (b.north east);
\end{tikzpicture}}
\caption{Control-flow hyper-graph}
\end{subfigure}
\begin{subfigure}{0.36\textwidth}
\centering
\resizebox{0.7\textwidth}{!}{%
\begin{tikzpicture}[op/.style={rectangle,draw,inner sep=3pt},node distance=0.2cm,font=\small]
  \node[op] (outer) {$\mu Z_1$};
  \node (a) [below=of outer] {$\mi{seq}[b_1{\sim}\cn{Ber}(0.5)]$};
  \node[op] (inner) [below=of a] {$\mu Z_2$};
  \node (b) [below=of inner] {$\mi{cond}[b_1{\vee}b_2]$};
  \node (c) [below left=0.2cm and -1.2cm of b] {$\mi{prob}[0.1]$};
  \node (b_to_outer) [below right=0.2cm and -0.5cm of b] {$Z_1$};
  \node (d) [below right=0.2cm and -1.2cm of c] {$\mi{cond}[b_1]$};
  \node (c_to_exit) [below left=0.2cm and -0.2cm of c] {$\varepsilon$};
  \node (e) [below left=0.2cm and -1.4cm of d] {$\mi{seq}[b_1{\sim}\cn{Ber}(0.2)]$};
  \node (f) [below right=0.2cm and 0cm of d] {$\mi{seq}[b_2{\sim}\cn{Ber}(0.8)]$};
  \node (e_to_inner) [below=of e] {$Z_2$};
  \node (f_to_inner) [below=of f] {$Z_2$};
  
  \draw (outer.south) edge[->] (a.north);
  \draw (a.south) edge[->] (inner.north);
  \draw (inner.south) edge[->] (b.north);
  \draw (b.south) edge[->] (c.north);
  \draw (b.south) ++ (5pt,0) edge[->] (b_to_outer.north west);
  \draw (c.south) ++ (-5pt,0) edge[->] (c_to_exit.north east);
  \draw (c.south) edge[->] (d.north);
  \draw (d.south) ++ (-3pt,0) edge[->] (e.north);
  \draw (d.south) ++ (3pt,0) edge[->] (f.north);
  \draw (e.south) edge[->] (e_to_inner.north);
  \draw (f.south) edge[->] (f_to_inner.north);
\end{tikzpicture}}
\vspace{-0.5em}
\caption{Regular infinite-tree expression AST}
\end{subfigure}
\begin{subfigure}{\textwidth}
\centering
\begin{small}
\[
\mu Z_1.\,
\mi{seq}[\m{act}_1]\Biggl(
  \mu Z_2.\,
  \mi{cond}[b_1{\vee}b_2]\biggl(
    \mi{prob}[0.1]\Bigl(
      \varepsilon,
      \mi{cond}[b_1]\bigl(
        \mi{seq}[\m{act}_2](
          Z_2
        ),
        \mi{seq}[\m{act}_3](
          Z_2
        )
      \bigr)
    \Bigr),
    Z_1
  \biggr)
\Biggr)
\]
\end{small}
\vspace{-1em}
\caption{Regular infinite-tree expression (where $\m{act}_1 \defeq b_1{\sim}\cn{Ber}(0.5)$, $\m{act}_2 \defeq b_1{\sim}\cn{Ber}(0.2)$, $\m{act}_3 \defeq b_2{\sim}\cn{Ber}(0.8)$)}
\end{subfigure}
\caption{An example program with nested loops and unstructured control-flow.}
\vspace{-1.3em}
\label{Fi:Appendix:ExampleProgramWithLoopsAndUnstructuredControlFlow}
\end{figure}

To represent the control-flow hyper-path of a probabilistic program,
we define the following ranked alphabet, where $\calA$ is the set of data actions and $\calL$ is the set of logical conditions, both introduced in \cref{Exa:Appendix:ProbabilisticBooleanPrograms}:
\begin{align*}
\calF & \defeq \{\varepsilon \} \cup \{ \mi{seq}[\m{act}]() \mid \m{act} \in \calA \}
\cup \{ \mi{cond}[\varphi](,) \mid \varphi \in \calL \} \cup \{ \mi{prob}[p](,) \mid p \in [0,1] \} \cup \{ \mi{ndet}(,) \} \\
 & \quad 
\cup \{ \mi{call}[X_i]() \mid \text{$X_i$ is a procedure} \}.
\end{align*}

\Omit{
\begin{remark}
  We omit the formal development of hyper-paths (as \textbf{possibly-infinite trees}) but include that in \Cref{Se:Appendix:TheoryOfHyperPaths}.
  %
  %
  Following a well-studied theory of infinite trees~\cite{JTCS:Courcelle83,kn:Courcelle90},
  we develop an interpretation that maps a regular infinite-tree expression to a possibly-infinite hyper-path,
  %
  and show that
  if two regular infinite-tree expressions have the same hyper-path interpretation,
  their respective algebraic interpretations (which we develop in \cref{Se:Soundness})
  are also the same.
  %
\end{remark}
}


We now define the correspondence between a regular infinite-tree expression and the CFHG of a probabilistic programs.
Recall that a 
CFHG
$H$ is a quadruple
$\tuple{V,E,v^\m{entry},v^\m{exit}}$, where each hyper-edge $e \in E$ is associated
with a command $\mi{Cmd}(e) \in \calF$.
%
%
Let $Z_v$ for each $v \in V$ be a free variable.
We then treat each hyper-edge $e = \tuple{v,(u_1,\cdots,u_k)} \in E$
as an \emph{equation} $Z_v = \mi{Cmd}(e)\; (Z_{u_1},\cdots,Z_{u_k})$,
where the right-hand side is a regular infinite-tree expression in $\m{RegExp}^\infty(\calF,\{Z_v \mid v \in V \})$ by the rule \textsc{(Node)}.
Thus, we can extract the following equation system from a
CFHG:
\begin{align*}
  Z_v & = \mi{Cmd}(e)\; (Z_{u_1},\cdots,Z_{u_k}) \quad \text{for $e = \tuple{v,(u_1,\cdots,u_k)} \in E$}, &
  Z_{v^\m{exit}} & = \varepsilon,
\end{align*}
where each free variable $Z_v$ appears exactly once as a left-hand side, because
every non-exit node has exactly one outgoing
hyper-edge.

We then apply Beki{\'c}'s theorem~\cite{kn:Bekic84} to compute a
regular infinite-tree expression $E_v \in \m{RegExp}^\infty(\calF,\emptyset)$ for each
left-hand side $Z_v$ ($v \in V$) as a closed-form solution to the equation system.
%
Informally, Beki{\'c}'s theorem allows splitting a mutual recursion into \emph{univariate} recursions;
in our setting, this amounts to treating the whole equation system extracted from a CFHG as a mutual recursion over
$\{Z_v \mid v \in V \}$ and obtaining for each left-hand side $Z_v$, an expression where all free variables are bound by univariate recursions, i.e., $\mu$-binders.
%
We present an implementation of Beki{\'c}'s theorem based on Gaussian elimination in \cref{Alg:GaussianElimination}.
%
\Omit{
The front-solving phase eliminates variables in $\{Z_i \mid i = 1,\cdots, n\}$ one-by-one,
in the sense that after the $i^{\textit{th}}$ step, the free variable $Z_i$ does \emph{not} appear in the
right-hand side of any equation $Z_j = R_j$ for $j \ge i$.
The back-solving phase then eliminates all variable occurrences from right-hand sides by,
at the $i^\textit{th}$ step, substituting the closed-form $R_i$ for the free variable $Z_i$ in the
equation $Z_j = R_j$ for $j < i$.
The key ingredient of this algorithm is the loop-solving step, which solves
a single recursive equation $Z_i = R_i$, where $Z_i$ appears in $R_i$, using the
$\mu$-binder and takes $\mu Z_i.\, R_i$ to be the solution.
}

\begin{algorithm}[tb!]
\caption{Beki{\'c}'s theorem based on Gaussian elimination}
\label{Alg:GaussianElimination}
\small
\begin{spacing}{0.8}
\begin{algorithmic}
  \Require An equation system $\{ Z_i = R_i \}_{i=1}^n$ where each $R_i \in \m{RegExp}^\infty(\calF, \{ Z_i \mid i = 1,\cdots, n\})$
  \Ensure A closed-form solution $\{Z_i = E_i\}_{i=1}^n$ where each $E_i \in \m{RegExp}^\infty(\calF,\emptyset)$
  \For{$i \gets 1$ to $n$} \Comment{Front-solving}
    \If{$Z_i$ appears in $R_i$}
      $R_i \gets \mu Z_i.\, R_i$ \Comment{Loop-solving}
    \EndIf
    \ForEach{$j > i$ such that $Z_i$ appears in $R_j$}
      $R_j \gets R_j \dplus_{Z_i} R_i$
    \EndFor
  \EndFor
  \For{$i \gets n$ to $2$} \Comment{Back-solving}
    \ForEach{$j < i$ such that $Z_i$ appears in $R_j$}
      $R_j \gets R_j \dplus_{Z_i} R_i$
    \EndFor
  \EndFor
  \State \Return $\{Z_i = R_i\}_{i=1}^n$
\end{algorithmic}
\end{spacing}
\end{algorithm}

\section{Proofs for Newton's method with $\omega$PMAs}
\label{Se:Appendix:Proofs}

\begin{definition}[\cref{De:RegularHyperPathDifferential}]
  Let $f(\vec{X}) \in \m{RegExp}^\infty(\calF^\alpha,\calK)$.
  The \emph{differential} of $f(\vec{X})$ with respect to $X_j$ at $\vec{\nu} \in \calM^n$ under a hole-valuation $\gamma : \calK \to \calM$,
  denoted by $\calD_{X_j} f_\gamma|_{\vec{\nu}}(\vec{Y})$, is defined as follows:
  {\footnotesize\[
  \calD_{X_j} f_\gamma|_{\vec{\nu}}(\vec{Y}) \defeq
  \begin{dcases*}
    \azero_M & if $f = c \in \calM$ \\
    \mi{seq}[c]( \calD_{X_j} g_\gamma|_{\vec{\nu}}(\vec{Y}) ) & if $f = \mi{seq}[c](g)$ \\
    \mi{seq}[\nu_k]( \calD_{X_j} g_\gamma|_{\vec{\nu}}( \vec{Y} ) ) & if $f = \mi{call}[X_k](g)$ and $k \neq j$ \\
    \oplus( \mi{call}_\m{lin}[Y_j;g_\gamma(\vec{\nu})] , \mi{seq}[\nu_j]( \calD_{X_j} g_\gamma|_{\vec{\nu}}(\vec{Y}) ) ) & if $f = \mi{call}[X_j](g)$ \\
    \mi{cond}[\varphi]( \calD_{X_j} g_\gamma|_{\vec{\nu}}(\vec{Y}) , \calD_{X_j} h_\gamma|_{\vec{\nu}}(\vec{Y}) ) & if $f = \mi{cond}[\varphi](g,h)$ \\
    \mi{prob}[p]( \calD_{X_j} g_\gamma|_{\vec{\nu}}(\vec{Y}) , \calD_{X_j} h_\gamma|_{\vec{\nu}}(\vec{Y}) ) & if $f = \mi{prob}[p](g,h)$ \\
    \ominus( \mi{ndet}( {\oplus }(g_\gamma(\vec{\nu}), \calD_{X_j} g_\gamma|_{\vec{\nu}}(\vec{Y}) ), {\oplus}(h_\gamma(\vec{\nu})  , \calD_{X_j} h_\gamma|_{\vec{\nu}}(\vec{Y}) ) ), f_\gamma(\vec{\nu}) ) & if $f = \mi{ndet}(g,h)$ \\
    \square & if $f = \square \in \calK$ \\
    \calD_{X_j} g_{\gamma[\square \mapsto h_\gamma(\vec{\nu})]}|_{\vec{\nu}}(\vec{Y}) \cdot_{\square} \calD_{X_j} h_\gamma|_{\vec{\nu}}(\vec{Y}) & if $f = g \cdot_{\square} h$ \\
    (\calD_{X_j} g_{\gamma[\square \mapsto f_\gamma(\vec{\nu})]}|_{\vec{\nu}}(\vec{Y}) )^{\infty_\square} & if $f = g^{\infty_\square}$
  \end{dcases*}
  .
  \]}
\end{definition}

\begin{figure}
\centering  
\begin{mathpar}\small
  \Rule{Const}
  { c \in \calM }
  { \Gamma \vdash c \Downarrow c \mid \emptyset }  
  \and
  \Rule{Seq}
  { c \in \calM \\ \Gamma \vdash E \Downarrow F \mid \Theta }
  { \Gamma \vdash \mi{seq}[c](E) \Downarrow \mi{seq}[c](F) \mid \Theta }
  \and
  \Rule{Cond}
  { \Gamma \vdash E_1 \Downarrow F_1 \mid \Theta_1 \\
    \Gamma \vdash E_2 \Downarrow F_2 \mid \Theta_2
  }
  { \Gamma \vdash \mi{cond}[\varphi](E_1,E_2) \Downarrow \mi{cond}[\varphi](F_1,F_2) \mid \Theta_1 \cup \Theta_2 }
  \and
  \Rule{Prob}
  { \Gamma \vdash E_1 \Downarrow F_1 \mid \Theta_1 \\
    \Gamma \vdash E_2 \Downarrow F_2 \mid \Theta_2
  }
  { \Gamma \vdash \mi{prob}[p](E_1,E_2) \Downarrow \mi{prob}[p](F_1,F_2) \mid \Theta_1 \cup \Theta_2 }
  \and
  \Rule{Ndet}
  { \Gamma \vdash E_1 \Downarrow F_1 \mid \Theta_1 \\
    \Gamma \vdash E_2 \Downarrow F_2 \mid \Theta_2
  }
  { \Gamma \vdash \mi{ndet}(E_1,E_2) \Downarrow \mi{ndet}(F_1,F_2) \mid \Theta_1 \cup \Theta_2 }
  \and
  \Rule{Hole}
  { }
  { \Gamma \vdash \square \Downarrow \Gamma(\square) \mid \emptyset }
  \and
  \Rule{Concatenation}
  { \Gamma[\square \mapsto \square] \vdash E_1 \Downarrow F_1 \mid \Theta_1 \\
    \Gamma \vdash E_2 \Downarrow F_2 \mid \Theta_2
  }
  { \Gamma \vdash E_1 \cdot_{\square} E_2 \Downarrow F_1 \cdot_{\square} F_2 \mid \Theta_1 \cup \Theta_2 }
  \and
  \Rule{Closure}
  { Z~\m{fresh} \\ \Gamma[ \square \mapsto Z ] \vdash E \Downarrow F \mid \Theta }
  { \Gamma \vdash E^{\infty_\square} \Downarrow Z \mid \Theta \cup \{ Z = F \}  }
  \and
  \Rule{Add}
  { \Gamma \vdash E_1 \Downarrow F_1 \mid \Theta_1 \quad
    \Gamma \vdash E_2 \Downarrow F_1 \mid \Theta_2
  }
  { \Gamma \vdash {\oplus}(E_1,E_2) \Downarrow {\oplus}(F_1,F_2) \mid \Theta_1 \cup \Theta_2 }
  \and
  \Rule{Sub}
  { \Gamma \vdash E_1 \Downarrow F_1 \mid \Theta_1 \quad
    \Gamma \vdash E_2 \Downarrow F_1 \mid \Theta_2
  }
  { \Gamma \vdash {\ominus}(E_1,E_2) \Downarrow {\ominus}(F_1,F_2) \mid \Theta_1 \cup \Theta_2 }
  \and
  \Rule{Call-Lin}
  { c \in \calM
  }
  { \Gamma \vdash \mi{call}_\m{lin}[Y_i;c] \Downarrow \mi{call}_\m{lin}[Y_i;c] \mid \emptyset }
\end{mathpar}
\caption{Rules for extracting a closure-free linear regular infinite-tree expression and an associated equation system from a linear regular infinite-tree expression.}
\label{Fi:Appendix:ConstraintExtraction}
\end{figure}

\begin{lemma*}[\cref{Lem:EquivalentToAlgebraicExpressions}]
  For any expression $E \in \m{RegExp}^\infty(\calF,\calK)$, there exists an expression $E' \in \m{RegExp}^\infty(\calF^\alpha,\calK)$ ,
  such that for any $\gamma : \calK \to \calM$ and $\vec{\nu} \in \calM^n$, it holds that $\scrM_\gamma\interp{E}(\vec{\nu}) = \calM_\gamma\interp{E'}(\vec{\nu})$.
\end{lemma*}
\begin{proof}
  By induction on the structure of $E$.
  Below shows the only two non-trivial cases.
  \begin{description}[labelindent=\parindent]
    \item[Case $E = \varepsilon$:] We set $E'$ to $\aone_M$.
    \item[{Case $E = \mi{seq}[\m{act}](E_1)$:}] By induction hypothesis, there exists $E_1'$ such that $\scrM_\gamma\interp{E_1}(\vec{\nu}) = \calM_\gamma\interp{E_1'}(\vec{\nu})$. We then set $E'$ to $\mi{seq}[ \interp{\m{act}}^\scrM ](E_1')$.
  \end{description}
\end{proof}

\begin{lemma}\label{Lem:Appendix:ExtractionSound}
  If $E \in \m{RegExp}^\infty(\calF^\alpha_\m{lin}, \calK)$,
  $\dom(\Theta) = \calZ$,
  $\Gamma \vdash E \Downarrow F \mid \Theta$,
  $\gamma : \calK \to \calM$,
  $\eta : (\mathrm{ran}(\Gamma) \setminus \calK) \to \calM$,
  $\vec{\nu} \in \calM^n$,
  $\iota : \calZ \to \calM$ is the least solution of $\Theta$ with respect to $\gamma \cup \eta$ and $\vec{\nu}$,
  and $(\gamma \cup \eta) \circ \Gamma = \gamma$,
  then
  $\calM_{\gamma}\interp{E}(\vec{\nu}) = \calM_{\gamma \cup \eta \cup \iota}\interp{F}(\vec{\nu})$.
\end{lemma}
\begin{proof}
  We proceed by induction on the derivation of $\Gamma \vdash E \Downarrow F \mid \Theta$.
  \begin{description}[labelindent=\parindent]
    \item[Case (\textsc{Const}):]
    \[
    \inferrule{ c \in \calM }{ \Gamma \vdash c \Downarrow c \mid \emptyset }
    \]
    
    We have $\calM_{\gamma}\interp{c}(\vec{\nu}) = c$.
    
    We have $\calM_{\gamma \cup \eta \cup \iota}\interp{c}(\vec{\nu}) = c$.
    
    Then we conclude this case. 
       
    \item[Case (\textsc{Seq}):]
    \[
    \inferrule
    { c \in \calM \\ \Gamma \vdash E_1 \Downarrow F_1 \mid \Theta }
    { \Gamma \vdash \mi{seq}[c](E_1) \Downarrow \mi{seq}[c](F_1) \mid \Theta }
    \]
    
    By induction hypothesis, we know $\calM_{\gamma}\interp{E_1}(\vec{\nu}) = \calM_{\gamma \cup \eta \cup \iota}\interp{F_1}(\vec{\nu})$.
    
    We have $\calM_{\gamma}\interp{\mi{seq}[c](E_1)}(\vec{\nu}) = c \otimes_M \calM_{\gamma}\interp{E_1}(\vec{\nu})$.
    
    We have $\calM_{\gamma \cup \eta \cup \iota}\interp{\mi{seq}[c](F_1)}(\vec{\nu}) = c \otimes_M \calM_{\gamma\cup\eta\cup\iota}\interp{F_1}(\vec{\nu})$.
    
    Then we conclude this case.
    
    \item[Case (\textsc{Cond}):]
    \[
    \inferrule
    { \Gamma \vdash E_1 \Downarrow F_1 \mid \Theta_1 \\
      \Gamma \vdash E_2 \Downarrow F_2 \mid \Theta_2
    }
    { \Gamma \vdash \mi{cond}[\varphi](E_1,E_2) \Downarrow \mi{cond}[\varphi](F_1,F_2) \mid \Theta_1 \cup \Theta_2 }
    \]
    
    Let $\iota_1 \defeq \iota|_{\dom(\Theta_1)}$ and $\iota_2 \defeq \iota|_{\dom(\Theta_2)}$. Then $\iota = \iota_1 \cup \iota_2$.
    
    By the premise, we know that $\Theta_1$ and $\Theta_2$ are two independent sets of equations, given $\gamma \cup \eta$ and $\vec{\nu}$.
    
    Thus, $\iota_i$ is the least solution of $\Theta_i$ ($i \in \{1,2\}$).
    
    By induction hypothesis, we know $\calM_{\gamma}\interp{E_1}(\vec{\nu}) = \calM_{\gamma \cup \eta \cup \iota_1}\interp{F_1}(\vec{\nu})$.
    
    By induction hypothesis, we know $\calM_{\gamma}\interp{E_2}(\vec{\nu}) = \calM_{\gamma \cup \eta \cup \iota_2}\interp{F_2}(\vec{\nu})$.
    
    We have $\calM_{\gamma}\interp{\mi{cond}[\varphi](E_1,E_2)}(\vec{\nu}) = \calM_{\gamma}\interp{E_1}(\vec{\nu}) \gcho{\varphi}_M \calM_{\gamma}\interp{E_2}(\vec{\nu})$.
    
    We have $\calM_{\gamma\cup\eta\cup\iota}\interp{\mi{cond}[\varphi](F_1,F_2)}(\vec{\nu}) = \calM_{\gamma\cup\eta\cup\iota_1}\interp{F_1}(\vec{\nu}) \gcho{\varphi}_M \calM_{\gamma\cup\eta\cup\iota_2}\interp{F_2}(\vec{\nu})$.
    
    Then we conclude this case.
    
    (Cases (\textsc{Prob}), (\textsc{Ndet}), (\textsc{Add}), and (\textsc{Sub}) are similar to this case.)
    
    \item[Case (\textsc{Hole}):]
    \[
    \inferrule{ }{ \Gamma \vdash \square \Downarrow \Gamma(\square) \mid \emptyset }
    \]
    
    We have $\calM_\gamma\interp{\square}(\vec{\nu}) = \gamma(\square)$.
    
    We have $\calM_{\gamma\cup\eta\cup\iota}\interp{\Gamma(\square)}(\vec{\nu}) = ((\gamma \cup \eta) \circ \Gamma)(\square)$.
    
    We conclude by the assumption $(\gamma \cup \eta) \circ \Gamma = \gamma$.
    
    \item[Case (\textsc{Concatenation}):]
    \[
    \inferrule
    { \Gamma[\square \mapsto \square] \vdash E_1 \Downarrow F_1 \mid \Theta_1 \\
      \Gamma \vdash E_2 \Downarrow F_2 \mid \Theta_2
    }
    { \Gamma \vdash E_1 \cdot_{\square} E_2 \Downarrow F_1 \cdot_\square F_2 \mid \Theta_1 \cup \Theta_2 }
    \]
    
    Let $\iota_1 \defeq \iota|_{\dom(\Theta_1)}$ and $\iota_2 \defeq \iota|_{\dom(\Theta_2)}$. Then $\iota = \iota_1 \cup \iota_2$.
    
    By the premise, we know that $\Theta_1$ and $\Theta_2$ are two independent sets of equations, given $\gamma \cup \eta$ and $\vec{\nu}$.
    
    Thus, $\iota_i$ is the least solution of $\Theta_i$ ($i \in \{1,2\}$).

    By induction hypothesis, we know $\calM_\gamma\interp{E_2}(\vec{\nu}) = \calM_{\gamma \cup \eta \cup \iota_2}\interp{F_2}(\vec{\nu})$.
    
    Let $\gamma' \defeq \gamma[\square \mapsto \calM_\gamma\interp{E_2}(\vec{\nu})]$.
    We still have $((\gamma' \cup \eta)) \circ (\Gamma[\square \mapsto \square]) = \gamma'$.
    
    By induction hypothesis, we know $\calM_{\gamma'}\interp{E_1}(\vec{\nu}) = \calM_{\gamma' \cup \eta \cup \iota_1}\interp{E_1}(\vec{\nu})$.
    
    We have $\calM_{\gamma}\interp{E_1 \cdot_\square E_2}(\vec{\nu}) = \calM_{\gamma'}\interp{E_1}(\vec{\nu})$.
    
    We have $\calM_{\gamma \cup \eta \cup \iota}\interp{F_1 \cdot_\square F_2}(\vec{\nu}) = \calM_{(\gamma \cup \eta \cup \iota_1)[\square \mapsto \calM_{\gamma \cup \eta \cup \iota_2}\interp{F_2}(\vec{\nu}) ]}\interp{F_1}(\vec{\nu}) = \calM_{\gamma' \cup \eta \cup \iota_1}\interp{F_1}(\vec{\nu})$.
    
    Then we conclude this case.
    
    \item[Case (\textsc{Closure}):]
    \[
    \inferrule
    { Z~\m{fresh} \\ \Gamma[\square \mapsto Z] \vdash E_1 \Downarrow F_1 \mid \Theta_1 }
    { \Gamma \vdash (E_1)^{\infty_\square} \Downarrow Z \mid \Theta_1 \cup \{Z = F_1\} }
    \]
    
    Let $\theta \defeq \iota(Z)$. Let $\iota_1 \defeq \iota|_{\dom(\Theta_1)}$.
    
    By assumption, $\iota$ is the least solution of $\Theta_1 \cup \{ Z = F_1 \}$, given $\gamma \cup \eta$ and $\vec{\nu}$.
    
    Thus, $\theta = \calM_{\gamma \cup \eta \cup \iota}\interp{F_1}(\vec{\nu})$.
    Let $\eta' \defeq \eta[Z \mapsto \theta]$. Then $\theta = \calM_{\gamma \cup \eta' \cup \iota_1}\interp{F_1}(\vec{\nu})$.
    
    Let $\gamma' \defeq \gamma[\square \mapsto \theta]$. We have $(\gamma' \cup \eta') \circ \Gamma[\square \mapsto Z] = \gamma'$.
    
    Also, $\iota_1$ is the least solution of $\Theta_1$, given $\gamma' \cup \eta'$ and $\vec{\nu}$.
    
    By induction hypothesis, we know $\calM_{\gamma'}\interp{E_1}(\vec{\nu}) = \calM_{\gamma' \cup \eta' \cup \iota_1}\interp{F_1}(\vec{\nu}) = \theta$.
    
    We have $\calM_\gamma\interp{(E_1)^{\infty_\square}}(\vec{\nu}) = \rho$ where $\rho$ is the least fixed-point
    of $\lambda \rho.\; \calM_{\gamma[\square \mapsto \rho]}\interp{E_1}(\vec{\nu})$.
    
    Indeed, $\theta$ is a fixed-point of $\lambda \rho.\; \calM_{\gamma[\square \mapsto \rho]}\interp{E_1}(\vec{\nu})$.
    It remains to show $\theta$ is the least one.
    
    Again, by assumption, $\theta$ is the least solution of $Z=F_1$, given $\gamma[\square \mapsto \rho] \cup \eta \cup \iota_1$ and $\vec{\nu}$.
    
    Suppose that $\rho \sqsubset_M \theta$.
    Let $\iota_1'$ be the least solution of $\Theta_1$ given $\gamma' \cup \eta[Z \mapsto \rho]$ and $\vec{\nu}$.
    By $\omega$-continuity, we know that $\iota_1' \aord_M \iota_1$.
        
    Thus, by induction hypothesis, we know $\calM_{\gamma[\square \mapsto \rho]}\interp{E_1}(\vec{\nu}) =
    \calM_{\gamma[\square \mapsto \rho] \cup \eta[Z \mapsto \rho] \cup \iota_1'}\interp{F_1}(\vec{\nu})$.
    
    But $\rho = \calM_{\gamma[\square \mapsto \rho]}\interp{E_1}(\vec{\nu})$.
    Thus $\rho$ is a solution of $Z=F_1$ given $\rho[\square \mapsto \rho] \cup \eta \cup \iota_1'$ and $\vec{\nu}$.
    
    Therefore, let $\iota' \defeq \iota_1' \cup \{ Z \mapsto \rho \}$, then $\iota'$ is a solution of $\Theta_1 \cup \{ Z = F_1 \}$ given $\gamma \cup \eta$ and $\vec{\nu}$.
    
    But then we have $\iota' \sqsubset_M \iota$, which contradicts the assumption that $\iota$ is the least solution.
    
    Then we conclude this case.
    
    \item[Case (\textsc{Call-Lin}):]
    \[
    \inferrule
    { c \in \calM }
    { \Gamma \vdash \mi{call}_\m{lin}[Y_i;c] \Downarrow \mi{call}_\m{lin}[Y_i;c] \mid \emptyset }
    \]
    
    We have $\calM_\gamma\interp{\mi{call}_\m{lin}[Y_i;c]}(\vec{\nu}) = \nu_i \otimes_M c$.
    
    We have $\calM_{\gamma\cup\eta\cup\iota}\interp{\m{call}_\m{lin}[Y_i;c]}(\vec{\nu}) = \nu_i \otimes_M c$.
    
    Then we conclude this case.
  \end{description}
\end{proof}

\begin{theorem*}[\cref{The:LinearRecursionSolvingSound}]
  The linear-equation-solving method presented in \cref{Se:SolvingLinearEquations} computes $\lfp_{\vec{\azero}_M}^{\aord_M} \lambda \vec{\theta}.\; \tuple{\calM_\gamma\interp{E_{Y_1}}(\vec{\theta}), \cdots, \calM_\gamma\interp{E_{Y_n}}(\vec{\theta})}$.
\end{theorem*}
\begin{proof}
  By assumption, we have $\{ \square \mapsto \square \}_{\square \in \calK} \vdash E_{Y_i} \Downarrow F_{Y_i} \mid \Theta_i$ for each $i$.
  
  Let $\m{solve}_\calK(\{Y_i=F_{Y_i}\}_{i=1}^n, \bigcup_{i=1}^n \Theta_i, \gamma)$ return $\iota : \dom(\bigcup_{i=1}^n \Theta_i) \to \calM$ and $\vec{\nu} \in \calM^n$, i.e., the least solution of $\{Y_i=F_{Y_i}\}_{i=1}^n \cup \bigcup_{i=1}^n \Theta_i$.
  
  The method then returns $\calM_{\gamma \cup \iota}\interp{F_{Y_i}}(\vec{\nu})$ for the solution of $Y_i$, but
  because $\m{solve}$ returns the least solution, the value coincides with $\nu_i$.
  
  Thus, here it is sufficient to prove $\vec{\nu}$ is the least fixed-point of $\lambda \vec{\theta}.\; \tuple{\calM_\gamma\interp{E_{Y_1}}(\vec{\theta}), \cdots, \calM_\gamma\interp{E_{Y_n}}(\vec{\theta})}$.
  
  If we fix $\vec{\theta} \in \calM^n$, then the sets $\Theta_i$'s are independent sets of equations under $\gamma$ and $\vec{\theta}$.
  Therefore, for each $i$, by \cref{Lem:Appendix:ExtractionSound} (set $\eta$ to $\{\}$),
  we know that $\calM_{\gamma}\interp{E_{Y_i}}(\vec{\theta}) = \calM_{\gamma \cup \iota}\interp{F_{Y_i}}(\vec{\theta})$.
  
  Therefore, $\vec{\nu}$ is a fixed point. It remains to show it is the least one.
  
  Suppose there is another fixed-point $\vec{\rho}$ that satisfies $\vec{\rho} \sqsubset_M \vec{\nu}$.
  
  For each $i$,
  Let $\iota_i'$ be the least solution of $\Theta_i$ given $\gamma$ and $\vec{\rho}$.
  By $\omega$-continuity, we know that $\iota_i' \sqsubseteq \iota_i \defeq \iota|_{\dom(\Theta_i)}$.
  
  Define $\iota' \defeq \bigcup_{i=1}^m \iota_i'$. Then $\iota' \aord_M \iota$.
  
  Then, again, by \cref{Lem:Appendix:ExtractionSound} (set $\eta$ to $\{\}$), for each $i$, we know that
  $\calM_{\gamma}\interp{E_{Y_i}}(\vec{\rho}) = \calM_{\gamma \cup \iota'}\interp{F_{Y_i}}(\vec{\rho})$.
  
  But $\rho_i = \calM_{\gamma}\interp{E_{Y_i}}(\vec{\rho})$ for each $i$.
  Thus $\iota'$ and $\vec{\rho}$ is a solution to $\{Y_i = F_{Y_i}\}_{i=1}^n \cup \bigcup_{i=1}^n \Theta_i$ under $\gamma$.
  
  But then we obtain $(\vec{\rho}, \iota') \sqsubset (\vec{\nu}, \iota)$, which contradicts the assumption
  that $(\vec{\nu}, \iota)$ is the least solution.
  
  Then we conclude the proof.
\end{proof}

\begin{lemma*}[\cref{Lem:DiffIsLinear}]
  If $\vec{f}$ is a vector of regular infinite-tree expressions in $\m{RegExp}^\infty(\calF^\alpha,\calK)$,
  it holds that
  $\calD \vec{f}_\gamma|_{\vec{\nu}}(\vec{Y})$ is a vector of regular infinite-tree expressions in $\m{RegExp}^\infty(\calF_\m{lin}^\alpha, \calK)$.
\end{lemma*}
\begin{proof}
  Straightforward because we only use the alphabet $\calF^\alpha_\m{lin}$ in the definition of differentials (see \cref{De:RegularHyperPathDifferential}).
\end{proof}

\begin{lemma}\label{Lem:Appendix:TaylorUni}
  Let $\vec{f}(\vec{X})$ be a vector of regular infinite-tree expression in $\m{RegExp}^\infty(\calF^\alpha,\calK)$,
  $\gamma,\gamma' : \calK \to \calM$ be two hole-valuations,
  and $\vec{u}, \vec{\nu}$ be two vectors in $\calM^n$.
  We have
  \[
  \vec{f}_\gamma(\vec{u}) \oplus_M ( \calD \vec{f}_\gamma|_{\vec{u}} (\vec{\nu}) )_{\gamma'} \aord_M f_{\gamma \oplus_M \gamma'}(\vec{u} \oplus_M \vec{\nu}).
  \]
\end{lemma}
\begin{proof}
  It suffices to show the inequality for each component separately, so W.L.O.G. let $\vec{f}(\vec{X}) = f(\vec{X})$ be a regular infinite-tree expression.
  We then proceed by induction on the structure of $f$.
  \begin{description}[labelindent=\parindent]
    \item[Case $f=c$:]
    We have $\calD f_\gamma|_{\vec{u}}(\vec{\nu}) = \azero_M$ and thus $c \oplus_M \azero_M \aord_M c$.
    
    \item[{Case $f=\mi{seq}[c](g)$:}]
    We have $f_\gamma(\vec{x}) = c \otimes_M g_\gamma(\vec{x})$ and $(\calD f_\gamma|_{\vec{x}}(\vec{y}))_{\gamma'} = c \otimes_M (\calD g_\gamma|_{\vec{x}}(\vec{y}))_{\gamma'}$. Thus
    \begin{align*}
      f_{\gamma \oplus_M \gamma'}(\vec{u} \oplus_M \vec{\nu}) &
      = c \otimes_M g_{\gamma \oplus_M \gamma'}(\vec{u} \oplus_M \vec{\nu}) \\
      & \sqsupseteq_M c \otimes_M (g_{\gamma}(\vec{u}) \oplus_M (\calD g_\gamma|_{\vec{u}}(\vec{\nu}))_{\gamma'}) & \text{(induction hypothesis)} \\
      & = (c \otimes_M g_\gamma(\vec{u})) \oplus_M (c \otimes_M (\calD g_\gamma|_{\vec{u}}(\vec{\nu}))_{\gamma'}) \\
      & = f_\gamma(\vec{u}) \oplus_M (\calD f_\gamma|_{\vec{u}}(\vec{\nu}))_{\gamma'} .
    \end{align*}
    
    \item[{Case $f=\mi{call}[X_i](g)$:}]
    We have $f_\gamma(\vec{x}) = x_i \otimes_M g_\gamma(\vec{x})$ and
    \begin{align*}
      \calD f_\gamma|_{\vec{x}}(\vec{Y})
      & = \calD_{X_i} f_\gamma|_{\vec{x}}(\vec{Y}) \oplus \bigoplus_{j \neq i} \calD_{X_j} f_\gamma|_{\vec{x}}(\vec{Y}) \\
      & = ( \mi{call}_\m{lin}[Y_i; g_\gamma(\vec{x})] \oplus \mi{seq}[x_j]( \calD_{X_i} g_\gamma|_{\vec{x}}(\vec{Y}))) \oplus \bigoplus_{j \neq i} \mi{seq}[x_i](\calD_{X_j} g|_{\vec{x}}(\vec{Y})), \\
      (\calD f_\gamma|_{\vec{x}}(\vec{y}))_{\gamma'}
      & = ((y_i \otimes_M g_\gamma(\vec{x})) \oplus_M (x_j \otimes_M (\calD_{X_i} g_\gamma|_{\vec{x}}(\vec{y}))_{\gamma'})) \oplus_M \bigoplus_{M, j \neq i} ( x_j \otimes_M (\calD_{X_j} g_\gamma|_{\vec{x}}(\vec{y}))_{\gamma'}) \\
      & = (y_i \otimes_M g_\gamma(\vec{x})) \oplus_M (x_j \otimes_M (\calD g_{\gamma}|_{\vec{x}}(\vec{y}))_{\gamma'}).
    \end{align*}
    Thus
    \begin{align*}
      f_{\gamma \oplus_M \gamma'}(\vec{u} \oplus_M \vec{\nu})
      & = (u_i \oplus_M \nu_i) \otimes_M g_{\gamma \oplus_M \gamma'}(\vec{u} \oplus_M \vec{\nu}) \\
      & \sqsupseteq_M (u_i \oplus_M \nu_i) \otimes_M (g_\gamma(\vec{u}) \oplus_M (\calD g_{\gamma}|_{\vec{u}}(\vec{\nu}))_{\gamma'}) \qquad\quad \text{(induction hypothesis)} \\
      & = (u_i \otimes_M g_\gamma(\vec{u})) \oplus_M (\nu_i \otimes_M g_{\gamma}(\vec{u})) \oplus_M ((u_i \oplus_M \nu_i) \otimes_M (\calD g_{\gamma}|_{\vec{u}}(\vec{\nu}))_{\gamma'}) \\
      & \sqsupseteq_M f_\gamma(\vec{u}) \oplus_M ( (\nu_i \otimes_M g_{\gamma}(\vec{u})) \oplus_M (u_i \otimes_M (\calD g_{\gamma}|_{\vec{u}}(\vec{\nu}))_{\gamma'}) ) \\
      & = f_\gamma(\vec{u}) \oplus_M (\calD f_\gamma |_{\vec{u}}(\vec{\nu}))_{\gamma'}.
    \end{align*}
    
    \item[{Case $f= \mi{cond}[\varphi](g,h)$:}]
    We have $f_\gamma(\vec{x}) = g_\gamma(\vec{x}) \gcho{\varphi}_M h_\gamma(\vec{x})$
    and $(\calD f_\gamma|_{\vec{x}}(\vec{y}))_{\gamma'} = (\calD g_\gamma|_{\vec{x}}(\vec{y}))_{\gamma'} \gcho{\varphi}_M (\calD h_\gamma|_{\vec{x}}(\vec{y}))_{\gamma'}$. Thus
    \begin{align*}
      f_{\gamma \oplus_M \gamma'}(\vec{u} \oplus_M \vec{\nu})
      & = g_{\gamma \oplus_M \gamma'}(\vec{u} \oplus_M \vec{\nu}) \gcho{\varphi}_M h_{\gamma \oplus_M \gamma'}(\vec{u} \oplus_M \vec{\nu}) \\
      & \sqsupseteq_M ( g_\gamma(\vec{u}) \oplus_M (\calD g_\gamma|_{\vec{u}}(\vec{\nu}))_{\gamma'}) \gcho{\varphi}_M
      ( h_\gamma(\vec{u}) \oplus_M (\calD h_\gamma|_{\vec{u}}(\vec{\nu}))_{\gamma'}) & \text{(induction hypothesis)} \\
      & = ( g_\gamma(\vec{u}) \gcho{\varphi}_M h_\gamma(\vec{u}) ) \oplus_M ((\calD h_\gamma|_{\vec{u}}(\vec{\nu}))_{\gamma'} \gcho{\varphi}_M (\calD h_\gamma|_{\vec{u}}(\vec{\nu}))_{\gamma'} ) & \text{(\cref{De:OmegaContinuousPreMarkovAlgebra})} \\
      & = f_\gamma(\vec{u}) \oplus_M (\calD f_\gamma|_{\vec{u}}(\vec{\nu}))_{\gamma'}.
    \end{align*}
    
    \item[{Case $f = \mi{prob}[p](g,h)$:}]
    We have $f_\gamma(\vec{x}) = g_\gamma(\vec{x}) \pcho{p}_M h_\gamma(\vec{x})$
    and $(\calD f_\gamma|_{\vec{x}}(\vec{y}))_{\gamma'} = (\calD g_\gamma|_{\vec{x}}(\vec{y}))_{\gamma'} \pcho{p}_M (\calD h_\gamma|_{\vec{x}}(\vec{y}))_{\gamma'}$. Thus
    \begin{align*}
      f_{\gamma \oplus_M \gamma'}(\vec{u} \oplus_M \vec{\nu})
      & = g_{\gamma \oplus_M \gamma'}(\vec{u} \oplus_M \vec{\nu}) \pcho{p}_M h_{\gamma \oplus_M \gamma'}(\vec{u} \oplus_M \vec{\nu}) \\
      & \sqsupseteq_M ( g_\gamma(\vec{u}) \oplus_M (\calD g_\gamma|_{\vec{u}}(\vec{\nu}))_{\gamma'}) \pcho{p}_M
      ( h_\gamma(\vec{u}) \oplus_M (\calD h_\gamma|_{\vec{u}}(\vec{\nu}))_{\gamma'}) & \text{(induction hypothesis)} \\
      & = ( g_\gamma(\vec{u}) \pcho{p}_M h_\gamma(\vec{u}) ) \oplus_M ((\calD h_\gamma|_{\vec{u}}(\vec{\nu}))_{\gamma'} \pcho{p}_M (\calD h_\gamma|_{\vec{u}}(\vec{\nu}))_{\gamma'} ) & \text{(\cref{De:OmegaContinuousPreMarkovAlgebra})} \\
      & = f_\gamma(\vec{u}) \oplus_M (\calD f_\gamma|_{\vec{u}}(\vec{\nu}))_{\gamma'}.
    \end{align*}
    
    \item[{Case $f = \mi{ndet}(g,h)$:}]
    We have $f_\gamma(\vec{x}) = g_\gamma(\vec{x}) \dashcup_M h_\gamma(\vec{x})$ and
    $(\calD f_\gamma|_{\vec{x}}(\vec{y}))_{\gamma'} = ( ( g_\gamma(\vec{x}) \oplus_M (\calD g_\gamma|_{\vec{x}}(\vec{y}))_{\gamma'} ) \dashcup_M ( h_\gamma(\vec{x}) \oplus_M (\calD h_\gamma|_{\vec{x}}(\vec{y}))_{\gamma'} ) ) \ominus_M f_\gamma(\vec{x})$.
    Thus
    \begin{align*}
      f_{\gamma \oplus_M \gamma'}(\vec{u} \oplus_M \vec{\nu})
      & = g_{\gamma \oplus_M \gamma'}(\vec{u} \oplus_M \vec{\nu}) \dashcup_M h_{\gamma \oplus_M \gamma'}(\vec{u} \oplus_M \vec{\nu}) \\
      & \sqsupseteq_M ((g_\gamma(\vec{u}) \oplus_M (\calD g_\gamma|_{\vec{u}}(\vec{\nu}))_{\gamma'} ) \dashcup_M (h_\gamma(\vec{u}) \oplus_M (\calD h_\gamma|_{\vec{u}}(\vec{\nu}))_{\gamma'} ) \qquad \text{(induction hypothesis)} \\
      & = f_\gamma(\vec{u}) \oplus_M ( ( ((g_\gamma(\vec{u}) \oplus_M (\calD g_\gamma|_{\vec{u}}(\vec{\nu}))_{\gamma'} ) \dashcup_M (h_\gamma(\vec{u}) \oplus_M (\calD h_\gamma|_{\vec{u}}(\vec{\nu}))_{\gamma'} ) ) \ominus_M f_\gamma(\vec{u}) ) \\
      & = f_\gamma(\vec{u}) \oplus_M (\calD f_\gamma|_{\vec{u}}(\vec{\nu}))_{\gamma'}.
    \end{align*}
    
    \item[{Case $f = \square$:}]
    We have $f_\gamma(\vec{x}) = \gamma(\square)$ and $(\calD f_\gamma|_{\vec{x}}(\vec{y}))_{\gamma'} = \gamma'(\square)$.
    Thus
    \begin{align*}
      f_{\gamma \oplus_M \gamma'}(\vec{u} \oplus_M \vec{\nu})
      & = \gamma(\square) \oplus_M \gamma'(\square) \\
      & = f_\gamma(\vec{u}) \oplus_M (\calD f_\gamma|_{\vec{u}}(\vec{\nu}))_{\gamma'}.
    \end{align*}
    
    \item[{Case $f = g \cdot_{\square} h$:}]
    We have $f_\gamma(\vec{x}) = g_{\gamma[\square \mapsto h_\gamma(\vec{x})]}(\vec{x})$
    and $(\calD f_\gamma|_{\vec{x}}(\vec{y}))_{\gamma'} = (\calD g_{\gamma[\square \mapsto h_\gamma(\vec{x})]}|_{\vec{x}}(\vec{y}) )_{\gamma'[\square \mapsto (\calD h_\gamma|_{\vec{x}}(\vec{y}))_{\gamma'} ]}  $.
    Thus
    \begin{align*}
      f_{\gamma \oplus_M \gamma'}(\vec{u} \oplus_M \vec{\nu})
      & = g_{(\gamma \oplus_M \gamma')[\square \mapsto h_{\gamma \oplus_M \gamma'}(\vec{u} \oplus_M \vec{\nu})]}( \vec{u} \oplus_M \vec{\nu}) \\
      & \sqsupseteq_M g_{(\gamma \oplus_M \gamma')[ \square \mapsto h_\gamma(\vec{u}) \oplus_M (\calD h_\gamma|_{\vec{u}}(\vec{\nu}))_{\gamma'} ] }( \vec{u} \oplus_M \vec{\nu} ) & \text{(induction hypothesis)} \\
      & = g_{ \gamma[\square \mapsto h_\gamma(\vec{u})] \oplus_M \gamma'[ \square \mapsto (\calD h_\gamma|_{\vec{u}}(\vec{\nu}))_{\gamma'} ] } (\vec{u} \oplus_M \vec{\nu} ) \\
      & \sqsupseteq_M g_{\gamma[\square \mapsto h_\gamma(\vec{u})]}(\vec{u}) \oplus_M  ( \calD g_{ \gamma[\square \mapsto h_\gamma(\vec{u})] }|_{\vec{u}}(\vec{\nu}) )_{ \gamma'[ \square \mapsto (\calD h_\gamma|_{\vec{u}}(\vec{\nu}))_{\gamma'} ] } & \text{(induction hypothesis)} \\
      & = f_\gamma(\vec{u}) \oplus_M (\calD f_\gamma |_{\vec{u}}(\vec{\nu}))_{\gamma'}.
    \end{align*}
    
    \item[{Case $f = (g)^{\infty_\square}$:}]
    We have $f_\gamma(\vec{x}) = \lfp_{\azero_M} \lambda \theta.\; g_{\gamma[\square \mapsto \theta]}(\vec{x})$
    and $(\calD f_\gamma|_{\vec{x}}(\vec{y}))_{\gamma'} = \lfp_{\azero_M} \lambda \theta.\; (\calD g_{\gamma[\square \mapsto f_\gamma(\vec{x})]}|_{\vec{x}}(\vec{y}))_{\gamma'[\square \mapsto \theta]} $.
    Let $\rho \defeq \lfp_{\azero_M} \lambda \rho.\; g_{\gamma[\square \mapsto \rho]}(\vec{u})$.
    Then by $\omega$-continuity, we have $\rho \aord_M \lfp_{\azero_M} \lambda \theta.\; g_{(\gamma \oplus_M \gamma')[\square \mapsto \theta]}(\vec{u} \oplus \vec{\nu})$.
    Thus
    \begin{align*}
      f_{\gamma \oplus_M \gamma'}(\vec{u} \oplus_M \vec{\nu})
      & = \lfp_{\azero_M} \lambda \theta.\; g_{(\gamma \oplus_M \gamma')[\square \mapsto \theta]}(\vec{u} \oplus \vec{\nu}) \\
      & = \lfp_\rho \lambda \theta.\; g_{\gamma[\square \mapsto \rho] \oplus_M \gamma[\square \mapsto \theta \ominus_M \rho]}(\vec{u} \oplus \vec{\nu}) \\
      & \sqsupseteq_M \lfp_\rho \lambda \theta. \; ( g_{\gamma[\square \mapsto \rho]}(\vec{u}) \oplus_M ( \calD g_{\gamma[\square \mapsto \rho]}|_{\vec{u}}(\vec{\nu}) )_{\gamma'[\square \mapsto \theta \ominus_M \rho]} ) \qquad \text{(induction hypothesis)} \\
      & = \lfp_\rho \lambda \theta. (\rho \oplus_M ( \calD g_{\gamma[\square \mapsto \rho]}|_{\vec{u}}(\vec{\nu}) )_{\gamma'[\square \mapsto \theta \ominus_M \rho]}) \\
      & = \bigsqcup_{n\in\bbN} L^n(\rho), \\
      & \text{where}~L \defeq \lambda \theta.\; (\rho \oplus_M H(\theta \ominus_M \rho) ), H \defeq \lambda\theta.\; ( \calD g_{\gamma[\square \mapsto \rho]}|_{\vec{u}}(\vec{\nu}) )_{\gamma'[\square \mapsto \theta]}.
    \end{align*}
    On the other hand, we have
    \begin{align*}
      f_\gamma(\vec{u}) \oplus_M (\calD f_\gamma|_{\vec{u}}(\vec{\nu}))_{\gamma'}
      & = \rho \oplus_M ( \lfp_{\azero_M} \lambda \theta.\; (\calD g_{\gamma[\square \mapsto f_\gamma(\vec{u})]}|_{\vec{u}}(\vec{\nu}))_{\gamma'[\square \mapsto \theta]} ) \\
      & = \rho \oplus_M ( \lfp_{\azero_M} \lambda \theta.\; (\calD g_{\gamma[\square \mapsto \rho]}|_{\vec{u}}(\vec{\nu}))_{\gamma'[\square \mapsto \theta]} ) \\
      & = \rho \oplus_M \bigsqcup_{n \in \bbN} H^n(\azero_M).
    \end{align*}
    We then prove by induction on $n$ that $L^n(\rho) = \rho \oplus_M H^n(\azero_M)$.
    
    When $n=0$, we have $L^0(\rho) = \rho$ and $\rho \oplus_M H^0(\rho) = \rho \oplus_M \azero_M = \rho$.
    
    For the induction step, we have
    \begin{align*}
      L^{n+1}(\rho) & = L(L^n(\rho)) \\
      & = \rho \oplus_M H(L^n(\rho) \ominus_M \rho) \\
      & = \rho \oplus_M H((\rho \oplus_M H^n(\azero_M)) \ominus_M \rho) & \text{(induction hypothesis)} \\
      & = \rho \oplus_M H(H^n(\azero_M)) \\
      & = \rho \oplus_M H^{n+1}(\azero_M).
    \end{align*}
  \end{description}
\end{proof}

\begin{lemma}\label{Lem:Appendix:Aux}
  Let $\vec{f}_{\{\}}(\vec{x}) \sqsupseteq_M \vec{x}$.
  For all $d \ge 0$, there exists a vector $\vec{e}^{(d)}(\vec{x})$ such that
  $\vec{f}^d_{\{\}}(\vec{x}) \oplus_M \vec{e}^{(d)}(\vec{x}) = \vec{f}^{d+1}_{\{\}}(\vec{x})$ and
  \[
  \vec{e}^{(d)}(\vec{x}) \sqsupseteq_M ( \calD \vec{f}_{\{\}}|_{\vec{f}^{d-1}_{\{\}}(\vec{x})}( (\calD \vec{f}_{\{\}} \mid_{\vec{f}^{d-2}_{\{\}}(\vec{x})}( \cdots (\calD \vec{f}_{\{\}}|_{\vec{x}}(\vec{e}^{(0)}(\vec{x}))_{\{\}} ) \cdots )))_{\{\}} )_{\{\}}
  \sqsupseteq_M (\calD \vec{f}_{\{\}}|_{\vec{x}}^d(\vec{e}^{(0)}(\vec{x})))_{\{\}}.
  \]
\end{lemma}
\begin{proof}
  By induction on $d$.
  For $d=0$, we set $\vec{e}^{(0)}(\vec{x})$ to be $\vec{f}_{\{\}}(\vec{x}) \ominus_M \vec{x}$.
  Let $d \ge 0$.
  \begin{align*}
    \vec{f}^{d+2}_{\{\}}(\vec{x}) & = \vec{f}_{\{\}}( \vec{f}^{d+1}_{\{\}} (\vec{x}) ) \\
    & = \vec{f}_{\{\}}( \vec{f}^d_{\{\}}(\vec{x}) \oplus_M \vec{e}^{(d)}(\vec{x}) ) & \text{(induction hypothesis)} \\ 
    & \sqsupseteq_M \vec{f}_{\{\}} ( \vec{f}_{\{\}}^d(\vec{x}) ) \oplus_M ( \calD \vec{f}_{\{\}}|_{\vec{f}_{\{\}}^d(\vec{x})}( \vec{e}^{(d)}(\vec{x}) ))_{\{\}} & \text{(\cref{Lem:Appendix:TaylorUni})} \\
    & \sqsupseteq_M \vec{f}_{\{\}}^{d+1}(\vec{x}) \\
    & \quad \oplus_M ( \calD \vec{f}_{\{\}}|_{\vec{f}^{d}_{\{\}}(\vec{x})}( (\calD \vec{f}_{\{\}} \mid_{\vec{f}^{d-1}_{\{\}}(\vec{x})}( \cdots (\calD \vec{f}_{\{\}}|_{\vec{x}}(\vec{e}^{(0)}(\vec{x}))_{\{\}} ) \cdots )))_{\{\}} )_{\{\}}. & \text{(induction hypothesis)}
  \end{align*}
  Therefore, there exists an $\vec{e}^{(d+1)}(\vec{x}) \sqsupseteq_M ( \calD \vec{f}_{\{\}}|_{\vec{f}^{d}_{\{\}}(\vec{x})}( (\calD \vec{f}_{\{\}} \mid_{\vec{f}^{d-1}_{\{\}}(\vec{x})}( \cdots (\calD \vec{f}_{\{\}}|_{\vec{x}}(\vec{e}^{(0)}(\vec{x}))_{\{\}} ) \cdots )))_{\{\}} )_{\{\}}$.
  Since $(\calD \vec{f}_{\{\}}|_{\vec{y}}(\cdot))_{\{\}}$ is monotone in $\vec{y}$
  and $\vec{x} \aord_M \vec{f}(\vec{x}) \aord_M \vec{f}^2(\vec{x}) \aord_M \cdots$,
  the second inequality also holds.
\end{proof}

\begin{lemma}\label{Lem:Appendix:StarSolve}
  The least solution of $\calD \vec{f}_\gamma|_{\vec{u}}(\vec{Y}) \oplus \vec{\nu} = \vec{Y}$
  under $\gamma'$ is $((\calD \vec{f}|_{\vec{u}}(\cdot))_{\gamma'})^{\circledast}(\vec{\nu})$,
  where $\vec{g}^\circledast(\vec{\nu}) \defeq \bigoplus_{M,n \in \bbN} \vec{g}^n(\vec{\nu})$.
\end{lemma}
\begin{proof}
  Set $\vec{g}(\vec{Y}) \defeq \calD \vec{f}_\gamma|_{\vec{u}}(\vec{Y}) \oplus \vec{\nu}$.
  Thus $\vec{g}_{\gamma'}(\vec{x}) = (\calD \vec{f}_\gamma|_{\vec{u}}(\vec{x}))_{\gamma'} \oplus_M \vec{\nu}$.
  By Kleene's fixed-point theorem,
  the least solution of $\vec{g}(\vec{Y}) = \vec{Y}$ under $\gamma'$
  is given by $\bigsqcup_{n \in \bbN} \vec{g}^n_{\gamma'}(\vec{\azero}_M) = ((\calD \vec{f}_\gamma|_{\vec{u}}(\cdot))_{\gamma'})^{\circledast}(\vec{\nu})$.
\end{proof}

\begin{theorem*}[\cref{The:NewtonConvergence}]
  Let $\vec{f}$ be a vector of regular infinite-tree expressions in $\m{RegExp}^\infty(\calF^\alpha,\emptyset)$.
  Then the Newton sequence is monotonically increasing, and it converges to the least fixed-point
  as least as fast as the Kleene sequence, i.e., for all $i \in \bbN$, we have
  \[
  \vec{\kappa}^{(i)} \aord_M \vec{\nu}^{(i)} \aord_M \vec{f}_{\{\}}(\vec{\nu}^{(i)}) \aord_M \vec{\nu}^{(i+1)} \aord_M \lfp_{\vec{\azero}_M}^{\aord_M} \vec{f}_{\{\}} = \textstyle\bigsqcup^{\uparrow}_{j \in \bbN} \vec{\kappa}^{(j)},
  \]
  where the Kleene sequence is defined as $\vec{\kappa}^{(j)} \defeq \vec{f}_{\{\}}^j ( \vec{\azero}_M)$ for $j \in \bbN$.
\end{theorem*}
\begin{proof}
  We proceed by induction on $i$. The base case $i=0$ is straightforward. For the induction step, let $i \ge 0$.
  Then we have
  \begin{align*}
    \vec{\kappa}^{(i+1)} & = \vec{f}_{\{\}}(\vec{\kappa}^{(i)}) \\
    & \aord_M \vec{f}_{\{\}}(\vec{\nu}^{(i)}) & \text{(induction: $\vec{\kappa}^{(i)} \aord_M \vec{\nu}^{(i)}$)} \\
    & = \vec{\nu}^{(i)} \oplus_M \vec{\delta}^{(i)} & \text{($\vec{\delta}^{(i)} \defeq \vec{f}_{\{\}}(\vec{\nu}^{(i)}) \ominus_M \vec{\nu}^{(i)}$)} \\
    & \aord_M \vec{\nu}^{(i)} \oplus_M (\calD \vec{f}_{\{\}}|_{\vec{\nu}^{(i)}}(\cdot))_{\{\}}^{\circledast}(\vec{\delta}^{(i)}) & \text{($\vec{\delta} \aord_M \vec{g}^\circledast(\vec{\delta})$)} \\
    & = \vec{\nu}^{(i+1)} & \text{(\cref{Lem:Appendix:StarSolve})} \\
    & = \vec{\nu}^{(i)} \oplus_M \vec{\delta}^{(i)} \oplus (\calD \vec{f}_{\{\}}|_{\vec{\nu}^{(i)}}(\cdot))_{\{\}}( (\calD \vec{f}_{\{\}}|_{\vec{\nu}^{(i)}}(\cdot))_{\{\}}^\circledast( \vec{\delta}^{(i)} )  ) & \text{($\vec{g}^\circledast(\vec{\delta}) = \vec{\delta} \oplus_M \vec{g}(\vec{g}^\circledast(\vec{\delta}))$)} \\
    & = \vec{f}_{\{\}}(\vec{\nu}^{(i)}) \oplus_M (\calD \vec{f}_{\{\}}|_{\vec{\nu}^{(i)}}(\cdot))_{\{\}}( (\calD \vec{f}_{\{\}}|_{\vec{\nu}^{(i)}}(\cdot))_{\{\}}^\circledast( \vec{\delta}^{(i)} )  ) \\
    & \aord_M \vec{f}_{\{\}}( \vec{\nu}^{(i)} \oplus_M  (\calD \vec{f}_{\{\}}|_{\vec{\nu}^{(i)}}(\cdot))_{\{\}}^\circledast( \vec{\delta}^{(i)} ) ) & \text{(\cref{Lem:Appendix:TaylorUni})} \\
    & = \vec{f}_{\{\}}( \vec{\nu}^{(i+1)} ). & \text{(\cref{Lem:Appendix:StarSolve})}
  \end{align*}
  Next, we prove $\vec{f}_{\{\}}(\vec{\nu}^{(i+1)}) \aord_M \vec{\nu}^{(i+2)}$:
  \begin{align*}
    \vec{f}_{\{\}}(\vec{\nu}^{(i+1)}) & = \vec{\nu}^{(i+1)} \oplus_M \vec{\delta}^{(i+1)} & \text{($\vec{\delta}^{(i+1)} \defeq \vec{f}_{\{\}}(\vec{\nu}^{(i+1)}) \ominus_M \vec{\nu}^{(i+1)}$)} \\
    & \aord_M \vec{\nu}^{(i+1)} \oplus_M (\calD \vec{f}_{\{\}}|_{\vec{\nu}^{(i)}}(\cdot))_{\{\}}^{\circledast}(\vec{\delta}^{(i+1)}) & \text{($\vec{\delta} \aord_M \vec{g}^\circledast(\vec{\delta})$)} \\
    & = \vec{\nu}^{(i+2)}. & \text{(\cref{Lem:Appendix:StarSolve})}
  \end{align*}
  
  By Kleene's fixed-point theorem, we know that $\lfp_{\vec{\azero}_M}^{\aord_M} \vec{f}_{\{\}} = \textstyle\bigsqcup^{\uparrow}_{j \in \bbN} \vec{\kappa}^{(j)}$.
  It remains to show that $\vec{\nu}^{(i)} \aord_M \lfp_{\vec{\azero}_M}^{\aord_M} \vec{f}_{\{\}}$ for all $i$.
  We proceed by induction on $i$. The base case $i=0$ is trivial.
  Let $i \ge 0$. We have
  \begin{align*}
    \vec{\nu}^{(i+1)} & = \vec{\nu}^{(i)} \oplus_M (\calD \vec{f}_{\{\}}|_{\vec{\nu}^{(i)}}(\cdot))_{\{\}}^{\circledast}(\vec{\delta}^{(i)}) & \text{(\cref{Lem:Appendix:StarSolve})} \\
    & = \vec{\nu}^{(i)} \oplus_M \bigoplus_{M,d\in\bbN} (\calD\vec{f}_{\{\}}|_{\vec{\nu}^{(i)}}(\cdot))^d(\vec{\delta}^{(i)}) \\
    & \aord_M \vec{\nu}^{(i)} \oplus_M \bigoplus_{M,d\in\bbN} \vec{e}^{(d)}(\vec{\nu}^{(i)}) & \text{(\cref{Lem:Appendix:Aux})} \\
    & = \bigsqcup_{d\in\bbN} \vec{f}^d_{\{\}}(\vec{\nu}^{(i)}) \\
    & \aord_M \lfp_{\vec{\azero}_M}^{\aord_M} \vec{f}_{\{\}},
  \end{align*}
  where the last step follows from the induction hypothesis $\vec{\nu}^{(i)} \aord_M \lfp_{\vec{\azero}_M}^{\aord_M} \vec{f}_{\{\}}$, thus $\vec{f}^d(\vec{\nu}^{(i)}) \aord_M \lfp_{\vec{\azero}_M}^{\aord_M} \vec{f}_{\{\}}$  for any $d \in \bbN$.
\end{proof}


\section{Soundness of \framework{}}
\label{Se:Appendix:Soundness}

In this section, we prove the soundness of \framework{}.
We use a recently proposed family of algebraic structures,
namely \emph{Markov algebras} (MAs)~\cite{PLDI:WHR18}, to specify concrete semantics of probabilistic programs.
We then introduce \emph{soundness relations} between MAs and $\omega$PMAs and show those
relations guarantee the soundness of program analyses in \framework{}.
We use some standard notions from domain theory:
directed-complete partial order (dcpo) and Scott-continuous function.

\noindent\textbf{Semantic foundation.}
%
%
We use the interpretation of regular infinite-tree expressions over MAs to define concrete semantics
of probabilistic programs.
%
%
  A \emph{Markov algebra} (MA) $\calM = \tuple{M,\aord_M, \otimes_M,\gcho{\varphi}_M, \pcho{p}_M, \dashcup_M, \azero_M, \aone_M}$ over a set $\calL$ of logical conditions is 
  almost the same as an $\omega$PMA, except that it does not contain the $\oplus_M$ operation but
  has an explicit partial order $\aord_M$, as well as it satisfies a different
  set of algebraic laws: $\tuple{M,\aord_M}$ forms a dcpo with $\azero_M$ as its least element;
  $\tuple{M,\otimes_M,\aone_M}$ forms a monoid;
  $\dashcup_M$ is idempotent, commutative, associative and for all $a,b \in M$ and $\varphi \in \calL, p\in[0,1]$
  it holds that $a \gcho{\varphi}_M b, a \pcho{p}_M b \le_M a \dashcup_M b$,
  where $\le_M$ is the semilattice ordering induced by $\dashcup_M$ (i.e., $a \le_M b$ if $a \dashcup_M b = b$);
  and $\otimes_M,\gcho{\varphi}_M, \pcho{p}_M, \dashcup_M$ are Scott-continuous. 
%
  An \emph{MA interpretation} is then a pair $\scrM = \tuple{\calM,\interp{\cdot}^\scrM}$
  where $\interp{\cdot}^\scrM$ maps data actions
  to $\calM$.

We consider probabilistic programs represented as an equation system $\{ X_i = E_{X_i} \}_{i=1}^n$,
where the $i^\textit{th}$ procedure has name $X_i$, $E_{X_i} \in \m{RegExp}^\infty(\calF,\emptyset)$ encodes
the regular infinite-tree expression through the CFHG of $X_i$.
Given a regular infinite-tree expression $E \in \m{RegExp}^\infty(\calF,\calK)$,
an MA interpretation $\scrM = \tuple{\calM,\interp{\cdot}^\scrM}$,
a valuation $\gamma : \calK \to \calM$,
and a procedure-summary vector $\vec{\nu} \in \calM^n$,
the interpretation of $E$ under $\gamma$ and $\vec{\nu}$, denoted by $\scrM_\gamma\interp{E}(\vec{\nu})$,
can be defined in the same way as the $\omega$PMA interpretations presented in \cref{Se:OmegaPMAsAndAlgebraicExpressions}.
We still define the interpretation for $\mu$-binders using least fixed-points,
because $\calM$ is a dcpo.

%


\begin{example}\label{Exa:Appendix:MAForBooleanPrograms}
  Consider probabilistic Boolean programs with a set $\m{Var}$ of program variables.
  Let $S \defeq 2^{\m{Var}}$ denote the state space of such programs.
  We formulate a demonic denotational semantics~\cite{book:MM05} by defining an MA interpretation.
  Let $\overline{S} \defeq \{ \mu : S \to [0,1] \mid \sum_{s\in S} \mu(s) \le 1 \}$ denote the set of \emph{sub-distributions} on $S$.
  We define a partial order on $\overline{S}$ by $\mu_1 \le \mu_2 \defeq \Forall{s \in S} \mu_1(s) \le \mu_s(s)$.
  We write $\bbC S$ to be the set of non-empty, up-closed, convex, and Cauchy-closed subsets of $\overline{S}$:
  a set $\calO$ of $\overline{S}$ is said to be
  \emph{up-closed} if $\mu \in \calO$ and $\mu \le \mu'$ imply $\mu' \in \calO$,
  \emph{convex} if $\mu_1,\mu_2 \in \calO$ implies $p{\cdot}\mu_1{+}(1{-}p){\cdot}\mu_2 \in \calO$ for any $p \in [0,1]$, and
  \emph{Cauchy-closed} if $\calO$ is a closed subset of some $N$-dimensional Euclidean space $\bbR^N$
  (in this example, we have $N = |S|$).
  \citeauthor{book:MM05} proved that $\bbH S \defeq S \to \bbC S$ forms a dcpo with respect to the ordering
  $
  r_1 \aord_C r_2 \defeq \Forall{s \in S} r_1(s) \supseteq r_2(s)
  $.
  
  For any $r \in \bbH S$, we can lift $r$ to a Scott-continuous mapping $\widehat{r} : \bbC S \to \bbC S$ as
  $
  \widehat{r} \defeq \lambda \calO. \; \{ \overline{f}(\mu) \mid f \in \bbD S, r \lesssim f, \mu \in \calO \},
  $
  where $\bbD S \defeq S \to \overline{S}$ is the set of state-to-distribution mappings,
  for any $f \in \bbD S$ we can lift $f$ to a Scott-continuous mapping $\overline{f} : \overline{S} \to \overline{S}$ as
  $\overline{f} \defeq \lambda \mu. \; \lambda s'.\; \sum_{s \in S} f(s)(s') \cdot \mu(s)$,
  and by $r \lesssim f$ we mean $\Forall{s \in S} f(s) \in r(s)$.
  
  We now formulate an MA $\calC = \tuple{ \bbH S, {\aord}_C, {\otimes}_C, \gcho{\varphi}_C, \pcho{p}_C, \dashcup_C, \azero_C, \aone_C}$ on $\bbH S$
  as follows, where we write $\varphi(s)$ for the truth value of $\varphi$ in state $s$ and $\m{ite}$ for the if-then-else operator.
  \begin{align*}
    r_1 \otimes_C r_2 & \defeq \lambda s.\; \widehat{r}_2(r_1(s)), &
    r_1 \pcho{p}_C r_2 & \defeq \lambda s.\; \{ p {\cdot} \mu_1 {+} (1{-}p) {\cdot} \mu_2 \mid \mu_1 \in r_1(s), \mu_2 \in r_2(s) \}, \\
    r_1 \gcho{\varphi}_C r_2 & \defeq \lambda s.\; \m{ite}(\varphi(s), r_1(s), r_2(s)), &
    r_1 \dashcup_C r_2 & \defeq \lambda s.\; \{ q {\cdot} \mu_1 {+} (1{-}q) {\cdot} \mu_2 \mid q {\in} [0,1], \mu_1 {\in} r_1(s), \mu_2 {\in} r_2(s) \}, \\
    \azero_C & \defeq \lambda s.\; \overline{S}, &
    \aone_C & \defeq \lambda s.\; \{ \lambda s'.\; \m{ite}(s = s', 1, 0) \}.
  \end{align*}
\end{example}

\noindent\textbf{Soundness relations.}
We adapt abstract interpretation~\cite{POPL:CC77} to justify \framework{}
by establishing a relation between the concrete and abstract semantics.
We express the relation via a \emph{soundness relation}, which is a binary relation between an MA interpretation
and an $\omega$PMA interpretation that is preserved by the algebraic operations.
Intuitively, a (concrete, abstract) pair in the relation should be read as ``the concrete element is
approximated by the abstract element.''

\begin{definition}\label{De:OmegaSoundnessRelations}
  Let $\scrC = \tuple{\calC,\interp{\cdot}^\scrC}$ be an MA interpretation
  and $\scrM = \tuple{\calM,\interp{\cdot}^\scrM}$ be an $\omega$PMA interpretation.
  We say ${-} \Vdash {-} \subseteq \calC \times \calM$ is a \emph{soundness relation},
  if $\interp{\m{act}}^\scrC \Vdash \interp{\m{act}}^\scrM$ for all data actions $\m{act} \in \calA$,
  $\azero_C \Vdash \azero_M$, $\aone_C \Vdash \aone_M$, and for all $x_1 \Vdash y_1$ and $x_2 \Vdash y_2$ we have
  \begin{itemize}[nosep,leftmargin=*]
    \item $x_1 \otimes_C x_2 \Vdash y_1 \otimes_M y_2$,
    $x_1 \gcho{\varphi}_C x_2 \Vdash y_1 \gcho{\varphi}_M y_2$ for all $\varphi \in \calL$,
    $x_1 \pcho{p}_C x_2 \Vdash y_1 \pcho{p}_M y_2$ for all $p \in [0,1]$,
    $x_1 \dashcup_C x_2 \Vdash y_1 \dashcup_M y_2$, and
    \item for any Scott-continuous $\vec{f} : \vec{C} \to \vec{C}$ and Scott-continuous $\vec{g} : \vec{D} \to \vec{D}$, the property ``$\vec{x} \Vdash \vec{y}$ implies $\vec{f}(\vec{x}) \Vdash \vec{g}(\vec{y})$
    for all $\vec{x},\vec{y}$'' implies $\lfp_{\vec{\azero}_C}^{\aord_C} \vec{f} \Vdash \lfp_{\vec{\azero}_M}^{\aord_M} \vec{g}$.
  \end{itemize}
\end{definition}

\begin{example} 
  Let $\scrC = \tuple{\calC,\interp{\cdot}^\scrC }$ be an MA interpretation where $\calC$ is an MA for
  probabilistic Boolean programs, which is introduced in \cref{Exa:Appendix:MAForBooleanPrograms}.
  Let $\scrB = \tuple{\calB, \interp{\cdot}^\scrB }$ be an $\omega$PMA interpretation where $\calB$ is the $\omega$PMA over matrices
  described in \cref{Exa:BayesianInferenceDomain}.
  We define the following approximation relation to indicate that the program analysis reasons about lower bounds:
  \[
  r \Vdash \mathbf{A} \iff \Forall{s \in S} \Forall{\mu \in r(s)} \Forall{s' \in S} \mu(s') \ge \mathbf{A}(s,s').
  \]
  %
\end{example}

The correctness of an interpretation
is then justified by the following soundness theorem, which followed by induction on the structure of
regular infinite-tree expressions.

\begin{theorem*}[\cref{The:InterSoundness}]
  Let $\scrC = \tuple{\calC, \interp{\cdot}^\scrC}$ be an MA interpretation.
  Let $\scrM = \tuple{\calM, \interp{\cdot}^\scrM}$ be an $\omega$PMA interpretation.
  Let ${\Vdash} \subseteq \calC \times \calM$ be a soundness relation.
  Then for any regular infinite-tree expression $E \in \m{RegExp}^\infty(\calF,\calK)$,
  $\gamma : \calK \to \calC$, $\gamma^\sharp : \calK \to \calM$ such that
  $\gamma(\square) \Vdash \gamma^\sharp(\square)$ for all $\square \in \calK$,
  $\vec{\nu} : \calC^n$, and $\vec{\nu}^\sharp : \calM^n$ such that $\nu_i \Vdash \nu^\sharp_i$ for $i=1,\cdots,n$,
  we have $\scrC_{\gamma}\interp{E}(\vec{\nu}) \Vdash \scrM_{\gamma^\sharp}\interp{E}(\vec{\nu}^\sharp)$.
\end{theorem*}
\begin{proof}
  By induction on the structure of $E$.
  \begin{description}[labelindent=\parindent]
    \item[{Case $E=\varepsilon$:}]\
    
    We have $\scrC_\gamma\interp{E}(\vec{\nu}) = \aone_C$.
    
    We have $\scrM_{\gamma^\sharp}\interp{E}(\vec{\nu}^\sharp) = \aone_M$.
    
    By \cref{De:OmegaSoundnessRelations}, we conclude that $\aone_C \Vdash \aone_M$.
    
    \item[{Case $E=\mi{seq}[\m{act}](E_1)$:}]\
    
    We have $\scrC_\gamma\interp{E}(\vec{\nu}) = \interp{\m{act}}^\scrC \otimes_C \scrC_\gamma\interp{E_1}(\vec{\nu})$.
    
    We have $\scrM_{\gamma^\sharp}\interp{E}(\vec{\nu}^\sharp) = \interp{\m{act}}^\scrM \otimes_M \scrM_{\gamma^\sharp}\interp{E_1}(\vec{\nu}^\sharp)$.
    
    By induction hypothesis, we have $\scrC_\gamma\interp{E_1}(\vec{\nu}) \Vdash \scrM_{\gamma^\sharp}\interp{E_1}(\vec{\nu}^\sharp)$.
    
    By \cref{De:OmegaSoundnessRelations}, we have $\interp{\m{act}}^\scrC \Vdash \interp{\m{act}}^\scrM$ and conclude by the fact that $\Vdash$ preserves extend operation.
  
    \item[{Case $E=\mi{call}[X_i](E_1)$:}]\
    
    We have $\scrC_\gamma\interp{E}(\vec{\nu}) = \nu_i \otimes_C \scrC_\gamma\interp{E_1}(\vec{\nu})$.
    
    We have $\scrM_{\gamma^\sharp}\interp{E}(\vec{\nu}^\sharp) = \nu^\sharp_i \otimes_M \scrM_{\gamma^\sharp}\interp{E_1}(\vec{\nu}^\sharp)$.
    
    By induction hypothesis, we have $\scrC_\gamma\interp{E_1}(\vec{\nu}) \Vdash \scrM_{\gamma^\sharp}\interp{E_1}(\vec{\nu}^\sharp)$.
    
    By assumption, we have $\nu_i \Vdash \nu_i^\sharp$ and conclude by the fact that $\Vdash$ preserves extend operation.
    
    \item[{Case $E=\mi{cond}[\varphi](E_1,E_2)$:}]\
    
    We have $\scrC_\gamma\interp{E}(\vec{\nu}) = \scrC_\gamma\interp{E_1}(\vec{\nu}) \gcho{\varphi}_C \scrC_\gamma\interp{E_2}(\vec{\nu})$.
    
    We have $\scrM_{\gamma^\sharp}\interp{E}(\vec{\nu}^\sharp) = \scrM_{\gamma^\sharp}\interp{E_1}(\vec{\nu}^\sharp) \gcho{\varphi}_M \scrM_{\gamma^\sharp}\interp{E_2}(\vec{\nu}^\sharp)$.
    
    By induction hypothesis, we have $\scrC_\gamma\interp{E_1}(\vec{\nu}) \Vdash \scrM_{\gamma^\sharp}\interp{E_1}(\vec{\nu}^\sharp)$ and
    $\scrC_\gamma\interp{E_2}(\vec{\nu}) \Vdash \scrM_{\gamma^\sharp}\interp{E_2}(\vec{\nu}^\sharp)$.
    
    We then conclude by the fact that $\Vdash$ preserves conditional-choice operation.
    
    Cases $\mi{prob}[p](E_1,E_2)$ and $\mi{ndet}(E_1,E_2)$ are similar to this case.
    
    \item[{Case $E=\square$:}]\
    
    We have $\scrC_{\gamma}\interp{E}(\vec{\nu}) = \gamma(\square)$.
    
    We have $\scrM_{\gamma^\sharp}\interp{E}(\vec{\nu}) = \gamma^\sharp(\square)$.
    
    We conclude by the assumption $\gamma(\square) \Vdash \gamma^\sharp(\square)$.
    
    \item[{Case $E=E_1 \cdot_{\square} E_2$:}]\
    
    We have $\scrC_{\gamma}\interp{E}(\vec{\nu}) = \scrC_{\gamma[ \square \mapsto \scrC_{\gamma}\interp{E_2}(\vec{\nu}) ]}\interp{E_1}(\vec{\nu})$.
    
    We have $\scrM_{\gamma^\sharp}\interp{E}(\vec{\nu}^\sharp) = \scrN_{\gamma^\sharp[ \square \mapsto \scrM_{\gamma^\sharp}\interp{E_2}(\vec{\nu}^\sharp) ]}\interp{E_1}(\vec{\nu}^\sharp)$.
    
    By induction hypothesis, we have $\scrC_\gamma\interp{E_2}(\vec{\nu}) \Vdash \scrM_{\gamma^\sharp}\interp{E_2}(\vec{\nu}^\sharp)$.
    
    Thus, $\gamma[ \square \mapsto \scrC_{\gamma}\interp{E_2}(\vec{\nu}) ] \Vdash \gamma^\sharp[ \square \mapsto \scrM_{\gamma^\sharp}\interp{E_2}(\vec{\nu}^\sharp) ]$.
    
    By induction hypothesis, we have $\scrC_{\gamma[ \square \mapsto \scrC_{\gamma}\interp{E_2}(\vec{\nu}) ]}\interp{E_1}(\vec{\nu}) \Vdash \scrM_{\gamma^\sharp[ \square \mapsto \scrM_{\gamma^\sharp}\interp{E_2}(\vec{\nu}^\sharp) ]}\interp{E_1}(\vec{\nu}^\sharp)$ and then conclude this case.
    
    \item[{Casse $E=(E_1)^{\infty_\square}$:}]\
    
    We have $\scrC_{\gamma}\interp{E}(\vec{\nu}) = \lfp_{\azero_C}^{\aord_C} \lambda \theta.\; \scrC_{\gamma[\square \mapsto \theta]}\interp{E_1}(\vec{\nu})$.
    Define $f(\theta) \defeq \scrC_{\gamma[\square \mapsto \theta]}\interp{E_1}(\vec{\nu})$.
    
    We have $\scrM_{\gamma^\sharp}\interp{E}(\vec{\nu}^\sharp) = \lfp_{\azero_M}^{\aord_M} \lambda \theta^\sharp.\; \scrM_{\gamma^\sharp[\square \mapsto \theta^\sharp]}\interp{E_1}(\vec{\nu}^\sharp)$.
    Define $g(\theta^\sharp) \defeq \scrM_{\gamma^\sharp[\square \mapsto \theta^\sharp]}\interp{E_1}(\vec{\nu}^\sharp)$.
    
    For any $\theta \Vdash \theta^\sharp$, by induction hypothesis, we have $f(\theta) \Vdash g(\theta^\sharp)$.
    
    Thus, by \cref{De:OmegaSoundnessRelations}, we have $\lfp_{\azero_C}^{\aord_C} f \Vdash \lfp_{\azero_M}^{\aord_M} g$ and
    conclude this case.
  \end{description}
\end{proof}

Let $\scrC = \tuple{\calC, \interp{\cdot}^\scrC}$ be an MA interpretation
and $\scrM = \tuple{\calM, \interp{\cdot}^\scrM}$ be an $\omega$PMA interpretation.
By \cref{The:InterSoundness} and the last property of \cref{De:OmegaSoundnessRelations}, we can show that if $P \defeq \{X_i = E_{X_i}\}_{i=1}^n$ is an interprocedural equation system where $E_{X_i} \in \m{RegExp}^\infty(\calF,\emptyset)$ for each $i=1,\cdots,n$, it holds that
\begin{equation}\label{Eq:InterSoundApprox}
\lfp_{\vec{\azero}_C}^{\aord_C} \lambda \vec{\theta}.\; \tuple{\scrC_{\{\}}\interp{E_{X_i}}(\vec{\theta}) }_{i=1,\cdots,n}
\Vdash
\lfp_{\vec{\azero}_M}^{\aord_M} \lambda \vec{\theta}^\sharp.\; \tuple{\scrM_{\{\}}\interp{E_{X_i}}(\vec{\theta}^\sharp) }_{i=1,\cdots,n},
\end{equation}
where the left-hand side of \cref{Eq:InterSoundApprox} stands for the concrete semantics of the procedures
and the right-hand side of \cref{Eq:InterSoundApprox} is their corresponding abstractions, i.e., $\scrC\interp{P} \Vdash \scrM\interp{P}$.

Our \framework{} framework first translates each $E_{X_i}$ to an equivalent algebraic expression $E_{X_i}'$ by \cref{Lem:EquivalentToAlgebraicExpressions},
and then applies Newton's method (\cref{De:PMANewtonSequence,Cor:NewtonConvergence}) to approximate
$\lfp_{\vec{\azero}_M}^{\aord_M} \lambda \vec{\theta}^\sharp.\; \tuple{\calM_{\{\}}\interp{E_{X_i}'}(\vec{\theta}^\sharp)}_{i=1,\cdots,n}$,
which is equivalent to the right-hand side of \cref{Eq:InterSoundApprox}.

\section{Case Study: Bayesian-Inference Analysis via Algebraic Decision Diagrams}
\label{Se:BayesianInferenceADD}

\begin{wrapfigure}{r}{0.42\textwidth}
\centering
\vspace{-1em}
\includegraphics[width=0.4\textwidth]{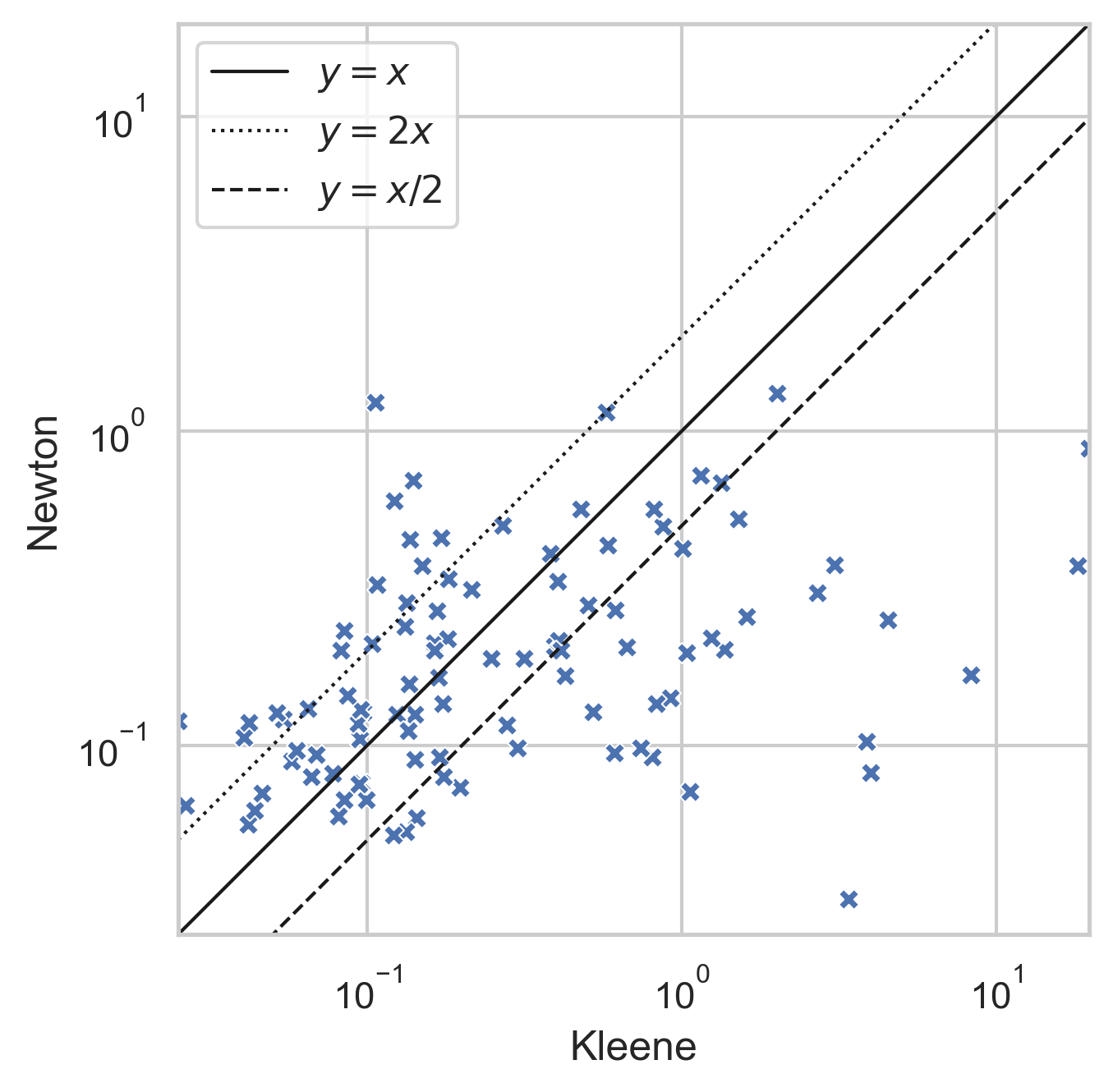}
\vspace{-0.5em}
\caption{Log-log scatter plots for running times (seconds) of Kleene iteration and \framework{}.}
\label{Fi:RuntimePlotBayesianInferenceADD}
\vspace{0.5em}
\end{wrapfigure}

In \cref{Se:BayesianInferenceAnalysis}, we discussed an instantiation
of \framework{} for Bayesian-inference analysis via matrices that encode distribution transformers.
The matrix-based encoding leads to a very high computational complexity because the instantiation needs
$2^n$-by-$2^n$ matrices to analyze programs with $n$ Boolean program variables.
Recall that in this article, we consider the style of Bayesian-inference analysis from prior work by \citet{FSE:CRN13}.
Their implementation uses \emph{Algebraic Decision Diagrams} (ADDs)~\cite{FMSD:BFG97} to represent
distributions compactly.
In this section, we develop an instantiation that uses ADDs to represent distribution transformers and
carry out Bayesian-inference analysis of probabilistic programs \emph{without} nondeterminism.

Let $\m{Var}$ be a set of Boolean program variables.
As discussed in \cref{Exa:BayesianInferenceDomain}, distribution transformers are mappings
from states ($2^\m{Var}$) to state-distributions ($2^{\m{Var}} \to [0,1]$) and can be encoded as matrices
of type $2^{\m{Var}} \times 2^{\m{Var}} \to [0,1]$.
In this section, we introduce a set $\m{Var}'$ that consists of primed copies of the variables in $\m{Var}$
and change the type of matrices to $2^{\m{Var}} \times 2^{\m{Var}'} \to [0,1]$.
The primed copies in $\m{Var}'$ represent the values of the variables in the \emph{post}-state of transformers
and the unprimed original ones represent the \emph{pre}-state.
We then use ADDs as a compact representation of mappings in $2^{\m{Var} \cup \m{Var}'} \to [0,1]$.
ADDs have been shown capable of performing many matrix operations including pointwise addition/subtraction,
matrix multiplication, and matrix inversion (via Gaussian elimination)~\cite{FMSD:BFG97}.
Therefore, to simplify the presentation in this section, we still use matrices to describe our development
and the all the matrix operations are realizable using ADDs.

Recall that without nondeterminism, the linear-recursion-solving strategy for Bayesian-inference analysis
essentially solves a system of linear matrix equations, i.e., each equation has the form
$Z_j = \mathbf{C}_j + \sum_{i,k} (\mathbf{A}_{i,j,k} \cdot Z_i \cdot \mathbf{B}_{i,j,k})$,
where $\mathbf{A}_{i,j,k}, \mathbf{B}_{i,j,k}, \mathbf{C}_j$ are known matrices.
The matrices of type $2^{\m{Var}} \times 2^{\m{Var}'} \to \bbR_{\ge 0} \cup \{\infty\}$
admit an $\omega$-continuous semiring with pointwise addition and matrix multiplication
as the combine and extend operations, respectively.
In addition, the Kleene-star of a matrix $\mathbf{A}$ satisfying that $\norm{A} < 1$ can be computed effectively via matrix inversion, where $\norm{\cdot}$ denotes the Frobenius norm
and $\mathbf{I}$ is the identity matrix:
\begin{equation*}
\mathbf{A}^\circledast \defeq \textstyle\sum_{i \in \bbN} \mathbf{A}^i = (\mathbf{I} - \mathbf{A})^{-1} .
\end{equation*}
%
%
We then adopt a technique proposed by \citet{POPL:RTP16} for solving linear equation systems over $\omega$-continuous semirings.
Their technique transforms a linear equation system to a \emph{regular} one, i.e., each equation has
the form $W_j = \mathbf{E}_j + \sum_{i} (\mathbf{D}_{i,j} \cdot W_i)$, via a \emph{tensor-product} operation $\tensor$ and a \emph{detensor-transpose} operation $\detensor$.
The obtained regular system is then solved by \citet{JACM:Tarjan81:Alg}'s path-expression algorithm.
In our development, the tensor-product and the detensor-transpose operations are defined as follows,
where $\mathbf{R}$ and $\mathbf{S}$ are $N$-by-$N$ matrices and $\mathbf{T}$ is an $N^2$-by-$N^2$ matrix:
\[
\mathbf{R} \tensor \mathbf{S} \defeq
\begin{pNiceMatrix}[small,xdots/shorten=5pt]
  r_{1,1} \mathbf{S} & \Cdots & r_{1,N} \mathbf{S} \\
  \Vdots & \Ddots & \Vdots \\
  r_{N,1} \mathbf{S} & \Cdots & r_{N,N} \mathbf{S}
\end{pNiceMatrix} , \qquad
\detensor(\mathbf{T}) \defeq
\left( \textstyle\sum_{k=1}^N \mathbf{T}_{(i-1)N+j, (k-1)N+k}    \right)_{i=1,\cdots,N ; j = 1,\cdots,N}.
\]
The detensor-transpose operation satisfies $\detensor(\mathbf{A} \tensor \mathbf{B}) = \mathbf{A} \cdot \mathbf{B}^\m{T}$ for matrices $\mathbf{A}$ and $\mathbf{B}$.
For a linear matrix equation $Z_j = \mathbf{C}_j + \sum_{i,k} (\mathbf{A}_{i,j,k} \cdot Z_i \cdot \mathbf{B}_{i,j,k})$,
we transform it to the regular equation
\[
W_j = ( \mathbf{I} \tensor \mathbf{C}_j^\m{T} + \textstyle\sum_{i,k} ( ( \mathbf{A}_{i,j,k} \tensor \mathbf{B}_{i,j,k}^\m{T} ) \cdot W_i ) ,
\]
and then the solution to the original linear system can be recovered from the solution to the obtained regular
system via $Z_j = \detensor(W_j)$.

\begin{example}
  Consider the following linear matrix equation where $Z$ is a $2$-by-$2$ matrix:
  \[
  Z = \begin{pNiceMatrix}[small] 0.2 & 0.3 \\ 0.1 & 0.4 \end{pNiceMatrix} + 
  \begin{pNiceMatrix}[small] 0.2 & 0.3 \\ 0.4 & 0.1 \end{pNiceMatrix} \cdot Z \cdot
  \begin{pNiceMatrix}[small] 0.1 & 0.9 \\ 0.8 & 0.2 \end{pNiceMatrix} .
  \]
  Using the tensor-product operation, we transform it to a regular equation:
  \[
  W = \begin{pNiceMatrix}[small] 0.2 & 0.1 & 0 & 0 \\ 0.3 & 0.4 & 0 & 0 \\ 0 & 0 & 0.2 & 0.1 \\ 0 & 0 & 0.3 & 0.4 \end{pNiceMatrix} + 
  \begin{pNiceMatrix}[small] 0.02 & 0.16 & 0.03 & 0.24 \\ 0.18 & 0.04 & 0.27 & 0.06 \\ 0.04 & 0.32 & 0.01 & 0.08 \\ 0.36 & 0.08 & 0.09 & 0.02 \end{pNiceMatrix} \cdot W .
  \]
  The solution to the regular equation on $W$ can be obtained by
  \[
  W = \begin{pNiceMatrix}[small] 0.02 & 0.16 & 0.03 & 0.24 \\ 0.18 & 0.04 & 0.27 & 0.06 \\ 0.04 & 0.32 & 0.01 & 0.08 \\ 0.36 & 0.08 & 0.09 & 0.02 \end{pNiceMatrix}^\circledast \cdot \begin{pNiceMatrix}[small] 0.2 & 0.1 & 0 & 0 \\ 0.3 & 0.4 & 0 & 0 \\ 0 & 0 & 0.2 & 0.1 \\ 0 & 0 & 0.3 & 0.4 \end{pNiceMatrix}
  = \left( \mathbf{I} - \begin{pNiceMatrix}[small] 0.02 & 0.16 & 0.03 & 0.24 \\ 0.18 & 0.04 & 0.27 & 0.06 \\ 0.04 & 0.32 & 0.01 & 0.08 \\ 0.36 & 0.08 & 0.09 & 0.02 \end{pNiceMatrix} \right)^{-1} \cdot \begin{pNiceMatrix}[small] 0.2 & 0.1 & 0 & 0 \\ 0.3 & 0.4 & 0 & 0 \\ 0 & 0 & 0.2 & 0.1 \\ 0 & 0 & 0.3 & 0.4 \end{pNiceMatrix}
  \approx \begin{pNiceMatrix}[small] 0.32 & 0.23 & 0.12 & 0.14 \\ 0.43 & 0.52 & 0.13 & 0.11 \\ 0.17 & 0.19 & 0.28 & 0.18 \\ 0.17 & 0.14 & 0.39 & 0.49 \end{pNiceMatrix} .
  \]
  Finally, we use the detensor-transpose operation to obtain the solution to the linear equation on $Z$:
  \[
  Z = \detensor(W) = \begin{pNiceMatrix}[small] W_{1,1}+W_{1,4} & W_{2,1} + W_{2,4} \\ W_{3,1} + W_{3,4} & W_{4,1}+W_{4,4} \end{pNiceMatrix} \approx \begin{pNiceMatrix}[small] 0.46 & 0.54 \\ 0.35 & 0.65 \end{pNiceMatrix} .
  \]
\end{example}

Our implementation of the abstract domain for ADD-based Bayesian-inference analysis consists of about
500 lines of code (excluding the CUDD package, which implements basic ADD operations).
We evaluated the performance of NPA-PMA against Kleene iteration on the same benchmark suite of
100 randomly generated programs as the suite mentioned in \cref{Se:BayesianInferenceAnalysis}.
The average number of Newton rounds performed by \framework{} is 9.58, whereas the average number
of Kleene rounds is 3,071.08.
\cref{Fi:RuntimePlotBayesianInferenceADD} presents a scatter plot that compares the running time of
Kleene iteration against \framework{}.
In terms of running time,
\framework{} achieves an overall 1.54x speedup compared with Kleene iteration.
Note that although the linear-recursion-solving strategy involves computationally heavy tensor-product
and detensor-transpose operations, the significant reduction in the number of iterations (3,071.08 $\to$ 9.58)
leads to an actual speedup.
This experimental result is much better than that of prior work by \citet{POPL:RTP16},
from which we adopt the technique of solving linear equation systems via tensor products.
\citeauthor{POPL:RTP16} used tensor products in the NPA framework for analyzing non-probabilistic programs
and developed an extension of NPA named NPA-TP.
However, in their experimental evaluation, NPA-TP turns out to be slower than chaotic iteration on a BDD-based
predicate-abstraction domain.
In our work,
we think it is the quantitative nature of probabilistic programs that unleashes the potential of Newton's method.

\section{Case Study: Higher-Moment Analysis of Accumulated Rewards}
\label{Se:HigherMomentAnalysis}


In \cref{Se:ThisWorkAConfluenceCentricAnalysisFramework}, we demonstrated the interprocedural part
of \framework{} using the termination-probability analysis of finite-state probabilistic programs with nondeterminism.
In this section, we augment the programs with a global reward accumulator, which can be incremented
via the data action $\kw{reward}(c)$, where $c \in \bbN$ is a nonnegative number.
Reasoning about accumulated rewards of probabilistic programs has many applications
(e.g.,
expected runtime analysis~\cite{PLDI:NCH18,LICS:EKM05,ESOP:KKM16},
expected accumulated-cost analysis~\cite{PLDI:WFG19,ICALP:EWY08},
and tail-bound analysis~\cite{TACAS:KUH19,PLDI:WHR21A}).
We consider the \emph{higher-moment} analysis of accumulated rewards,
where the $k^{\textit{th}}$ moment of a random variable $X$ is $\bbE[X^k]$.
Note that this analysis is an extension of termination-probability analysis
because the zeroth moment of the accumulated reward corresponds to the termination probability.

For analyzing $k^{\textit{th}}$ moments,
we define an $\omega$PMA $\calR_k = \tuple{\overline{\bbR}_{\ge 0}^{k+1}, \oplus_R, \otimes_R, \gcho{\varphi}_R, \pcho{p}_R, \dashcup_R, \azero_R, \aone_R}$ as an extension of \emph{expectation semirings}~\cite{EMNLP:LE09}
and \emph{moment semirings}~\cite{PLDI:WHR21A}, where $\overline{\bbR}_{\ge_0} \defeq \bbR_{\ge 0} \cup \{ \infty \}$,
and the operations and constants are defined as follows.
\begin{align*}
  \vec{u} \oplus_R \vec{v} & \defeq  \tuple{u_i + v_i}_{i=0,\cdots,k}, &
  \vec{u} \pcho{p}_R \vec{v} & \defeq \tuple{ p \cdot u_i + (1{-}p) \cdot v_i}_{i=0,\cdots,k}, &
  \azero_R & \defeq \tuple{0,0,\cdots, 0},
  \\
  \vec{u} \otimes_R \vec{v} & \defeq \tuple{ \textstyle\sum_{j=0}^i \binom{i}{j} \cdot u_j \cdot v_{i-j} }_{i=0,\cdots,k}, &
  \vec{u} \dashcup_R \vec{v} & \defeq \tuple{ \max(u_i,v_i) }_{i=0,\cdots,k}, &
  \aone_R & \defeq \tuple{1,0,\cdots,0}.
\end{align*}
We do not define $\gcho{\varphi}_R$ because the considered finite-state probabilistic programs do not contain conditional branching.
The $\omega$-continuous semiring $\tuple{\overline{\bbR}_{\ge 0}^{k+1}, \oplus_R, \otimes_R, \azero_R, \aone_R}$
admits the partial order $\vec{u} \aord_R \vec{v} \defeq \vec{u} \le \vec{v}$ (i.e., pointwise comparison)
and the subtraction operation $\vec{u} \ominus_R \vec{v} \defeq \tuple{u_i - v_i}_{i=0,\cdots,k}$.
We can then formulate an $\omega$PMA interpretation $\scrR = \tuple{\calR,\interp{\cdot}^\scrR}$,
where $\interp{\kw{reward}(c)}^\scrR \defeq \tuple{c^i}_{i=0,\cdots,k}$ for $c \in \bbN$.
To prove soundness, we again use the concrete semantics presented in \cref{Exa:Appendix:MAForBooleanPrograms}
with the state space $S \defeq \bbN$.
We define the following approximation relation to indicate that the program analysis reasons about upper bounds on the moments
of the accumulated reward:
\[
r \Vdash \vec{u} \iff \Forall{n \in \bbN} \Forall{\mu \in r(n)} \textstyle\bigwedge_{i=0}^k \sum_{n' \in \bbN} \mu(n') \cdot (n' - n)^i \le u_i.
\]

%
The equation system obtained from a linearly recursive program is a max-linear numeric equation system;
thus, a linear-recursion-solving strategy $\m{solve}$ can be obtained via \emph{linear programming} (LP).
%
Consider a numeric equation system $\vec{Z} = \vec{f}(\vec{Z})$ where each $f_i(\vec{Z})$ is either
a linear expression over $\vec{Z}$ or the maximum of two variables $\max(Z_j,Z_k)$ for some $j,k$.
The least solution of $\vec{Z} = \vec{f}(\vec{Z})$ can then be obtained via the following LP:
{\small\begin{alignat*}{6}
  \textbf{minimize} \quad & \textstyle\sum_i Z_i, &
  \quad \textbf{s.t.} \quad & Z_i = b_i + \textstyle\sum_{j} a_{i,j} \cdot Z_j & & \quad \text{for each $i$ such that $f_i(\vec{Z}) = c_i + \textstyle \sum_j a_{i,j} \cdot Z_j$}, \\[-3pt]
   & & & Z_i \ge Z_j, Z_i \ge Z_k & & \quad \text{for each $i$ such that $f_i(\vec{Z}) = \max(Z_j,Z_k)$}, \\[-3pt]
  & & & Z_i \ge 0   & & \quad \text{for each $i$}.
\end{alignat*}}

Our implementation of the abstract domain for the higher-moment analysis consists of about 300 lines of code.
We evaluated the performance of \framework{} for second-moment analysis against Kleene iteration on
100 randomly generated finite-state probabilistic programs,
each of which consists of 500 procedures such that each procedure takes one of the following forms:
(i) $\mi{prob}[p](  \mi{call}[X_i]( \varepsilon ) ,  \mi{call}[X_j]( \varepsilon ) )$
for some probability $p$ and procedures $X_i,X_j$;
(ii) $\mi{prob}[p_1]( \mi{prob}[p_2]( \mi{call}[X_i]( \varepsilon ), \mi{call}[X_j]( \varepsilon ) ), \mi{seq}[\kw{reward}(1)](\varepsilon) )$
for some probabilities $p_1,p_2$ and procedures $X_i,X_j$; or
(iii) $\mi{call}[X_i]( \mi{call}[X_j] (\varepsilon))$ for some procedures $X_i,X_j$.
\begin{wrapfigure}{r}{0.42\textwidth}
\centering
\vspace{-1em}
\includegraphics[width=0.4\textwidth]{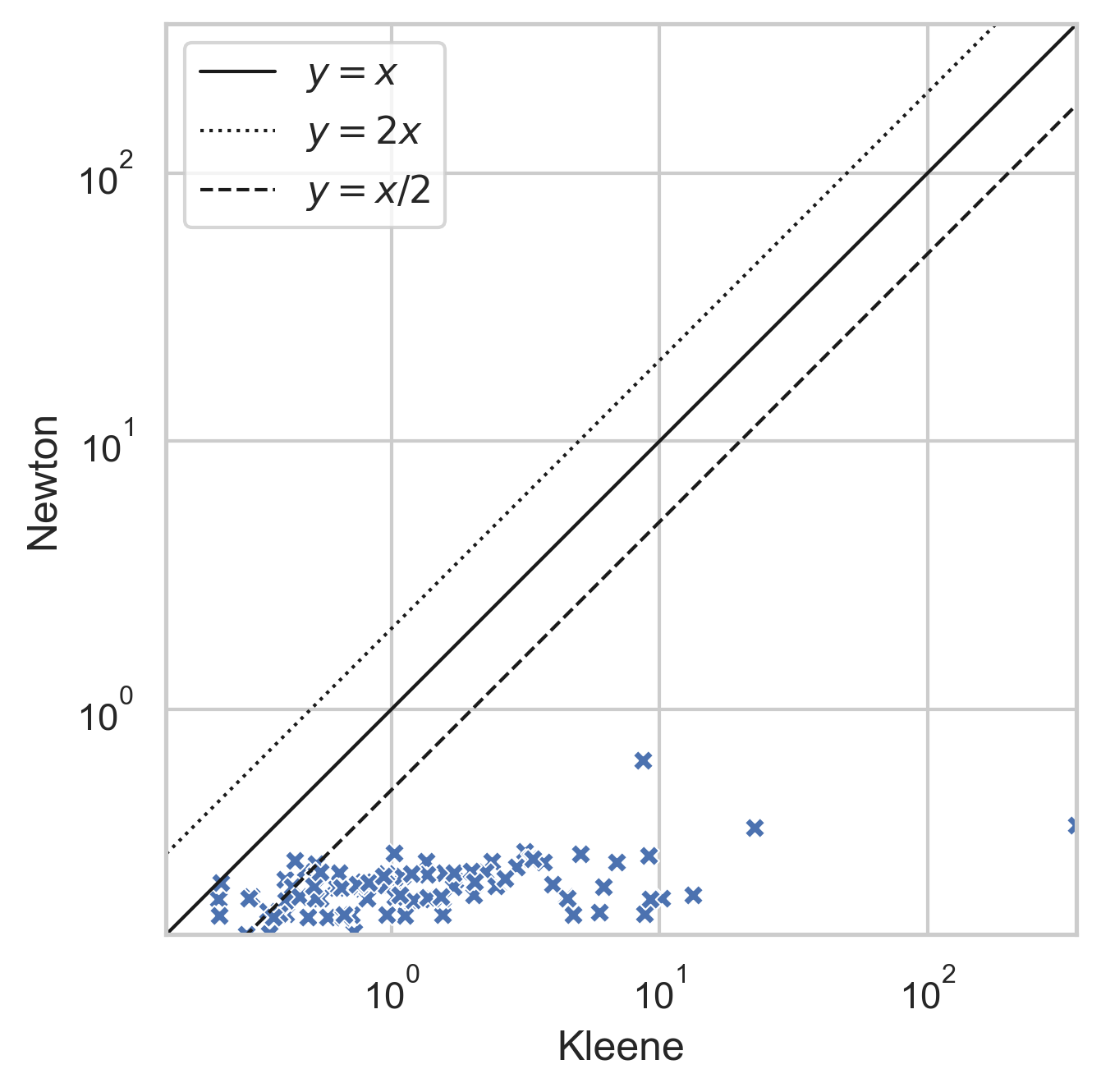}
\vspace{-0.5em}
\caption{Log-log scatter plots for running times (seconds) of Kleene iteration and \framework{}.}
\label{Fi:RuntimePlotMomentAnalysis}
\vspace{-2em}
\end{wrapfigure}%
On this benchmark suite of 100 programs,
the average number of Newton rounds performed by \framework{} is 8.99,
whereas the average number of Kleene rounds is 9,579.07.
\cref{Fi:RuntimePlotMomentAnalysis} presents a scatter plot that compares the running time (in seconds) of Kleene iteration
against \framework{} for moment-of-accumulated-reward analysis.
We observe that the overall performance of \framework{} is better than Kleene iteration
with an average speedup of 5.77x.

\section{Case study: Expectation-Invariant Analysis}
\label{Se:Appendix:ExpectationInvariantAnalysis}

\cref{Ta:FullExpectationInvariantAnalysis} presents the results of the evaluation
for the expectation-invariant analysis described in \cref{Se:ExpectationInvariantAnalysis}.

\begin{table}
\centering
\footnotesize
\caption{Expectation-invariant analysis. (Time is in seconds.)}
\label{Ta:FullExpectationInvariantAnalysis}
\begin{tabular}{c c c l}
\hline
\# & Program & Time & Invariants (selected) \\ \hline
1 & 2d-walk & 0.085399 & $ \bbE[\mi{dist1}'] \le \mi{dist1} + 0.5, \bbE[\mi{x1}'] \le \mi{x1} + 0.25$ \\
2 & aggregate-rv & 0.000688 & $\bbE[i'] \le i + 1, \bbE[r'] \le r + 0.5, \bbE[x'] \le x + 0.5$ \\
3 & biased-coin & 0.001104 & $\bbE[\mi{x1}'] \le \mi{x1} + 0.5 , \bbE[\mi{x2}'] \le \mi{x2} + 0.5$ \\
4 & binom-update & 0.000674 & $\bbE[n'] \le n + 1, \bbE[x'] \le x + 0.25$ \\
5 & coupon5 & 0.003276 & $ \bbE[\mi{count}'] \le \mi{count} + 1, \bbE[i'] \le i + 0.2$ \\
6 & dist & 0.000499 & $\bbE[x'] \le x, \bbE[y'] \le y, \bbE[z'] \le 0.5 x + 0.5 y$ \\
7 & eg & 0.002978 & $\bbE[x'] \le x + 3, \bbE[y'] \le y + 3, \bbE[z'] \le 0.25 z + 0.75$ \\
8 & eg-tail & 0.008772 & $\bbE[x'] \le x + 3, \bbE[y'] \le y + 3, \bbE[z'] \le 0.25 z + 0.75$ \\
9 & hare-turtle & 0.000657 & $\bbE[h'] \le h + 2.5, \bbE[r'] \le r + 2.5, \bbE[t'] \le t + 1$ \\
10 & hawk-dove & 0.003477 & $ \bbE[\mi{p1bal1}'] \le \mi{p1bal1} + 2, \bbE[\mi{p1bal2}'] \le \mi{p1bal2} + 1$ \\
11 & mot-ex & 0.001039 & $\bbE[r'] \le r + 0.75, \bbE[x'] \le x + 0.75, \bbE[y'] \le y + 1.5$ \\
12 & recursive & 0.001473 & $\bbE[x'] \le x + 9$ \\
13 & uniform-dist & 0.000422 & $\bbE[g'] \le 2g + 0.5, \bbE[n'] \le 2n$ \\
\hline
14 & coupon-five & 0.001739 & $\bbE[\mi{coupons}'] \le 5, \bbE[t'] \le t + 25$ \\
15 & coupon-five-fsm & 0.005259 & $\bbE[t'] \le t + 11.4167$ \\
16 & dice & 0.001575 & $\bbE[r'] \le 3.5, \bbE[t'] \le t + 1.33333$ \\
17 & exponential & 0.040215 & $\bbE[n'] \le 1, \bbE[r'] \le 0.2$ \\ 
18 & geometric & 0.001840 & $\bbE[i'] \le n, \bbE[t'] \le -2560 i + 2560 n + t$ \\
19 & non-linear-recursion & 0.001642 & $\bbE[t'] \le t + 0.333333 x + 2.33333, \bbE[x'] \le 0.666667 x + 0.166667$ \\
20 & random-walk & 0.001842 & $\bbE[t'] \le 2n+t+2 \mi{x\_neg}-2 \mi{x\_pos}, \bbE[\mi{x\_neg}'] \le 0.5n + 1.5 \mi{x\_neg} - 0.5 \mi{x\_pos}$ \\
21 & random-walk-uneven & 0.001664 & $\bbE[t'] \le 2n+t+2\mi{x\_neg}-2\mi{x\_pos}+2, \bbE[\mi{x\_neg}'] \le n+2\mi{x\_neg} - \mi{x\_pos}+1$ \\
22 & unbiased & 0.000624 & $\bbE[r'] \le 1.5, \bbE[t'] \le t + 11.1111$ \\
\hline
\end{tabular}  
\vspace{2em}
\end{table}

\section{Case Study: Expectation-Recurrence Analysis}
\label{Se:Appendix:ExpectationRecurrenceAnalysis}

\cref{Ta:FullExpectationRecurrenceAnalysis} presents the evaluation results for the expectation-recurrence analysis described in \cref{Se:ExpectationRecurrenceAnalysis}.

\begin{table}
\centering
\footnotesize
\caption{Expectation-recurrence analysis. (Time is in seconds and ``T/O'' means ``time out after 3 minutes.'')}
\label{Ta:FullExpectationRecurrenceAnalysis}
\begin{tabular}{c c c l}
\hline
\multirow{2}{*}{\#} & \multirow{2}{*}{Program} & Time & \multirow{2}{*}{Invariants (selected) derived by our work} \\ 
& & (our work / \textsc{Polar}) & \\ \hline
1 & 50coinflips & 108.920 / \textbf{0.448} & $\bbE[\mi{r0}'] = 1 - (\frac{1}{2})^{n}, \bbE[\mi{total}'] = 50(1 -  (\frac{1}{2})^{n-1} )$ \\
2 & bimodal-x & \textbf{0.069} / 0.191 & $\bbE[\mi{xlow}'] = -\frac{1}{2} n, \bbE[\mi{xup}'] = \frac{1}{4}n, \bbE[x'] = -\frac{1}{4} n$ \\
3 & dbn-component-health & \textbf{0.109} / 0.181 & $\bbE[\mi{health}'] = (\frac{9}{100})^n, \bbE[\mi{obs}'] = \frac{2683}{1000} (\frac{9}{100})^n + \frac{701}{1000}$ \\
4 & dbn-umbrella & T/O / \textbf{0.456} & - \\
5 & gambling & \textbf{0.054} / 0.079 & $\bbE[\mi{bet}'] = \frac{1}{2}n + 1, \bbE[\mi{money}'] = 0$ \\
6 & geometric & \textbf{0.151} / 0.153 & $\bbE[\mi{stop}'] = \frac{1}{2}, \bbE[x'] = 2^n$ \\
7 & hawk-dove & \textbf{0.063} / 0.148 & $\bbE[\mi{p1bal}'] = n, \bbE[\mi{p2bal}'] = n$ \\
8 & hermann3 & T/O / \textbf{0.294} & - \\
9 & las-vegas-search & \textbf{0.059} / 0.288 & $\bbE[\mi{found}'] = 1 -  (\frac{20}{21})^n, \bbE[\mi{attempts}'] = \frac{20}{21} n $ \\
10 & random-walk-1d & \textbf{0.039} / 0.077 & $\bbE[x'] = 1$ \\[3pt]
\multirow{2}{*}{11} & \multirow{2}{*}{randomized-response} & \multirow{2}{*}{\textbf{0.093} / 0.168} & $\bbE[\mi{n1}'] = \frac{5}{4}n - p n, \bbE[\mi{p1}'] = \frac{1}{4}n + \frac{1}{2} p n,$ \\
& & & $\enskip \bbE[\mi{ntrue}'] = \frac{3}{4}n, \bbE[\mi{nfalse}'] = \frac{1}{4} n$ \\[3pt]
\multirow{2}{*}{12} & \multirow{2}{*}{rock-paper-scissors} & \multirow{2}{*}{\textbf{0.110} / 0.215} & $\bbE[\mi{p1bal}'] = n(\mi{p2}(-\mi{q2}-2\mi{q3}+1) + \mi{p3}(\mi{q2}-\mi{q3})+\mi{q3}),$ \\
& & & $\enskip \bbE[\mi{p2bal}'] = n( -(\mi{p2} - 1) \mi{q2} + \mi{p2} \mi{q3} - \mi{p3}(2\mi{q2}+\mi{q3}-1)) $ \\
\hline
13 & lane-keeping & \textbf{0.087} / 0.193 & $ \bbE[z'] = 0, \bbE[x'] = (\frac{4}{5})^n x, \bbE[y'] = y+\frac{1}{2}x - \frac{1}{2} (\frac{4}{5})^n x$ \\[3pt]
\multirow{2}{*}{14} & \multirow{2}{*}{running-example} & \multirow{2}{*}{0.273 / \textbf{0.182}} & $\bbE[x']=\frac{3}{2}n + 1, \bbE[y'] =  -9n - \frac{37}{5} (\frac{2}{3})^n + \frac{67}{5} (\frac{3}{2})^n  - 5, $ \\
& & & $\enskip \bbE[z'] = \frac{9}{4} n^2 + \frac{19}{4} n - \frac{201}{10}  (\frac{3}{2})^n - \frac{37}{5} (\frac{2}{3})^n + \frac{57}{2}$ \\
\hline
\end{tabular}
\end{table}
\fi

\end{document}